\documentclass[preprint,12pt,3p]{elsarticle}

\journal{Artificial Intelligence}

\usepackage{amsmath}
\usepackage{amsthm, amssymb}
\usepackage{amsfonts}
\usepackage{hyperref}
\hypersetup{
    colorlinks   = true,
    linkcolor    = red, %
    urlcolor     = teal, %
    citecolor    = blue %
}
\usepackage{cleveref}
\usepackage{todonotes}
\usepackage{enumerate}
\usepackage{url}
\usepackage{xspace}
\usepackage{tikz}
\usetikzlibrary{calc,fit,backgrounds,decorations.pathmorphing,decorations.pathreplacing,positioning,calc,scopes,backgrounds,decorations.pathmorphing,fit,math}
\usepackage{subcaption}
\usepackage[inline, shortlabels]{enumitem}
\usepackage{natbib}

\usepackage[mathlines]{lineno}

\newcommand{\N}{\mathbb{N}}
\newcommand{\Z}{\mathbb{Z}}

\newcommand{\TSS}{\textsc{Target Set Selection}\xspace}
\newcommand{\TSSshort}{\textsc{TSS}\xspace}
\newcommand{\twoTSS}{\textsc{$2$-Opinion Target Set Selection}\xspace}
\newcommand{\twoTSSshort}{\textsc{$2$OTSS}\xspace}
\newcommand{\PSI}{\textsc{Partitioned Subgraph Isomorphism}\xspace}
\newcommand{\PSIshort}{\textsc{PSI}\xspace}
\newcommand{\VC}{\textsc{Vertex Cover}\xspace}

\newcommand{\DP}{\operatorname{DP}}
\newcommand{\fin}{\ensuremath{\mathcal{T}}}
\newcommand{\cost}{\operatorname{cost}}
\newcommand{\bm}[1]{\mathbf{#1}}
\newcommand{\low}{\operatorname{low}}
\newcommand{\high}{\operatorname{high}}
\newcommand{\expr}{\operatorname{expr}}

\newcommand{\Oh}[1]{\ensuremath{{\mathcal{O}(#1)}}}

\newcommand{\NPh}{\textsf{NP}-hard\xspace}
\newcommand{\FPT}{\textsf{FPT}\xspace}
\newcommand{\XP}{\textsf{XP}\xspace}
\newcommand{\W}[1][1]{\textsf{W[#1]}\xspace}
\newcommand{\Wh}[1][1]{\textsf{W[#1]}-h\xspace}
\newcommand{\Whard}[1][1]{\textsf{W[#1]}-hard\xspace}
\newcommand{\Whness}[1][1]{\textsf{W[#1]}-hardness\xspace}
\newcommand{\Wcomplete}[1][1]{\textsf{W[#1]}-complete\xspace}

\newcommand{\tw}{\operatorname{tw}}
\newcommand{\w}{\omega}
\newcommand{\pw}{\operatorname{pw}}
\newcommand{\td}{\operatorname{td}}
\renewcommand{\d}{\delta}

\newtheorem{theorem}{Theorem}
\newtheorem{lemma}{Lemma}
\newtheorem{observation}{Observation}
\newtheorem{rrule}{Reduction Rule}

\theoremstyle{definition}
\newtheorem{definition}{Definition}

\theoremstyle{remark}
\newtheorem{claim}{Claim}
\newtheorem{corollary}{Corollary}
\newtheorem*{example*}{Example}

\newenvironment{claimproof}[1][\proofname]{%
  \proof[#1]
}{\endproof}

\crefname{observation}{Observation}{Observations}
\crefname{rrule}{Reduction Rule}{Reduction Rules}

\newcommand{\yesI}{\texttt{yes}-instance\xspace}

\newcommand{\convexpath}[2]{
	[
	create hullnodes/.code={
		\global\edef\namelist{#1}
		\foreach [count=\counter] \nodename in \namelist {
			\global\edef\numberofnodes{\counter}
			\node at (\nodename) [draw=none,name=hullnode\counter] {};
		}
		\node at (hullnode\numberofnodes) [name=hullnode0,draw=none] {};
		\pgfmathtruncatemacro\lastnumber{\numberofnodes+1}
		\node at (hullnode1) [name=hullnode\lastnumber,draw=none] {};
	},
	create hullnodes
	]
	($(hullnode1)!#2!-90:(hullnode0)$)
	\foreach [
	evaluate=\currentnode as \previousnode using \currentnode-1,
	evaluate=\currentnode as \nextnode using \currentnode+1
	] \currentnode in {1,...,\numberofnodes} {
		let
		\p1 = ($(hullnode\currentnode)!#2!-90:(hullnode\previousnode)$),
		\p2 = ($(hullnode\currentnode)!#2!90:(hullnode\nextnode)$),
		\p3 = ($(\p1) - (hullnode\currentnode)$),
		\n1 = {atan2(\y3,\x3)},
		\p4 = ($(\p2) - (hullnode\currentnode)$),
		\n2 = {atan2(\y4,\x4)},
		\n{delta} = {-Mod(\n1-\n2,360)}
		in
		{-- (\p1) arc[start angle=\n1, delta angle=\n{delta}, radius=#2] -- (\p2)}
	}
	-- cycle
}

\begin{document}

\begin{frontmatter}

\title{Balancing the Spread of Two Opinions in Sparse Social Networks\tnoteref{aaaiSA}}
\tnotetext[aaaiSA]{A preliminary summary of some of the results appeared as a student abstract in the Proceedings of the 36th AAAI Conference on Artificial Intelligence (AAAI~'22)~\cite{KnopSS2022}.}

\author[CTU]{Dušan Knop} %
\author[CTU,AGH]{Šimon Schierreich}
\author[CTU]{Ondřej Suchý}

\affiliation[CTU]{
    organization={Czech Technical University in Prague},%
    city={Prague},
    country={Czech Republic}
}

\affiliation[AGH]{
    organization = {AGH University of Krakow},
    city         = {Krakow},
    country      = {Poland}
}

\begin{abstract}
    Inspired by the famous \TSS problem, we propose a new discrete model to simultaneously spread two opinions within a social network and perform an initial study of its complexity. Here, we are given a social network, a seed-set of agents for each opinion, two thresholds for each agent, a budget, and a number of rounds. The first threshold represents the willingness of an agent to adopt an opinion if the agent has no opinion at all, while the second threshold states the willingness to acquire a second opinion if the agent already has one. The goal is to add at most budget-many agents to the initial seed-sets such that the process started with these extended seed-sets stabilizes within the given number of rounds, with each agent having either both opinions or none. That is, our goal is to ensure that the spread of opinions is \emph{balanced}.
	
	We show that the problem is \NPh, and thus we study the problem from the perspective of parameterized complexity. In particular, we show that the problem is \FPT when parameterized by the number of rounds, the maximum threshold, and the treewidth combined. This algorithm also applies to the combined parameter, the treedepth and the maximum threshold. Finally, we show that the problem is \FPT when parameterized by the vertex cover number, the $3$-path vertex cover number, or the vertex integrity of the input network alone. To complement our tractability results, we show that the problem is \Whard with respect to a) the sizes of the initial seed-sets and the feedback-vertex set number combined, even if all thresholds are bounded by a constant, and b) the budget, the $4$-path vertex cover number, and the feedback-vertex set number combined, even if every activation process stabilizes in at most $4$ rounds.
\end{abstract}


\end{frontmatter}

\section{Introduction}\label{sec:introduction}

In network dynamics, there are two main mathematical models to study influence in social networks~\cite{ShakarianBAER2015}. In both models, there is a certain set of ``active'' agents, called \emph{seed-set}, initially having some opinion or secret. Next, in discrete rounds, some new agents become activated because of influence from previously active agents in their neighborhood. The models then differ in the way in which influence is propagated through the network.

The first model is known under the name \emph{Independent cascades}~\cite{GoldenbergLM2001,EasleyK10}. Here, every agent has some predefined probability of, once activated, influencing each of her neighbors. 
In the \emph{Linear threshold} model~\cite{KempeKT03,KempeKT05,KempeKT15}, on the contrary, every agent selects her own threshold at random, and once the sum of influence of active agents in the neighborhood of the agent exceeds this threshold, the agent becomes activated as well. 

In this work, we build on the deterministic version of the second of the aforementioned models of influence spread, which is also known as the computational problem called \TSS (\TSSshort). \TSSshort was introduced by \citet{DomingosR01} in order to study the influence of direct marketing on social networks; see, e.g., the monographs of \citet{ChenCL13,BorgsBCL14,KempeKT15} and references therein. Using the threshold model of \citet{KempeKT15}, we formally define \TSSshort as follows.
We are given a social network modeled as an undirected graph~$G = (V,E)$, a \emph{threshold value} $f(v) \in \N$ for every agent~$v \in V$, and a budget~$B\in \N$.
The task is to select a set of agents~$T \subseteq V$, where $|T|\leq B$, so that when the following \emph{activation process}
\[
P^0 = T
\qquad \text{and} \qquad
P^{i+1} = P^i \cup \Big\{ v \in V \setminus P^i \,\Big|\, f(v) \le \left| P^i \cap N(v) \right| \Big\}
\]
stabilizes
with the set~$P^{\fin}$, that is, $P^{\fin} = P^{\fin+1}$, the size of~$P^{\fin}$ is maximized; an important setting being we require $P^{\fin} = V$.

Both models have been generalized in many follow-up works---the most important for us are those modeling the spread of more than one opinion (e.g., different products).
In this line of research, most of the works focused on antagonistic opinions, that is, the agents can only acquire a single opinion; see, e.g., \cite{CarnesNWZ2007,BharathiKS07,BudakAA2011}.
To the best of our knowledge, \citet{MyersL12} were the first who considered that the agent can acquire several opinions (though in a rather different setting), while \citet{LuCL2015} were the first to consider the computational complexity of problems comprising such opinions.
Later, \citet{GarimellaGPT17} discussed the importance of balancing the outreach of the two opinions being spread in the social network (for two opinions). %

These works are mainly motivated by the diffusion in a social network where it is desirable for the considered opinions to spread roughly equally and thus prevent unwanted behavior of the social networks, such as, e.g., \emph{echo bubbles} or \emph{strengthening of extreme agents}; see, e.g., \cite[Figure~1]{ConoverRFGMF11}.
In other words, the most desired outcome is that, after a similar activation process as above stabilizes, every agent, in the case of two opinions, either has no opinion or has both.

\begin{example*}
	Suppose that there are two (rather antagonistic) opinions in a social network---if one knows which of these is the truthful one, they might, e.g., block the other.
	What can we do if we cannot tell which one is truthful?
	If many agents receive only the opinion that later turns out to be deceptive, then these agents might feel deceived by the network and leave it.
	Instead, we decide to help both opinions spread evenly.
	The agents receiving two different opinions can take adequate measures to react to the situation.
	Real world examples of this kind of information spread include, e.g., two experts having different opinions on covid-19 vaccination and both claiming having data supporting their opinion---how can an AI or even a human agent tell which of the two is trustworthy? %
\end{example*}

We formally define the problem of our interest as follows.

\paragraph*{The Model}
In the \twoTSS problem (\twoTSSshort for short), we are given a simple undirected graph~$G = (V,E)$, (possibly empty) seed sets~${S_a, S_b \subseteq V}$, threshold functions ${f_1, f_2 \colon V \to \N}$, a positive integer $\fin \in \N$, and a positive integer~$B \in \N$.
The task is to select two additional seed sets~$T_a, T_b \subseteq V$ with $|T_a| + |T_b| \le B$ such that the following \emph{activation process} terminates with $P_a^{\fin}= P_a^{\fin+1} = P_b^{\fin}= P_b^{\fin+1}$.
The initialization is by setting
\begin{equation*}
	P_a^0 = S_a \cup T_a
	\qquad\mbox{and}\qquad
	P_b^0 = S_b \cup T_b
\end{equation*}
and the process continues with
\begin{align*}
	P_c^{i+1} &=
	\Big\{ v \in V \setminus \left( P_a^{i} \cup P_b^{i} \right) \,\Big|\, f_1(v) \le \left| P_c^{i} \cap N(v) \right|\Big\} %
	\cup
	\Big\{ v \in P_{\neg c}^{i} \,\Big|\, f_2(v) \le \left| P_c^{i} \cap N(v) \right|\Big\}
	\cup %
	P_c^{i}
\end{align*}
for $i \geq 0$ and $c \in \{a,b\}$, where $\neg c$ is the element in the set $\{ a,b \} \setminus \{c\}$.
That is, an agent~$v$ gains opinion $c$ if she has no opinion and the number of her neighbors already having opinion~$c$ is at least $f_1(v)$ or if she already has the other opinion and the number of her neighbors already having opinion $c$ is at least $f_2(v)$.
It is not hard to see that the above process terminates in at most~$2n$ steps, where $n = |V|$, since in each non-terminal step we add at least one vertex to~$P_a^i \cup P_b^i$ (i.e., $\left| P_a^i \setminus P_a^{i-1} \right| + \left| P_b^i \setminus P_b^{i-1} \right| \ge 1$).
Thus, letting $\fin = 2n$ corresponds to putting no restriction on the length of the process. For a running example of the activation process, we refer the reader to \Cref{fig:running_example}.

Note that an agent with $f_1(v)=0$ is for sure included in $P_c^1$, unless it is in $P_{\neg c}^{0} \setminus P_c^0$.
Similarly, an agent with $f_2(v)=0$ gains the opinion $c$ at most one round after gaining opinion $\neg c$.
Conversely, an agent with $f_1(v) \ge \deg(v)+1$ or $f_2(v) \ge \deg(v)+1$ can never pass this threshold and, hence, can be only included in $P_c^i$ through the seed sets or possibly passing the $f_1$ threshold in the latter case.

There are several aspects of the problem that might seem unrealistic (such as requiring a perfectly balanced spread).
However, for the ease of exposition, we decided to keep the model as simple as possible.
We discuss possible variations, including imbalanced gain of opinions, agent-preferred opinion, or probabilistic thresholds, in \Cref{sec:variants}.%

\paragraph*{Versatility of the Model}
Note that we model the network as an undirected graph.
Furthermore, %
we do not have opinion-specific threshold values, but rather distinguish the order in which the agent is exposed to the opinions in question.
We believe that this is an interesting setting as it can capture different agents' mindsets; see also \cite{MyersL12,LuCL2015}:

\begin{description}
	\item[Competing opinions ($f_1(v) < f_2(v)$)]
	First, an agent~$v$ might be tougher towards the second opinion (by, e.g., setting $f_2(v)$ to some multiple of $f_1(v)$) for which we can see an application in modeling of, e.g., sharing tweets of political leaders on social networks. 
	An agent convinced of ``her truth'' might be hard to conceive other viewpoints.
	\item[Complementary opinions ($f_1(v) > f_2(v)$)]
	Second, an agent~$v$ might infer the second opinion easily (by, e.g., setting $f_2(v)$ to some fraction of $f_1(v)$) for which we can see an application in modeling of, e.g., spread of different viruses in a population, since the first virus in a human body decreases its ability to resist other viruses.
	An agent might also actively seek further viewpoints on the topic.
	\item[Independent opinions ($f_1(v) = f_2(v)$)] Third, for an agent, the decision to accept the opinion might be independent of whether she already acquired the other opinion.
\end{description}

Let us stress that (even in the first mindset) if both opinions are strong enough at the same time, then the agent receives both opinions (and evaluates the number of neighboring agents having a certain opinion only against the threshold value~$f_1(v)$).
This aspect of the model might be considered slightly counterintuitive. As already mentioned, we address possible variations of the model (not only in this regard) later in \Cref{sec:variants}.

\begin{figure}[tb!]
	\centering
	\definecolor{blue}{HTML}{C0DDFF}
	\definecolor{red}{HTML}{EE4440}
	\begin{subfigure}[t]{0.32\textwidth}
		\centering
		\begin{tikzpicture}[thick,scale=0.7, every node/.style={scale=0.7}]
    \node[draw,circle,label={[align=left]200:{\tiny$f_1=2$\\\tiny$f_2=1$}}] (v0) at (0,0) {$v_0$};
    \node[draw,circle,label={above:\tiny$f_1=f_2=5$}] (v1) at (-2,0) {$v_1$};
    \node[draw,circle,label={left:\tiny$f_2=2$},fill=blue] (v2) at (0,2) {$v_2$};
    \node[draw,circle,label={[align=left]below:\tiny$f_1=1$\\\tiny$f_2=6$}] (v3) at (2,0) {$v_3$};
    \node[draw,circle,label={left:\tiny$f_1=f_2=2$}] (v4) at (0,-2) {$v_4$};
    \node[draw,circle,fill=blue,label={right:\tiny$f_2=1$}] (v5) at (2,3) {$v_5$};
    \node[draw,circle,left color=blue,right color=red] (v6) at (4,0) {$v_6$};
    \node[draw,circle,fill=red,label={below:\tiny$f_2=1$}] (v7) at (-1.5,-3) {$v_7$};
    \node[draw,circle,label={below:\tiny$f_1=2,f_2=1$}] (v8) at (1.5,-3) {$v_8$};
    \node[draw,circle,fill=red,label={below:\tiny$f_2=1$}] (v9) at (3,-2) {$v_9$};
    
    \draw (v0) edge (v1) edge (v2) edge (v3) edge (v4);
    \draw (v2) edge (v5) edge (v3);
    \draw (v4) edge (v7) edge (v8);
    \draw (v8) edge (v9);
    \draw (v3) edge (v6);
\end{tikzpicture}
		\caption{The initial state of the network. Agents $v_2,v_5$ are part of the initial seed set $S_a$, agents $v_7,v_9$ are part of the initial seed set $S_b$, and the agent $v_6$ is part of both seed sets $S_a$ and $S_b$. The remaining agents have neither opinion.}
		\label{fig:runnin_example_initial_setting}
	\end{subfigure}
	\hfill
	\begin{subfigure}[t]{0.32\textwidth}
		\centering
		\begin{tikzpicture}[thick,scale=0.7, every node/.style={scale=0.7}]
    \node[draw,circle,label={[align=left]200:{\tiny$f_1=2$\\\tiny$f_2=1$}}] (v0) at (0,0) {$v_0$};
    \node[draw,circle,label={above:\tiny$f_1=f_2=5$}] (v1) at (-2,0) {$v_1$};
    \node[draw,circle,label={left:\tiny$f_2=2$},fill=blue] (v2) at (0,2) {$v_2$};
    \node[draw,circle,label={[align=left]below:\tiny$f_1=1$\\\tiny$f_2=6$},left color=blue,right color=red] (v3) at (2,0) {$v_3$};
    \node[draw,circle,label={left:\tiny$f_1=f_2=2$}] (v4) at (0,-2) {$v_4$};
    \node[draw,circle,fill=blue,label={right:\tiny$f_2=1$}] (v5) at (2,3) {$v_5$};
    \node[draw,circle,left color=blue,right color=red] (v6) at (4,0) {$v_6$};
    \node[draw,circle,fill=red,label={below:\tiny$f_2=1$}] (v7) at (-1.5,-3) {$v_7$};
    \node[draw,circle,label={below:\tiny$f_1=2,f_2=1$}] (v8) at (1.5,-3) {$v_8$};
    \node[draw,circle,fill=red,label={below:\tiny$f_2=1$}] (v9) at (3,-2) {$v_9$};
    
    \draw (v0) edge (v1) edge (v2) edge (v3) edge (v4);
    \draw (v2) edge (v5) edge (v3);
    \draw (v4) edge (v7) edge (v8);
    \draw (v8) edge (v9);
    \draw (v3) edge (v6);
\end{tikzpicture}
		\caption{In the first round, agent $v_3$ acquires both opinions $a$ and $b$ at the same time, since the value of threshold function $f_1$ is $1$, and she has at least one neighbor with opinion $a$ and at least one neighbor with opinion $b$. Note that if the agent would gain only one opinion in this round, she would not gain the second opinion at all.
		}
		\label{fig:runnin_example_round1}
	\end{subfigure}
	\hfill
	\begin{subfigure}[t]{0.32\textwidth}
		\centering
		\begin{tikzpicture}[thick,scale=0.7, every node/.style={scale=0.7}]
    \node[draw,circle,label={[align=left]200:{\tiny$f_1=2$\\\tiny$f_2=1$}},fill=blue] (v0) at (0,0) {$v_0$};
    \node[draw,circle,label={above:\tiny$f_1=f_2=5$}] (v1) at (-2,0) {$v_1$};
    \node[draw,circle,label={left:\tiny$f_2=2$},fill=blue] (v2) at (0,2) {$v_2$};
    \node[draw,circle,label={[align=left]below:\tiny$f_1=1$\\\tiny$f_2=6$},left color=blue,right color=red] (v3) at (2,0) {$v_3$};
    \node[draw,circle,label={left:\tiny$f_1=f_2=2$}] (v4) at (0,-2) {$v_4$};
    \node[draw,circle,fill=blue,label={right:\tiny$f_2=1$}] (v5) at (2,3) {$v_5$};
    \node[draw,circle,left color=blue,right color=red] (v6) at (4,0) {$v_6$};
    \node[draw,circle,fill=red,label={below:\tiny$f_2=1$}] (v7) at (-1.5,-3) {$v_7$};
    \node[draw,circle,label={below:\tiny$f_1=2,f_2=1$}] (v8) at (1.5,-3) {$v_8$};
    \node[draw,circle,fill=red,label={below:\tiny$f_2=1$}] (v9) at (3,-2) {$v_9$};
    
    \draw (v0) edge (v1) edge (v2) edge (v3) edge (v4);
    \draw (v2) edge (v5) edge (v3);
    \draw (v4) edge (v7) edge (v8);
    \draw (v8) edge (v9);
    \draw (v3) edge (v6);
\end{tikzpicture}
		\caption{In the second round, agent $v_0$ gains opinion $a$ since he finally has at least $f_1(v_0)$ neighbors with opinion $a$. She cannot obtain opinion $b$ because there is not enough neighbors with this opinion.}
		\label{fig:runnin_example_round2}
	\end{subfigure}
	
	\vspace{0.45cm}
	
	\begin{subfigure}[t]{0.32\textwidth}
		\centering
		\begin{tikzpicture}[thick,scale=0.7, every node/.style={scale=0.7}]
    \node[draw,circle,label={[align=left]200:{\tiny$f_1=2$\\\tiny$f_2=1$}},left color=blue,right color=red] (v0) at (0,0) {$v_0$};
    \node[draw,circle,label={above:\tiny$f_1=f_2=5$}] (v1) at (-2,0) {$v_1$};
    \node[draw,circle,label={left:\tiny$f_2=2$},fill=blue] (v2) at (0,2) {$v_2$};
    \node[draw,circle,label={[align=left]below:\tiny$f_1=1$\\\tiny$f_2=6$},left color=blue,right color=red] (v3) at (2,0) {$v_3$};
    \node[draw,circle,label={left:\tiny$f_1=f_2=2$}] (v4) at (0,-2) {$v_4$};
    \node[draw,circle,fill=blue,label={right:\tiny$f_2=1$}] (v5) at (2,3) {$v_5$};
    \node[draw,circle,left color=blue,right color=red] (v6) at (4,0) {$v_6$};
    \node[draw,circle,fill=red,label={below:\tiny$f_2=1$}] (v7) at (-1.5,-3) {$v_7$};
    \node[draw,circle,label={below:\tiny$f_1=2,f_2=1$}] (v8) at (1.5,-3) {$v_8$};
    \node[draw,circle,fill=red,label={below:\tiny$f_2=1$}] (v9) at (3,-2) {$v_9$};
    
    \draw (v0) edge (v1) edge (v2) edge (v3) edge (v4);
    \draw (v2) edge (v5) edge (v3);
    \draw (v4) edge (v7) edge (v8);
    \draw (v8) edge (v9);
    \draw (v3) edge (v6);
\end{tikzpicture}
		\caption{In the third round, the agent $v_0$ obtains opinion $b$. This would not be possible without acquiring opinion $a$ in the previous round.}
		\label{fig:runnin_example_round3}
	\end{subfigure}
	\hfill
	\begin{subfigure}[t]{0.32\textwidth}
		\centering
		\begin{tikzpicture}[thick,scale=0.7, every node/.style={scale=0.7}]
        \node[draw,circle,label={[align=left]200:{\tiny$f_1=2$\\\tiny$f_2=1$}},left color=blue,right color=red] (v0) at (0,0) {$v_0$};
        \node[draw,circle,label={above:\tiny$f_1=f_2=5$}] (v1) at (-2,0) {$v_1$};
        \node[draw,circle,label={left:\tiny$f_2=2$},left color=blue,right color=red] (v2) at (0,2) {$v_2$};
        \node[draw,circle,label={[align=left]below:\tiny$f_1=1$\\\tiny$f_2=6$},left color=blue,right color=red] (v3) at (2,0) {$v_3$};
        \node[draw,circle,label={left:\tiny$f_1=f_2=2$},fill=red] (v4) at (0,-2) {$v_4$};
        \node[draw,circle,fill=blue,label={right:\tiny$f_2=1$}] (v5) at (2,3) {$v_5$};
        \node[draw,circle,left color=blue,right color=red] (v6) at (4,0) {$v_6$};
        \node[draw,circle,fill=red,label={below:\tiny$f_2=1$}] (v7) at (-1.5,-3) {$v_7$};
        \node[draw,circle,label={below:\tiny$f_1=2,f_2=1$}] (v8) at (1.5,-3) {$v_8$};
        \node[draw,circle,fill=red,label={below:\tiny$f_2=1$}] (v9) at (3,-2) {$v_9$};
        
        \draw (v0) edge (v1) edge (v2) edge (v3) edge (v4);
        \draw (v2) edge (v5) edge (v3);
        \draw (v4) edge (v7) edge (v8);
        \draw (v8) edge (v9);
        \draw (v3) edge (v6);
    \end{tikzpicture}
		\caption{In the fourth round, there are finally enough neighbors of agent $v_4$ and he gets opinion~$b$. The same holds for agent $v_2$ which has both opinions now.}
		\label{fig:runnin_example_round4}
	\end{subfigure}
	\hfill
	\begin{subfigure}[t]{0.32\textwidth}
		\centering
		\begin{tikzpicture}[thick,scale=0.7, every node/.style={scale=0.7}]
    \node[draw,circle,label={[align=left]200:{\tiny$f_1=2$\\\tiny$f_2=1$}},left color=blue,right color=red] (v0) at (0,0) {$v_0$};
    \node[draw,circle,label={above:\tiny$f_1=f_2=5$}] (v1) at (-2,0) {$v_1$};
    \node[draw,circle,label={left:\tiny$f_2=2$},left color=blue,right color=red] (v2) at (0,2) {$v_2$};
    \node[draw,circle,label={[align=left]below:\tiny$f_1=1$\\\tiny$f_2=6$},left color=blue,right color=red] (v3) at (2,0) {$v_3$};
    \node[draw,circle,label={left:\tiny$f_1=f_2=2$},fill=red] (v4) at (0,-2) {$v_4$};
    \node[draw,circle,label={right:\tiny$f_2=1$},left color=blue,right color=red] (v5) at (2,3) {$v_5$};
    \node[draw,circle,left color=blue,right color=red] (v6) at (4,0) {$v_6$};
    \node[draw,circle,fill=red,label={below:\tiny$f_2=1$}] (v7) at (-1.5,-3) {$v_7$};
    \node[draw,circle,label={below:\tiny$f_1=2,f_2=1$},fill=red] (v8) at (1.5,-3) {$v_8$};
    \node[draw,circle,fill=red,label={below:\tiny$f_2=1$}] (v9) at (3,-2) {$v_9$};
    
    \draw (v0) edge (v1) edge (v2) edge (v3) edge (v4);
    \draw (v2) edge (v5) edge (v3);
    \draw (v4) edge (v7) edge (v8);
    \draw (v8) edge (v9);
    \draw (v3) edge (v6);
\end{tikzpicture}
		\caption{The activation process stabilizes after $5$ rounds. Agent $v_1$ has no opinion due to the value of her threshold functions. The opinion $a$ is blocked by agent~$v_4$, so we need to add an agent into the additional seed set $T_a$ to achieve a balanced spread.}
		\label{fig:runnin_example_round5}
	\end{subfigure}
	
	\caption{A running example of the activation process described by the model.}
	\label{fig:running_example}
\end{figure}

\subsection{Our Contribution}

In this work, we propose a new discrete model for spreading information with at least two opinions---via the related computation problem \twoTSS---and perform an initial study of its (parameterized) complexity.

\begin{figure}[bt!]
	\begin{minipage}{0.49\linewidth}
		\centering
		\begin{tikzpicture}
			\node (vc) at (-0.75,-0.25) {$\operatorname{vc}$};
			\node (3pvc) at (-0.75,-1) {$\operatorname{3-pvc}$};
			\node (4pvc) at (-0.25,-2) {$\operatorname{4-pvc}$};
			\node (vi) at (0.75,-1.75) {$\operatorname{vi}$};
			\node (fvs) at (-1.5,-2.5) {$\operatorname{fvs}$};
			\node (td) at (0,-2.75) {$\operatorname{td}^\dagger$};
			\node (pw) at (0,-3.5) {$\operatorname{pw}$};
			\node (tw) at (-0.75,-4.25) {$\operatorname{tw}^\ddagger$};
			
			\node at (1.8, -2) {\FPT};
			\draw [shorten <= 0.25cm, shorten >= 0.25cm] (2,-2.95) to[out=140,in=-30] (-3,-0.6);
			\node at (-2.75, -2.5) {\Wh};
			
			\draw (vc) -- (3pvc);
			\draw (3pvc) -- (fvs);
			\draw (3pvc) -- (4pvc);
			\draw (3pvc) -- (vi);
			\draw (fvs) -- (tw);
			\draw (4pvc) -- (td);
			\draw (vi) -- (td);
			\draw (td) -- (pw);
			\draw (pw) -- (tw);
			
			\begin{pgfonlayer}{background}
				\fill[green!20] \convexpath{3pvc,vc,vi.north}{15pt};
				\fill[red!40] \convexpath{fvs,4pvc.south,td,pw,tw}{15pt};
			\end{pgfonlayer}
		\end{tikzpicture}
	\end{minipage}
	\begin{minipage}{0.49\textwidth}
		\centering
		\begin{tikzpicture}
			\tikzstyle{policko}=[rectangle, minimum height=.55cm, minimum width=1.7cm]
			\tikzstyle{polickoG}=[policko, fill=green!20]
			\tikzstyle{polickoR}=[policko, fill=red!40]
			\tikzstyle{trojpolicko}=[rectangle, minimum height=.5cm, minimum width=5.1cm]
			\tikzstyle{sipka}=[ultra thick, >=stealth, ->]
			\tikzstyle{tabulka}=[thick]
			
			\node[polickoG] (CC) {};
			\node[polickoG, right=0cm of CC] (CP) {};
			\node[polickoG, below=0cm of CC] (PC) {};
			\node[polickoG, right=0cm of PC] (PP) {\, Thm~\ref{thm:twoTSSIsFPTWrtRoundsFmaxTw}};
			
			\node[polickoR, right=0cm of CP] (CI) {Thm~\ref{thm:twoTSSIsHardForTDandConstRounds}};
			\node[polickoR, below=0cm of CI] (PI) {};
			\node[polickoR, below=0cm of PI] (II) {};
			\node[polickoR, left=0cm of II] (IP) {};
			\node[polickoR, left=0cm of IP] (IC) {Thm~\ref{thm:twoTSSIsHardForTWandConstFMax}};
			
			\node[policko, left=0cm of CC] (CClab) {\texttt{const}};
			\node[policko, left=0cm of PC] (PClab) {\texttt{param}};
			\node[policko, left=0cm of IC] (IClab) {\texttt{input}};
			\node[policko, above=0cm of CClab] (rounds) {};
			
			\node[policko, above=0.07cm of CC] (CClab2) {\texttt{const}};
			\node[policko, above=0cm of CP] (CPlab) {\texttt{param}};
			\node[policko, above=0.03cm of CI] (CIlab) {\texttt{input}};
			
			\node[trojpolicko, above=0cm of CPlab] (threshold) {thresholds $f_{\max}$};
			\node[policko, rotate=90, left=.2cm of CClab.north west] (treewidth) {rounds};
			
			\draw[tabulka] (II.south east) to ($(IClab.south west) - (.4,0)$) to ($ (threshold.north west) - (1.7,0) - (.4,0)$) to (threshold.north east) -- cycle;
			\draw[tabulka] (threshold.north west) to (IC.south west);
			\draw[tabulka] ($(rounds.south west) - (.4,0)$) to (CI.north east);
			
			\draw (CPlab.north west) to (IP.south west);
			\draw (CPlab.north east) to (IP.south east);
			
			\draw (PClab.north west) to (PI.north east);
			\draw (PClab.south west) to (PI.south east);
			
			\draw[sipka] ($(PP) - (.6,0)$) -- (PC.center);
			\draw[sipka] ($(PP) - (.6,0.03)$) -- ($(CP) - (.6,0)$);
			\draw[sipka] ($(PP) - (.6,0)$) -- (CC.center);
			
			\draw[sipka] ($(CI) - (0,.2)$) -- (II.center);
			\draw[sipka] ($(IC) + (.7,0)$) -- (II.center);
		\end{tikzpicture}
	\end{minipage}
	\caption{An overview of our results. Both figures use green (light-gray) to highlight \FPT results while the red (dark gray) is used for \Whard[1] results. 
    \\
    The left figure gives a complexity classification of \twoTSSshort w.r.t. different structural parameters ($\operatorname{vc}$ is the vertex cover number, $\operatorname{k-pvc}$ is the $k$-path vertex cover number, $\operatorname{vi}$ is the vertex integrity, $\operatorname{fvs}$ is the feedback vertex number, $\operatorname{td}$ is the treedepth, $\operatorname{pw}$ is the pathwidth, and $\operatorname{tw}$ is the treewidth). 
	The problem becomes \FPT if we combine the parameter marked $\dagger$ with the maximum threshold, and the same holds if we combine the structural parameter marked $\ddagger$ with the number of rounds and the maximum threshold. 
	\\
    The table on the right then presents a detailed summary of our results for the \twoTSSshort{} problem parameterized by treewidth. We investigate many settings where certain parts of the input may appear either as constants (\texttt{const}), parameters (\texttt{param}), or unrestricted (\texttt{input}).
	}
	\label{fig:results_overview}
\end{figure}

The two most natural parameters for this problem are the size of the seed-sets (i.e., $|S_a|+|S_b|$) and the solution size (i.e., the budget~$B$).
The \twoTSSshort problem is \Whard with respect to both of these parameters (\Cref{thm:twoTSSIsHardForTWandConstFMax}) and thus, following the approach taken for the single opinion, we focus on studying its complexity for the structural parameters of the input graph.
Limiting the structure of the input is further motivated by the study of \citet{ManiuSJ19}, who showed that many real-world networks exhibit a bounded treewidth.
However, we show (\Cref{thm:twoTSSIsHardForTWandConstFMax}) that the \twoTSSshort problem is \Whard for the combination of parameters $|S_a|+|S_b|$ and the treewidth of the input graph, even if we assume that $f_1(v) \le f_2(v) \le 3$ for all $v \in V(G)$.
It is worth noting that the activation process is quite long in this case; thus, one may ask what happens if we limit its length~$\fin$.
The \twoTSSshort problem remains \Whard for the combination of parameters $|S_a|+|S_b|$ and the treewidth of the input graph, even if we assume that any activation process stabilizes in 4 rounds (\Cref{thm:twoTSSIsHardForTDandConstRounds}) and $f_1(v) \le f_2(v)$ for all $v \in V(G)$.
In fact, the above-mentioned hardness results are even stronger and show that the problem is \Whard with respect to the more restrictive structural parameters such as the pathwidth, the feedback vertex number, and the $4$-path vertex cover number.

The algorithmic counterpart to the two hardness results is the following.
We show that \twoTSSshort is fixed-parameter tractable when parameterized by the number of rounds of the activation process, the maximum threshold, and the treewidth of the input graph combined (\Cref{thm:twoTSSIsFPTWrtRoundsFmaxTw}).
The same algorithm applies for the combined parameters, the treedepth of the input graph, and the maximum threshold; in fact, we show that if the treedepth is bounded, then so is the length of any activation process (which might be of independent interest).
Note that we do not assume anything about the relation of $f_1(v)$ and~$f_2(v)$.
Furthermore, \twoTSSshort is in \FPT when parameterized by the vertex cover number (\Cref{thm:twoTSSIsFPTWrtVC}), the $3$-path vertex cover number (\Cref{thm:3-pvc}), and the vertex integrity (\Cref{thm:vertex_integrity}). We note that the latter two algorithms also apply to the original \TSSshort.
We summarize our complexity results in \Cref{fig:results_overview}.

\paragraph{Theoretical Significance of the Results}
In the past few decades, many algorithmic and hardness results have been obtained for various structural graph parameters (in parameterized complexity).
Quite interestingly, only a few problems are known to be fixed-parameter tractable when parameterized by the treedepth, while being \Whard when the parameter is the pathwidth of the input graph.
These, to the best of our knowledge, are the following: the \textsc{Mixed Chinese Postman} problem~\cite{GutinJW16}, the \textsc{Length Bounded Cut} problem~\cite{DvorakK18}, \textsc{Geodesic Set}~\cite{KellerhalsK20}, and \textsc{Integer Programming}~\cite{EisenbrandHKKLO19}, where the parameter is combined---the largest coefficient of the constraint matrix~$\mathbb{A}$ and the treedepth of either the primal or the dual Gaifman graph of~$\mathbb{A}$.

\subsection{Related Work}

\paragraph{Target Set Selection}
As already mentioned, our work builds on foundations from the study of \TSS. %
This problem itself was introduced in the context of direct marketing on social networks~\cite{DomingosR01}; followed by more results and applications (see, e.g., the works of \citet{RichardsonD02,KempeKT03,EasleyK10} for applications).
The initial research direction focused on studying the stochastic setting where threshold values are selected at random and we try to maximize the expected spread of information~\cite{KempeKT03,KempeKT05,MosselR07,KempeKT15}.
It is not surprising that \TSSshort is \NPh (the ``static'' variant where threshold values are part of the input); in addition, both of its natural optimization variants (either maximize the size of $P^{\fin}$ or minimize the size of~$T$ so that $P^{\fin}$ is the whole vertex set) are hard to approximate. In fact, for the minimization variant \citet{Chen09} gave a polylogarithmic approximation lower bound which holds even if the input graph is bipartite and has a bounded degree and all threshold values are either~$1$ or~$2$.

It is not hard to see that \TSSshort itself captures, e.g., \VC. It follows from \citet{Mathieson10} and \citet{LuoMS18} that from the parameterized perspective \TSSshort is \Wcomplete[P] (which implies it is \Whard) for the standard parameter (size of the solution set~$T$). Given all this, it is not surprising that some special variants have been introduced, such as when all threshold values are bounded by a constant or the \emph{majority} variant in which all threshold values are set to half of the neighborhood size. \citet{Cicalese14} and \citet{Cicalese15} considered a variant in which the number of rounds of the activation process is bounded.

The systematic study of \TSSshort for structural parameters was initiated by \citet{BenZwiHLN09}. Together with sequel works~\cite{Ben-ZwiHLN11,NichterleinNUW13,ChopinNNW14,DvorakKT18,Hartmann18} this yields a fairly complete understanding of tractability for structural parameters:

\begin{itemize}
	\item
	\TSSshort is \FPT when parameterized by the vertex cover number, the bandwidth, or the feedback edge number,
	\item
	the majority variant is \FPT for the neighborhood diversity or the twin-cover number,
	\item
	the majority variant is \Whard for the treedepth (or the modularwidth\footnote{Modularwidth is a structural parameter incomparable with treewidth, which is usually small on dense graphs. Modularwidth possesses a nice algebraic decomposition of a graph. For a formal definition of modularwidth, see, e.g., the work of \citet{GajarskyLO2013}.}), and
	\item
	\TSSshort{} is \FPT when parameterized by the cliquewidth if all threshold values are bounded by a constant.
\end{itemize}

Later, several variants of \TSSshort were introduced and studied, such as in a dynamic~\cite{Schierreich2023,DeligkasEGS2023}, or geometric environment~\cite{DvorakKS2024}. Notably, none of these variants considers more than one opinion being spread. 

\paragraph{Information Diffusion Models}
As we already mentioned, there are two main models for information diffusion, namely the Independent cascades model and the Linear threshold model as described above.
Originally both these models are stochastic (predefined probability of activation along each edge in the first model, random threshold for each agent in the second one). 
However, computing the expected number of agents activated by a certain seed set is \textsf{\#P}-complete~\cite{ChenWW10}.
Thus, the typical approach to estimating this number is to sample the network several times and average over the results~\cite{BudakAA2011,ChenCCKLRSWWY11,LiZCGSL2014,TongWD21,TongWD18,TongDW18}.
Therefore, it becomes essential to simulate behavior in a deterministic setting, as was argued by, e.g., \cite[Lemma~1]{LuCL2015}. 
It should be pointed out that our model is (probably) the simplest and most natural combination of the deterministic version of the Linear threshold model of~\citet{KempeKT03} described above, the spread of two nonexclusive opinions~\citet{LuCL2015}, and the balance of opinions of~\citet{GarimellaGPT17}. Let us now briefly describe the latter two. 

\paragraph{Models with Multiple Opinions}
In the \emph{Comparative IC model} of~\citet{LuCL2015}, the information diffusion follows the Independent cascades model described above, but there is a separate stage of decision on adoption of the opinion. 
Each agent can be in 4 possible states with respect to each opinion, namely: \emph{idle, suspended, adopted, rejected}. Initially, she is idle with respect to both opinions $a$ and $b$. 
Once she is informed of a neighbor adopting an opinion, she uses her own \emph{node level automaton} to decide upon adopting the opinion. 
Here, if she is idle with respect to $a$ and idle or suspended with respect to $b$, she adopts $a$ with probability $q_{a|\emptyset}$, otherwise she suspends it. 
This decision can be later \emph{reconsidered} when she adopts $b$ such that it is either adopted (with overall probability $q_{a|b}$) or ultimately rejected. 
If she gets informed about $ a$ while she has already adopted $b$, then she either adopts $a$ with probability $q_{a|b}$ or rejects it. 
Similarly, for the adoption of $b$. 
The \emph{global adoption probabilities} $q_{a|\emptyset}$, $q_{a|b}$, $q_{b|\emptyset}$, and $q_{b|a}$ are the same for the whole network. 

\citet{GarimellaGPT17} considered two opposing campaigns being propagated through the network according to the Independent cascades model from initial seed sets $I_a$ and~$I_b$. The task is to select two additional seed sets $S_a$ and $S_b$ with $|S_a| +|S_b| \le k$ so as to maximize the expected number of balanced users, that is, users exposed to information from both campaigns or none. They considered two different settings, a \emph{``heterogeneous''} one, where the propagation of each campaign is independent of the other, and a \emph{``correlated''} one, where one campaign is propagated along an edge if and only if the other is.

\paragraph{Comparison with our Model}
In the previously proposed models, each agent has \emph{opinion-specific thresholds}. It is worth mentioning that the model we propose is the first that allows a certain level of interaction between the two opinions spreading in the network that can be tuned on the agent level. Thus, our model enables the study of spreading opinions, where if an agent receives one of them, this directly affects their ability/willingness to acquire the other (which was already discussed by \citet{LuCL2015} at the network level, also cf. \citet{MyersL12}).

An important difference between the model of \citet{LuCL2015} and the one presented here is that we allow some agents to view the two opinions as complementary whereas others might view them as competing while \citet{LuCL2015} set this ``on the network level'' (i.e., it is the same for all agents) in their model.
It is worth mentioning that our model is also simpler in some aspects than the one of \citet{LuCL2015}, since in our model the thresholds are perfectly symmetrical for the opinions, i.e., they do not depend on whether the opinion to be acquired is $a$ or $b$ (in their model there might be a difference, but this is again on the network level). We discuss the possibilities of making our model asymmetric in \Cref{sec:variants}.

As we have already mentioned, the study of related problems that concern more than one opinion has been initiated in works of~\cite{MyersL12,LuCL2015}.
Both \citet{GarimellaGPT17} and \citet{BeckerCDG19} study related problems in a stochastic setting, that is, the diffusion process is randomized, and the task is to maximize the expected number of balanced agents with the given budget or to minimize the cost of achieving a prescribed number of balanced agents. In both cases, they distinguish the heterogeneous and correlated settings. While \citet{GarimellaGPT17} considers two campaigns, \citet{BeckerCDG19} generalizes it to three or more, where the agent is considered balanced if it is confronted with none or at least a given number of opinions.
In these works, the authors mostly study approximation algorithms and computational hardness of approximation, while in our work, we are interested in exact algorithms.%

\section{Preliminaries}\label{sec:preliminaries}

In this section, we provide an overview of the notation used in our work. 
Let $k\in \N$ be a natural number. We denote by $[k]$ the set $\{1,\ldots,k\}$.

\paragraph{Graph Theory}
Since we model the social network of agents as a graph, we rely on basic graph notions.
A graph is a pair~$G = (V,E)$, where the elements of~$V$ are \emph{vertices} (representing agents) and $E \subseteq \binom{V}{2}$ is the set of \emph{edges}. In some considerations we also allow \emph{loops}, that is, the edges connecting $v$ with $v$ for any $v\in V$ to be part of $E$. If $G$ is a graph, we denote the set of its vertices by $V(G)$ and the set of edges by $E(G)$, unless the sets are clear from the context. We abbreviate their sizes as $n=|V(G)|$ and $m=|E(G)|$.

Let $G=(V,E)$ be a graph, and let $v\in V$ be any vertex. The \emph{neighborhood of a vertex~$v$} is the set $N_G(v) = \{ u \in V \mid \{u,v\} \in E \}$ and the \emph{degree} of $v$, denoted $\deg_G(v)$, is the size of its neighborhood $|N_G(v)|$. We often omit the subscript that indicates a graph if there is no confusion.

For a set of vertices~$U$ \emph{the subgraph induced by~$U$} is the graph $\big(U, E \cap \binom{U}{2}\big)$.
For a set of vertices~$S$ the graph $G - S$ is the subgraph induced by the set $V \setminus S$. If the set $S$ consists of a single vertex $v$, we write $G-v$ instead of $G-\{v\}$.
For further notions of graph theory, we refer the reader to the monograph of \citet{Diestel17}.

\paragraph{Parameterized Complexity}
An instance of a \emph{parameterized problem} $L\subseteq \Sigma^*\times \N$ over a finite fixed alphabet $\Sigma$ is a pair $(x,k)\in\Sigma^*\times\N$, where $x$ is the input to the problem and $k \in \N$ is the value of \emph{parameter}.

A~parameterized problem $L$ is called \emph{fixed-parameter tractable} (is in \FPT) if it is possible to decide whether $(x,k) \in L$ in $f(k) \cdot |x|^\Oh{1}$ time, where $f\colon \N \to \N$ is a computable function. If the running time of an algorithm is $f(k) \cdot |x|^{g(k)}$, where $f,g\colon\N\to\N$ are computable functions, then the algorithm is called \emph{slice-wise polynomial} and proves that $L$ is in \XP.

The class \W contains both problems in \FPT and also problems that are not believed to be in \FPT.
Therefore, if a problem is proven \Whard by a parameterized reduction, then it is presumably not in \FPT.
A parameterized reduction from a parameterized problem $L$ to a parameterized problem $Q$ is an algorithm that on input of an instance $(x,k)$ of a problem~$L$ produces in $f(k) \cdot |x|^\Oh{1}$ time an instance $(x',k')$ of the problem $Q$ such that $(x,k) \in L$ if and only if $(x',k') \in Q$ and such that $k' \le g(k)$, where $f,g\colon \N \to \N$ are computable functions.

For a more comprehensive overview of parameterized complexity theory, we refer the reader to the monograph of \citet{CyganFKLMPPS15}.

\paragraph{Treewidth, Pathwidth, and Treedepth} As we study \twoTSSshort in sparse social networks, our intention is to use some structural parameter that is small in such networks. One of the best known parameters for sparse networks is \emph{treewidth}. 

\begin{definition}[Tree decomposition]\label{def:treeDecomposition}
	A \emph{tree decomposition} of a graph~$G$ is a triple $(T, \beta, r)$, where~$T$ is a tree rooted at node~$r$ and $\beta \colon V(T) \to 2^{V(G)}$ is a mapping that satisfies:
	\begin{enumerate}[label=\roman*)]
		\item $\bigcup_{x \in V(T)} \beta(x) = V(G)$;
		\item For every $\{u, v\} \in E(G)$ there exists a node $x \in V(T)$, such that $u, v \in \beta(x)$;
		\item For every $u \in V(G)$ the nodes $\{x \in V(T) \mid u \in \beta(x)\}$ form a connected subtree of~$T$.
	\end{enumerate}
\end{definition}

To distinguish between the vertices of a tree decomposition and the vertices of the underlying graph, we use the term \emph{node} for the vertices of a given tree decomposition.

The \emph{width} of a tree decomposition $(T, \beta, r)$ is $\max_{x \in V(T)} |\beta(x)|-1$.
The \emph{treewidth} of a graph~$G$ (denoted $\tw(G)$) is the minimum width of a~tree decomposition of~$G$ over all its decompositions. The \emph{pathwidth} of a graph~$G$ (denoted $\pw(G)$) is the minimum width of a~tree decomposition $(T, \beta, r)$ of~$G$ over all decompositions for which $T$ is a path.

Let $G = (V,E)$ be a graph.
The \emph{treedepth} of a graph, denoted $\td(G)$, is defined as follows.%
\[
\td(G) = \begin{cases}
ˇ	1 &\mbox{if } |V(G)| = 1, \\[0.2cm]
1+\min_{v\in V(G)}\td(G - v) &\mbox{if $G$ is connected with } |V(G)| > 1, \\[0.2cm]
\max_{i \in [k]} \td(G_i) & \text{\begin{minipage}{.5\textwidth}if $G_1, \ldots, G_k$ are connected components of~$G$.\end{minipage}}
\end{cases}
\]

\paragraph{Vertex Cover, $\kappa$-Path Vertex Cover, Vertex Integrity, Feedback Vertex Number}

We also consider further measures used for sparse graphs to obtain a more detailed picture of the complexity.

\begin{definition}[Vertex cover number]
	Let $G = (V,E)$ be a graph. A set $U \subseteq V$ is a \emph{vertex cover} if for each edge $\{u,v\}\in E$ at least one endpoint is part of $U$. The \emph{vertex cover number} of $G$ is the minimum size of a vertex cover in~$G$.
\end{definition}

If $U$ is a vertex cover of a graph $G$, then the graph $G - U$ contains only isolated vertices, i.e., in $G- U$, there is no path with more than one vertex as a subgraph. The next parameter relaxes this restriction and allows longer paths to remain in $G-U$.

\begin{definition}[$\kappa$-path vertex cover number]
	Let $G = (V,E)$ be a graph and let $\kappa\in\N$ be a positive integer. A set $U \subseteq V$ is a \emph{$\kappa$-path vertex cover} ($\kappa$-pvc) if $G - U$ does not contain any path with $\kappa$ vertices as a subgraph.
	The \emph{$\kappa$-path vertex cover number} of $G$ is the minimum size of a $\kappa$-pvc of~$G$.
\end{definition}

We observe that a $2$-pvc of~$G$ is, in fact, equivalent to the vertex cover of~$G$.

\begin{definition}[Vertex integrity]
	Let $G = (V,E)$ be a graph and let $q \in \mathbb{N}$ be a positive integer. A set $U \subseteq V$ is a \emph{modulator to components of size~$q$} if $G - U$ does not contain any connected component with size at least $q+1$.
	The least $q$ such that there is a modulator $U$ to components of size~$q$ with $|U| \le q$, is the \emph{vertex integrity} of $G$.%
\end{definition}

In our hardness reductions, we also use the structural parameter called the \emph{feedback vertex number}, which is defined as follows.

\begin{definition}[Feedback vertex number]
	Let $G = (V,E)$ be a graph.  A set $U \subseteq V$ is a \emph{feedback vertex set of~$G$} if ${G - U}$ is a forest, i.e., if it does not contain any cycle as a subgraph. The \emph{feedback vertex number} of $G$ is the minimum size of a feedback vertex set of~$G$.
\end{definition}

\section{Basic Observations}

One can observe that setting $T_a=S_b$ and $T_b=S_a$ leads to a balanced outcome. Indeed, with this choice, both opinions begin spreading from the same set $S_a \cup S_b$, and since the thresholds are symmetric in the two opinions, the two diffusion processes coincide, and every agent ends up acquiring either both opinions or none. Therefore, the interesting task is to find smaller seed sets. In particular, without loss of generality, we can assume that $B < |S_a|+|S_b|$, as otherwise we can use the previous observation.

It is easy to note that we can always assume $S_a \cap S_b=\emptyset$ (unless $f_i$ is defined as some function of $\deg (v)$) as otherwise we can remove $S_a \cap S_b$ from the graph and subtract $|N(v) \cap S_a \cap S_b|$ from both $f_1(v)$ and $f_2(v)$ for each $v$ remaining in the graph.
This is formalized in the following reduction rule.

\begin{rrule}\label{rr:vertexWithBothOpinions}
	Let $(G,S_a,S_b,f_1,f_2,\fin,B)$ be an instance of \twoTSSshort.
	Suppose that there exists a vertex $v \in S_a \cap S_b$.
	We remove $v$ from~$G$ and define new instance~\[(G - v, S_a \setminus \{v\}, S_b \setminus \{v\}, \widehat{f}_1, \widehat{f}_2,\fin,B),\] where for $i = 1,2$ we have
	\(
	\widehat{f}_i(u) =
	\begin{cases}
		\max\{0,f_i(u)-1\}  & \text{if } u \in N(v) \\
		f_i(u)    & \text{otherwise.}
	\end{cases}
	\)
\end{rrule}

Correctness of this rule, that is, the equivalence of the original instance $(G,S_a,S_b,f_1,f_2,\fin,B)$ with the new instance $(G - v, S_a \setminus \{v\}, S_b \setminus \{v\}, \widehat{f}_1, \widehat{f}_2,\fin,B)$ is immediate.

We use the following lemma several times in the paper to bound the length of the activation process.
It is not hard to see that a vertex can gain an opinion $c$ in round $i>1$ only if in round $i-1$ either it gained the opinion $\neg c$ or some of its neighbors gained the opinion $c$. Repeating this argument and going from round $t$ to round $0$ we obtain a walk of length $t$ in $G$ with a loop added to each vertex. 

\begin{lemma}\label{lem:rounds_walk}
    Let $G$ be a graph, $\mathring{G}$ be obtained from $G$ by adding a loop to each vertex.
	Let $t \ge 1$ and suppose that there exist sets $P^0_a$, and~$P^0_b$ such that the activation process in~$G$ from these sets takes at least~$t$ rounds, i.e., $P^t_a \neq P^{t-1}_a$ or $P^t_b \neq P^{t-1}_b$.
	Then there is a walk $w$ of length~$t-1$ in~$\mathring{G}$ such that each vertex appears at most twice on $w$.	
\end{lemma}

\begin{proof}%
	The claim is trivial for $t=1$, let us assume that $t\ge 2$.
	Let $c_t \in \{a,b\}$ be such that $P^t_{c_t} \setminus P^{t-1}_{c_t} \neq \emptyset$.
	Let~$v^t$ be an arbitrary vertex in $P^t_{c_t} \setminus P^{t-1}_{c_t}$.
	For $i = t-1, \ldots, 1$, note that for each vertex $v \in P^{i+1}_{c_{i+1}} \setminus P^{i}_{c_{i+1}}$ we have $N(v) \cap \left( P^{i}_{c_{i+1}} \setminus P^{i-1}_{c_{i+1}} \right) \neq \emptyset$ or $v \in P^{i}_{\neg c_{i+1}} \setminus P^{i-1}_{\neg c_{i+1}}$.
	Now, if $N(v^{i+1}) \cap \left( P^{i}_{c_{i+1}} \setminus P^{i-1}_{c_{i+1}} \right) \neq \emptyset$, then let~$v^i$ be an arbitrary vertex in $N(v^{i+1}) \cap \left( P^{i}_{c_{i+1}} \setminus P^{i-1}_{c_{i+1}} \right)$ and $c_i=c_{i+1}$.
	Otherwise, let $v^i =v^{i+1}$ and $c_i=\neg c_{i+1}$.
	Each vertex appears in the sequence $v^1, \ldots, v^t$ at most twice, once for each opinion and
	 $v^1, \ldots, v^t$ is a walk of length $t-1$ in~$\mathring{G}$.	
\end{proof}

Next we use the above lemma to bound the length of the activation process in terms of the size of a modulator~$U$ to components of size~$q$. 

\begin{lemma}
    Let $I=(G,S_a,S_b,f_1,f_2,\fin,B)$ be an instance of \twoTSSshort such that the graph $G$ has a modulator~$U$ to components of size~$q$ and let $k = |U|$.
    If $\fin \ge (2q+1)(2k+1)$, then the instance is equivalent to instance $I'=(G,S_a,S_b,f_1,f_2,\fin',B)$, where $\fin'=(2q+1)(2k+1)$.
\end{lemma}

\begin{proof}
    Clearly, if $I'$ is a \yesI, then so is $I$.
    Conversely, if there is a solution $T_a, T_b$ for $I$ such that the activation process stabilizes within $\fin'$ rounds, then it is also a solution for $I'$.
    Hence, assume for the sake of contradiction that there is a solution $T_a, T_b$ for $I$ such that the activation process from these sets takes at least~$\fin'+1$ rounds.
    Let $\mathring{G}$ be obtained from $G$ by adding a loop to each vertex.
    By \Cref{lem:rounds_walk}, there is a walk $w$ of length at least~$\fin'+1-1=(2q+1)(2k+1)$ in~$\mathring{G}$ such that each vertex appears at most twice on $w$.
    Vertices of $U$ appear at most $2k$ times on $w$ in total.
    Hence, removing these vertices splits the walk into at most $2k+1$ parts, removing at most $4k$ edges.
    Let $w_1$ be a longest of these parts.
    Its length is at least 
    \[\left\lceil\frac{\fin'-4k}{2k+1}\right\rceil = \left\lceil\frac{(2q+1)(2k+1)-4k}{2k+1}\right\rceil = \left\lceil\frac{(2q-1)(2k+1)+4k+2-4k}{2k+1}\right\rceil \ge 2q,\] 
    i.e., it contains at least $2q+1$ vertex occurrences.
	As $w_1$ is a walk not containing any vertex of $U$, it must be contained in a single component $C$ of $\mathring{G} - U$. 
	However, as $C$ is of size at most $q$ and $w_1$ contains each vertex at most twice, we have a contradiction. 
	This finishes the proof.
\end{proof}

As $3$-pvc is also a modulator to components of size~2, we obtain the following corollary.

\begin{corollary}\label{cor:modulatorImpliesParametricBoundONTheNumberOfRounds}
	Let $I=(G,S_a,S_b,f_1,f_2,\fin,B)$ be an instance of \twoTSSshort.
	\begin{enumerate}
		\item If $G$ has a 3-pvc of size~$k$, then we may assume that $\fin \le 10 \cdot k + 5$.
		\item If $G$ has vertex integrity~$k$, then we may assume that $\fin \in \Oh{k^2}$.
	\end{enumerate}
\end{corollary}

Next, we discuss formal relation of our model to the original \TSS problem.

\begin{lemma}\label{lem:2otss_generalizes_tss}
	Every instance $I_\TSSshort$ of \TSS can be reduced to an equivalent instance $I_\twoTSSshort$ of \twoTSS in linear time.
\end{lemma}
\begin{proof}
	Let $I_\TSSshort = (G,f,B)$. We create an instance $I_\twoTSSshort=(G,S_a=\emptyset,S_b=V(G),f_1,f_2=f,\fin,B)$ of \twoTSS by taking the same network, setting $S_a = \emptyset$, $S_b = V(G)$, treating $f_1 (v) = 1$ and $f_2 (v) = f (v)$ for each agent $v\in V$, and letting $\fin=2|V(G)|$. We claim that $I_\TSSshort$ has solution if and only if the corresponding $I_\twoTSSshort$ instance has a solution.
	
	First, assume that the instance $I_\TSSshort$ has the solution $S\subseteq V$. By setting $T_a = S$ the activation process runs as follows. At the beginning, all agents obtain opinion $b$, and agents from $S$ obtain also the opinion $a$. The process continues with spreading the opinion over the network. Since all agents already have opinion $b$, only the second threshold $f_2$ is in use. Therefore, the \twoTSSshort activation process is reduced to the activation process of the original \TSSshort problem, which by assumption spreads the opinion over the network. We chose the set $S_b = V$ and at the end of the activation process all agents have both opinions and the solution is ${T_a = S, T_b = \emptyset}$.
	
	Next, let $T_a$, $T_b$ be a solution for $I_\twoTSSshort$. 
	Then the corresponding $I_\TSSshort$ instance has solution $S=T_a$ as the budget is not exceeded and the activation process stabilizes since $f \equiv f_2$ with opinion spread all over the network as in~$I_\twoTSSshort$.
	
	Finally, it is easy to see that the reduction runs in linear time, as we only set the seed-set $S_b$ to $V$ and update the threshold function value for each agent. Hence the proof is complete.
\end{proof}

The preceding lemma formally captures that the \twoTSSshort problem is indeed a generalization of the original \TSSshort problem. Consequently, all the hardness results for \TSSshort are directly carried over to our model. Thus, \twoTSS is \NPh, \Wcomplete[P] for the size of target-sets, and \Whard for the structural parameters the treewidth, the treedepth, and the modularwidth. 
Conversely, algorithmic results for the more general \twoTSSshort{} problem carry over to \TSSshort.

In \Cref{sec:3pvc,sec:vi}, we use \Cref{lem:2otss_generalizes_tss} to give an \FPT algorithm for \TSSshort if we parameterize with the $3$-path vertex cover number and the vertex integrity of the underlying graph, respectively.

\section{Positive Results}

\subsection{Vertex Cover}\label{sec:vc}

In this section, we show that \twoTSSshort is \FPT with respect to the vertex cover number of the underlying graph.  
For the rest of this subsection, let~$U$ be a minimum size vertex cover in~$G$ and let $k=|U|$.

Our algorithm is based on a couple of reduction rules that either solve the instance directly or adjust the instance so that once no reduction rule can be applied, we end up with an equivalent instance $\mathcal{I'}$ such that the remaining budget~$B$ in $\mathcal{I'}$ is at most~$2k$. We always apply the first applicable reduction rule (including \Cref{rr:vertexWithBothOpinions}) exhaustively and in the order in which they appear.

We recall that \Cref{rr:vertexWithBothOpinions} removes all agents that initially have both opinions. The next rule adds to the solution all vertices that already have one opinion and cannot obtain the second opinion via the standard activation process.

\begin{rrule}\label{rr:vertexInSNeverGetSecondOpinion}
	Let $(G,S_a,S_b,f_1,f_2,\fin,B)$ be an instance of \twoTSSshort.
	Suppose that there is a vertex $v \in S_c$ for $c \in \{a,b\}$ with $f_2(v) > \deg(v)$.
	We return the instance $(G,S_a\cup\{v\},S_b\cup\{v\},f_1,f_2,\fin,B-1)$.
\end{rrule}
\begin{proof}[Safeness]
	Suppose that $v \in S_a$ (the other case follows by a symmetric argument).
	Since \Cref{rr:vertexWithBothOpinions} is not applicable, we get $v \notin S_b$.
	Observe that $f_2(v) > \deg(v)$ implies that if $v \notin T_b$, then $v \notin P^{\fin}_b$ --- thus the process is not balanced.
	Consequently, $v \in T_b$ must hold for every solution of the given instance. The only way to achieve this is to add $v$ to $T_b$, and we have to decrease the budget by $1$ for this operation.
	Finally, note that it does not matter whether we put~$v$ in $T_b$ or $S_b$.
	We see that the two instances are equivalent.
\end{proof}

The third rule bounds the maximum value of the threshold function for each agent.

\begin{rrule}\label{rr:vertexNotInVCfValueReduction}
	Let $(G,S_a,S_b,f_1,f_2,\fin,B)$ be an instance of \twoTSSshort and suppose that there exists a vertex $v\in V$ with $f_i(v) > \deg(v) + 1$ for any $i\in\{1,2\}$.
	Then we return $(G,S_a,S_b,\widehat{f}_1,\widehat{f}_2,\fin,B)$, where for $i\in\{1,2\}$ we have
	\mbox{$ \widehat{f}_i(v) = \min\left\{f_i(v),\deg(v)+1\right\}$} and $\widehat{f}_i(u) = f_i(u)$ for every $u\in V\setminus\{v\}$.
\end{rrule}
\begin{proof}[Safeness]
	Suppose that $v\in V$ is such that $f_1(v) \geq \deg(v) + 1$. Since the instance is reduced with respect to \Cref{rr:vertexInSNeverGetSecondOpinion}, we have $v\notin S_a\cup S_b$. Regardless of the particular value of $f_1(v)$, it is true that $v$ will never gain the first opinion by the activation process, because there are not enough neighbors. This property remains unchanged even when the value of $f_1(v)$ is decreased to $\deg(v)+1$. The same argument works for the second opinion.
\end{proof}

Finally, the last reduction rule of this section directly solves a given instance if the budget is big enough. Roughly speaking, if the budget allows us to make initial seed sets equal, then only the first threshold is applied, and the outcome is always balanced.

\begin{rrule}\label{rr:boundBudgetByVCNumber}
	Let $I = (G,S_a,S_b,f_1,f_2,\fin,B)$ be an instance of \twoTSSshort and $U$ be a minimum size vertex cover in $G$.
	If $B\geq|U\setminus S_a|+|U\setminus S_b|$, then output a trivial \yesI.
\end{rrule}
\begin{proof}[Safeness]
	Let $B \ge |U\setminus S_a|+|U\setminus S_b|$. Then we set additional seed sets $T_a,T_b$ to $T_a = U\setminus S_a$ and $T_b = U\setminus S_b$ respectively. After this initialization, every vertex $u\in U$ already has both opinions. The remaining vertices either have one opinion or no opinion at all.
	
	We denote by $v$ a vertex with a single opinion. It is clear that $v\notin U$ and $v \in S_a$ or $v \in S_b$. Because we cannot apply \Cref{rr:vertexInSNeverGetSecondOpinion} on $I$, we have that $f_2(v)$ is at most $\deg(v)$. However, we know that all neighbors of $v$ have both opinions already, and due to this, the vertex~$v$ gets a second opinion in the next round of the activation process.
	
	The last case not yet discussed is when the vertex $w$ has no opinion at all. For such vertices, we have $w\notin U$ and $w \notin S_a\cup S_b$. Since we cannot apply \Cref{rr:vertexNotInVCfValueReduction} we know that $f_1(w) \leq \deg(w) + 1$. On the one hand, when the inequality is strict, then $w$ gets both opinions by a similar argument as in the previous case. On the other hand, when $f_1(w) = \deg(w) + 1$, then $w$ remains without any opinion at all.
\end{proof}

Assuming that the instance is reduced with respect to the previously defined reduction rules, we can employ a brute-force-like algorithm to show the desired result.

\begin{theorem}\label{thm:twoTSSIsFPTWrtVC}
	\twoTSS is fixed-parameter tractable parameterized by the vertex cover number $k$.
\end{theorem}
\begin{proof}
	Let $(G,S_a,S_b,f_1,f_2,\fin,B)$ be an instance of \twoTSSshort, $k$ be the vertex cover number of~$G$, and~$U$ be a $k$-sized vertex cover.
	We assume that the input instance is reduced with respect to the presented reduction rules.
	Consequently, since \Cref{rr:boundBudgetByVCNumber} cannot be applied, we have $B < 2k$.
	
	We define an equivalence relation~$\sim$ on $V\setminus U$ such that for $u,v\in V\setminus U$ we have $u\sim v$ if and only if $N(u)=N(v)$, and for $i=1,2$ we have $f_i(u) = f_i(v)$.
	The cardinality of the quotient set $(V\setminus U)/{\sim}$ is at most $2^k\cdot(k+1)\cdot(k+1)$ as we have $2^k$ possible different neighborhoods and $k+1$ possible values for each of $f_1$ and $f_2$. Together with the vertex cover vertices $U$ we have
	$k+2^k\cdot(k+1)\cdot(k+1)$ different types of vertices from which we must select at most $\min\{B,k\}$ vertices into the seed set~$T_a$ and at most $\min\{B,k\}$ into~$T_b$. That is, at most $2k$ vertices with at most $k+2^k\cdot(k+1)^2+1 \le 2^k\cdot(k+2)^2$ options each, giving us at most $\left(2^k(k+2)^2\right)^{2k} =2^{2k^2 + \Oh{k\log k}}$
	options in total. For each option, we add the guessed vertices into $T_a$ and $T_b$, respectively, and simulate the activation process. If the activation process results in a balanced outcome for at least one option, we return \emph{yes}; otherwise, the result is \emph{no}. The simulation can be done in quadratic time, and since we try all possible sets of additional seed sets, the algorithm is correct. Overall, the algorithm runs in $2^{2k^2 + \Oh{k\log k}}\cdot n^2$ time, that is, \twoTSSshort is in \FPT when parameterized by the vertex cover number.
\end{proof}

We note here that the best algorithm for the single opinion \TSSshort parameterized by the vertex cover number $k$, which is due to \citet{BanerjeeMP22}, runs in $2^\Oh{k\log k}\cdot n^\Oh{1}$ time. We leave as an open question whether our algorithm for more general \twoTSSshort can be improved to match the running time of \TSSshort, or whether our running time is asymptotically optimal under the standard theoretical assumptions, such as the Exponential-Time Hypothesis~\cite{ImpagliazzoP01}.

\subsection{Three-Path Vertex Cover}\label{sec:3pvc}

In the previous section, we parameterized by the vertex cover number to obtain a feasible algorithm.  Next, we generalize the result and use the $3$-path vertex cover number as a parameter to obtain \FPT{} algorithm for \twoTSSshort. Thus, suppose that the input graph has a $3$-pvc~$U$ with size bounded by the parameter~$k$. Note that $3$-pvc is also a modulator to components of size~2. 

First, note that \Cref{cor:modulatorImpliesParametricBoundONTheNumberOfRounds} allows us to first guess for every $u \in U$ and each opinion~$c \in \{a,b\}$ the round~$r_c(u)$ in which~$u$ receives~$c$ (or~$\infty$ if this never happens), as there are $(10k+5)^{2k}$ possible options in total.
However, we only consider guesses where for each $u \in S_c \cap U$ we have $r_c(u)=0$.
Next, we apply \Cref{rr:vertexWithBothOpinions,rr:vertexInSNeverGetSecondOpinion,rr:vertexNotInVCfValueReduction} exhaustively.
We observe that \Cref{rr:vertexNotInVCfValueReduction} yields for $v \in X$ (where $X = V(G) \setminus U$) that
\(
\widehat{f}_i(v) \le \deg(v)+1 \le k+2 \,,
\)
since every component of $G \setminus U$ has size at most two.

Our algorithm is based on the so-called $N$-fold integer programming. 
$N$-fold integer programming ($N$-fold IP) is the problem of minimizing a separable convex objective (for us it suffices to minimize a linear objective) over a set of constraints of a specific form.
Let $r \in \N$ and $s_i, t_i  \in \N$ for every $i \in [N]$.
The IP has $d = \sum_{i \in [N]} t_i$ variables partitioned into $N$ so-called \emph{bricks}.
The $i$-th brick is denoted $x^{(i)}$ and contains $t_i$ variables.
The constraints have the following form:
\begin{align}
D_1 x^{(1)} + D_2 x^{(2)} + \cdots + D_N x^{(N)} &= \bm{b}_0    \label{eq:NFold:linking}  \\
A_i x^{(i)}                        &= \bm{b}_i    & \forall i \in [N]    \label{eq:NFold:local}\,   \\
\bm{l}_i \le x^{(i)}                 &\le \bm{u}_i  & \forall i \in [N]    \label{eq:NFold:box},
\end{align}
where we have $D_i \in \Z^{r \times t_i}$, $A_i \in \Z^{s_i \times t_i}$, $\bm{b}_0 \in \Z^{r}$, $\bm{b}_i \in \Z^{s_i}$ and $\bm{l}_i,\bm{u}_i \in \Z^{t_i}$ for every $i \in [N]$. 
Let us denote $s = \max_{i \in [N]} s_i$ and recall that the dimension is $d = \sum_{i \in [N]} t_i$.
The constraints~\eqref{eq:NFold:linking} are the so-called \emph{linking constraints}, followed by the \emph{local constraints}~\eqref{eq:NFold:local}, and the \emph{box constraints}~\eqref{eq:NFold:box}.
That is, there are $r$ linking constraints and, apart from these, each variable is only involved in at most $s$ local constraints and one box constraint.
We call variables with a coefficient $0$ in all linking constraints \emph{local}.

The current best algorithm solving the $N$-fold IP in \( (rs\Delta)^{\Oh{r^2s+rs^2}} d\log d \log \operatorname{val}_{\max}\) time is by \citet{EisenbrandHKKLO19} (see also \citet{EisenbrandHK18,KouteckyLO18}), where \mbox{$\Delta = \max_{i \in [N]} \Big( \max \big( \|D_i\|_\infty, \|A_i\|_\infty \big) \Big)$} and $\operatorname{val}_{\max}$ is the maximum feasible value of the objective.

\paragraph*{$N$-Fold IP for a Guess}
Suppose now that we are given a fixed guess $(r_c(u))_{c \in \{a,b\}, u \in U}$ of the activation rounds for all vertices in $U$.
Recall that for $u \in S_a$ we must have $r_a(u)=0$ and for $u \in S_b$ we must have $r_b(u)=0$, otherwise the guess is incorrect.
By \Cref{cor:modulatorImpliesParametricBoundONTheNumberOfRounds} we may assume that $\fin \le {10\cdot k + 5}$.
That is, for each $u \in U$ we are given the activation rounds $r_a(u)$ and $r_b(u)$.
We design an $N$-fold IP with binary variables $x^{c,t}_v$ for $t \in \{ 0,\ldots, \fin, \fin+1 \}$, $c \in \{ a,b \}$, and $v \in V$ (later we will add some auxiliary variables).
The intuition is that a variable $x^{c,t}_v$ is set to~$1$ if and only if the vertex $v \in V$ has opinion~$c$ in round~$t$.
We stress here that the values of these variables for vertices in~$U$ are already fixed by our guess above (that is, these are already fully determined by the box constraints).
Similarly, we require that $x^{c,0}_v = 1$ for every $v \in S_c$.
As we required that $r_c(u)=0$ whenever $u \in S_c \cap U$, this is consistent.

The intuition is that we will impose the constraints for each vertex $v \in V$ and an opinion~$c \in \{ a,b \}$ in such a way that we secure that
\begin{enumerate}
	\item it is impossible for $v$ to receive $c$ prior or in the round $t$ if $x^{c,t}_v = 0$ and
	\item $v$ receives $c$ at latest in the round $t$ if $x^{c,t}_v = 1$.
\end{enumerate}
Note that this implies that $x^{c,t}_v$ indeed expresses whether $v$ has $c$ in round $t$.
For vertices in the set $U$ we use the linking constraints for this purpose whereas for the vertices in $X$ we use the local constraints.
We discuss the box constraints for each variable separately (if we say a variable~$x$ is binary, we assume we have added the constraints $0 \le x \le 1$).
Note that this is in compliance with our knowledge about each vertex, which is somewhat orthogonal for these groups of vertices.
For a vertex $u \in U$, we know (by the guess) the exact round in which $u$ receives each opinion, but we do not know which neighbors are activated before this happens.
For a vertex $v \in V \setminus U$ we know which neighbors in $U$ have which opinion prior to $v$ receiving it (note that this amounts to all but at most one neighbor of~$v$) but we do not know the round in which $v$ receives its opinions (or if it actually happens).

\paragraph{Working with (Local) Expressions in IP}
We often use the notation $[\expr]$ that evaluates to $1$ if $\expr$ holds and evaluates to $0$ if not and introduces a new auxiliary variable $z^{[\expr]}$ (to store the result).
To illustrate this, suppose we want to evaluate the expression $[x^{a}+ x^{b} = 0]$ for two binary variables $x^{a}$ and $x^{b}$.
First, we add an auxiliary binary variable~$z^{ab} = z^{[x^{a}+ x^{b} = 0]} = [x^{a}+ x^{b} = 0]$ to store the result---as we shall see this requires; see, e.g., \cite[Proof of Theorem~3.4]{KnopKM20}
\begin{itemize}
	\item coefficients bounded by~$1$ (this is thanks to the fact that $x^{a}$ and $x^{b}$ are binary, i.e., this value depends on the size of the domain of the variables) and
	\item $\Oh{1}$ additional constraints.
\end{itemize}
To that end, we add the following constraints.
\begin{align}
	z^{ab} \le 1 - x^{a} \quad\text{and}\quad z^{ab} \le 1 - x^{b} \label{eq:binaryexpr:lowerbounds} \\
	z^{ab} \ge 1 - x^{a} - x^{b} \label{eq:binaryexpr:upperbound}
\end{align}
The equivalence of the two is rather straightforward to verify; see~\Cref{tab:binaryexpr:values}.
It is worth noting that if $x^{a}$ and $x^{b}$ are from the same brick, then all the introduced constraints are local.
\begin{table}
    \begin{center}
	\begin{tabular}{cccccc}
		$x^{a}$ & $x^{b}$ & $z^{ab}$ & $1-x^{a}$ & $1-x^{b}$ & $1-x^{a}-x^{b}$ \\ \hline
		0 & 0 & 1 & 1 & 1 & 1\\
		1 & 0 & 0 & 0 & 1 & 0\\
		0 & 1 & 0 & 1 & 0 & 0\\
		1 & 1 & 0 & 0 & 0 & $-1$
	\end{tabular}
	\end{center}
 	\caption{%
		The table of values for $z^{ab} = z^{[x^{a}+ x^{b} = 0]} = [x^{a}+ x^{b} = 0]$.
		Note that rows where $z^{ab} = 0$ fulfill that at least one of the conditions \eqref{eq:binaryexpr:lowerbounds} evaluates the right-hand side to~$0$ while the right-hand side of \eqref{eq:binaryexpr:upperbound} evaluates to at most~$0$.
		For the first row, the right-hand side of all the conditions \eqref{eq:binaryexpr:lowerbounds} and \eqref{eq:binaryexpr:upperbound} evaluates to~$1$.
	}\label{tab:binaryexpr:values}
\end{table}

In our $N$-fold ILP, we use another ``trick''---it is possible to multiply binary variables (i.e., one can rewrite such an expression to linear constraints).
Again, let us illustrate this on $z^{ab} = x^{a} \cdot x^{b}$.
We add the following constraints; see \Cref{tab:multiply:values}.
\begin{align}
	z^{ab} \le x^{a} \quad\text{and}\quad z^{ab} \le x^{b} \label{eq:multiply:lowerbounds} \\
	z^{ab} \ge x^{a} + x^{b} - 1 \label{eq:multiply:upperbound}
\end{align}
It is not hard to see that one can extend this approach for more than two variables.
Again, if both $x^{a}$ and $x^{b}$ are from the same brick, then all the introduced constraints are local.
\begin{table}
        \begin{center}
	\begin{tabular}{cccc}
		$x^{a}$ & $x^{b}$ & $z^{ab}$ & $x^{a} + x^{b}-1$ \\ \hline
		0 & 0 & 0 & $-1$ \\
		1 & 0 & 0 & 0 \\
		0 & 1 & 0 & 0 \\
		1 & 1 & 1 & 1
	\end{tabular}
	\end{center}
	\caption{%
	The table of values for $z^{ab} = x^{a} \cdot x^{b}$.
	Note that rows where $z^{ab} = 0$ fulfill that at least one of the conditions \eqref{eq:multiply:lowerbounds} evaluates the right-hand side to~$0$ while the right-hand side of \eqref{eq:multiply:upperbound} evaluates to at most~$0$.
	For the last row, the right-hand side of all the conditions \eqref{eq:multiply:lowerbounds} and \eqref{eq:multiply:upperbound} evaluates to~$1$.
	}\label{tab:multiply:values}
\end{table}

\paragraph{Linking (Global) Constraints}
First, we describe the linking constraints of the constructed $N$-fold IP, i.e., the constraints for vertices $u \in U$.
Recall that for these vertices we know the values $r_a(u)$ and $r_b(u)$; therefore, these conditions are easier to formulate.
We first describe the types of constraints that we will use, and then we explain which of them we use based on case distinction on $r_c(u)$.
We will use the following constraint types:
\begin{itemize}
 \item For $1 \le r_c(u) < \infty$, the following constraint ensures that $u$ has enough neighbors to gain opinion $c \in \{a,b\}$ in the round $r_c(u)$ provided that it did not have any opinion in round $r_c(u)-1$:
        \begin{equation}
			\sum_{v \in N(u)} x^{c,r_c(u) - 1}_v \ge f_1(u) \label{eq:twotssNFold:globalLowerbound:case1} \,.
		\end{equation}
  \item For $2 \le r_c(u) < \infty$, the following constraint ensures that $u$ does not have enough neighbors to gain opinion $c \in \{a,b\}$ in the round prior to $r_c(u)$ provided that it did not have any opinion in round $r_c(u)-2$:		
        \begin{equation}
			\sum_{v \in N(u)} x^{c,r_c(u) - 2}_v \le f_1(u)-1 \label{eq:twotssNFold:globalUpperbound:case1} \,.
		\end{equation}
  	\item For $r_c(u)= \infty$, the following constraint ensures that $u$ does not have enough neighbors to gain opinion $c \in \{a,b\}$ even in round $\fin+1$ provided that it did not have any opinion in round $\fin$:		
        \begin{equation}
			\sum_{v \in N(u)} x^{c,\fin}_v \le f_1(u)-1 \label{eq:twotssNFold:globalUpperbound:case2} \,.
		\end{equation}	
	\item For $r_a(u) < r_b(u) < \infty$ the following constraint ensures that $u$ actually has at least $f_2(u)$ neighbors with~$b$ in the round $r_b(u)$ and, therefore, gains the opinion $b$:
	    \begin{equation}
            \sum_{v \in N(u)} x^{b,r_b(u) - 1}_v \ge f_2(u) \label{eq:twotssNFold:globalLowerbound:case2} \,.
        \end{equation}
    \item For $1 \le r_a(u) < r_b(u) < \infty$ the following constraint ensures that~$u$ does not acquire~$b$ before or together with~$a$:
        \begin{equation}
			\sum_{v \in N(u)} x^{b,r_a(u) - 1}_v \le f_1(u)-1 \label{eq:twotssNFold:globalUpperbound:case3} \,.
		\end{equation}  
	\item Finally, for $r_a(u)+2 \le r_b(u) < \infty$ the following constraint ensures that $u$ does not gain opinion $b$ prior to round $r_b(u)$, that is, it should not have enough neighbors with opinion~$b$:
		\begin{equation}
			\sum_{v \in N(u)} x^{b,r_b(u) - 2}_v \le f_2(u)-1 \label{eq:twotssNFold:globalUpperbound:case4} \,.
		\end{equation}
\end{itemize}
Equipped with these constraints, for each vertex $u \in U$ we do the following case distinction (we only list the cases with $r_a(u) \le r_b(u)$, the others are symmetric):
\begin{description}
	\item[$r_a(u) = r_b(u)=0$:] This case requires no constraints.
	\item[$r_a(u) = r_b(u)=1$:] In this case, we use for both $a$ and $b$ constraint \eqref{eq:twotssNFold:globalLowerbound:case1}.
	\item[$2\le r_a(u) = r_b(u) < \infty$:] In this case, we use for both $a$ and $b$ constraints \eqref{eq:twotssNFold:globalLowerbound:case1} and \eqref{eq:twotssNFold:globalUpperbound:case1}.
	\item[$r_a(u) = r_b(u) = \infty$:] In this case, we use for both $a$ and $b$ constraint \eqref{eq:twotssNFold:globalUpperbound:case2}.
	\item[$r_a(u) = 0$, $r_b(u) \ge 1$:] No constraints are required for $a$.
	We add constraint \eqref{eq:twotssNFold:globalLowerbound:case2} for $b$. 
	If $2 \le r_b(u) < \infty$, then we also add constraint \eqref{eq:twotssNFold:globalUpperbound:case4} for $b$.
	\item[$1 = r_a(u) < r_b(u) < \infty$:] We add constraint \eqref{eq:twotssNFold:globalLowerbound:case1} for~$a$.
	We add constraints \eqref{eq:twotssNFold:globalLowerbound:case2} and \eqref{eq:twotssNFold:globalUpperbound:case3} for~$b$. 
	If $r_b(u) \ge 3$, then we also add constraint \eqref{eq:twotssNFold:globalUpperbound:case4} for~$b$.
	\item[$2 \le r_a(u) < r_b(u) < \infty$:] We add constraints \eqref{eq:twotssNFold:globalLowerbound:case1} and~\eqref{eq:twotssNFold:globalUpperbound:case1} for~$a$.
	We add constraints \eqref{eq:twotssNFold:globalLowerbound:case2} and~\eqref{eq:twotssNFold:globalUpperbound:case3} for~$b$. 
	If $r_a(u)+2 \le r_b(u)$, then we also add constraint \eqref{eq:twotssNFold:globalUpperbound:case4} for~$b$.
\end{description}
Note that with these constraints we specify for each vertex that the round $r_c(u)$ is the round when for the first time the appropriate threshold is met.
We do not consider the case $r_a(u) < \infty$ and $r_b(u)= \infty$ (or vice-versa) as this does not lead to a balanced outcome.
Note that the total number of these constraints is~$\Oh{k}$ since there are at most five for each vertex $u\in U$.
That is, the parameter~$r$ of the $N$-fold IP fulfills~$r \in \Oh{k}$.

\paragraph{Local Constraints}
We again want to set the lower- and upperbounds on the number of neighbors with a particular opinion for each vertex.
Note that this time we have to do this ``on the fly'' as we do not know the rounds $r_a(v)$ and $r_b(v)$ for the vertices $v \in V \setminus U$.
In order to do so, we define for a vertex $v \in X$, ${t \in \{ 0,\ldots, \fin+1 \}}$, and $c \in \{ a,b \}$ the value~$\varphi^{c,t}_v$ expressing the number of vertices~$u \in N(v) \cap U$ for which we have guessed $r_c(u) \le t$.
Note that from the view of the IP, this is not a variable, but a constant (depending only on the guess). 
Now, for a value $f \in \N$ we add an auxiliary local binary variable~$z^{c,t}_v(f)$ with the following behavior:
\begin{equation}\label{eq:twotssNFold:settingZ}
z^{c,t}_v(f) =
	\begin{cases}
		\left[ x^{c,t-1}_w + \varphi^{c,t-1}_v \ge f \right] & \text{if } N(v)\setminus U = \{w\}\text{ for some }w \\[\medskipamount]
		\left[ \varphi^{c,t-1}_v \ge f \right] & \text{otherwise.} \\
	\end{cases} \,
\end{equation}
Note that this constraint is only needed when~$f - \varphi^{c,t-1}_v \in \{0,1\}$ (otherwise, we can directly set~$z^{c,t}_v(f)$ to either~$0$ or~$1$).
Consequently, we only need~\eqref{eq:twotssNFold:settingZ} when it translates to $[x^{c,t-1}_w]$; in any case, this is a linear constraint as it directly translates to $z^{c,t}_v(f) = x^{c,t-1}_w$.
Note that in certain sense variables $z^{c,t}_v(f)$ only serves as a ``syntactic sugar''.
Now, we require the following.
\begin{align}
	x^{c,t}_v &\ge z^{c,t}_v(f_1(v)) \cdot \left[ x^{a,t-1}_v + x^{b,t-1}_v = 0 \right] \label{eq:twotssNFold:boundingXT:1} \\
	x^{c,t}_v &\ge z^{c,t}_v(f_2(v)) \cdot \left[ x^{\neg c,t-1}_v \right] \label{eq:twotssNFold:boundingXT:2} \\
	x^{c,t}_v &\ge x^{c,t-1}_v \label{eq:twotssNFold:boundingXT:3} \\
	x^{c,t}_v &\le z^{c,t}_v(f_1(v)) \cdot \left[ x^{a,t-1}_v + x^{b,t-1}_v = 0 \right] + z^{c,t}_v(f_2(v)) \cdot \left[ x^{\neg c,t-1}_v \right] + x^{c,t-1}_v \label{eq:twotssNFold:boundingXT:4}
\end{align}
It is not hard to see that these constraints can be made linear using~$\Oh{1}$ additional constraints and auxiliary variables (while keeping the largest coefficient at most~$2$).

Note that due to the condition~\eqref{eq:twotssNFold:boundingXT:3} we have that the variables $x^{c,t}_v$ can only be $0$ up to certain round and must be $1$ from this round on (if they ever evaluate to $1$).
Observe that at least one right-hand side of conditions \eqref{eq:twotssNFold:boundingXT:1} and \eqref{eq:twotssNFold:boundingXT:2} evaluates to~$0$ for every value of $t$.
Let us now proceed case by case:
\begin{enumerate}
	\item
	If $v$ does not have any of the opinions, then $\left[ x^{a,t-1}_v + x^{b,t-1}_v = 0 \right]$ evaluates to~$1$ and therefore the whole right-hand side of \eqref{eq:twotssNFold:boundingXT:1} evaluates to $1$ if and only if $v$ has at least $f_1(v)$ neighbors with opinion~$c$; thus forcing $x^{c,t}_v = 1$.
	Otherwise, it evaluates to~$0$ imposing no restriction on the value of $x^{c,t}_v$.
	\item 
	Suppose $v$ already has opinion $b$.
	Then $\left[ x^{b,t-1}_v \right]$ evaluates to~$1$ and therefore the whole right-hand side of \eqref{eq:twotssNFold:boundingXT:2} evaluates to $1$ if and only if $v$ has at least $f_2(v)$ neighbors with opinion~$a$; thus forcing $x^{a,t}_v = 1$.
	Otherwise, it evaluates to~$0$ imposing no restriction on the value of $x^{a,t}_v$.
\end{enumerate}
Let us now investigate the condition~\eqref{eq:twotssNFold:boundingXT:4}.
Note that the right-hand side of~\eqref{eq:twotssNFold:boundingXT:4} evaluates to~$1$ if and only if at least one right-hand side of~\eqref{eq:twotssNFold:boundingXT:1}--\eqref{eq:twotssNFold:boundingXT:3} evaluates to $1$.
That is, if none of these evaluates to~$1$ (they impose no restriction on the binary variable $x^{c,t}_v$), then this condition ensures it evaluates to~$0$; otherwise, it does not impose any restriction on the value of $x^{c,t}_v$ as $x^{c,t}_v$ is binary.

We require that the process is stable after $\fin$ rounds.
In order to do that we add constraints for the round $\fin+1$ and require that no vertex gained a new opinion in that round.
\begin{equation}
	x^{c,\fin}_v = x^{c,\fin+1}_v \qquad \forall v \in V,\, \forall c \in \{a,b\} \label{eq:twotssNFold:ProcessStabilizes}
\end{equation}
Note that this is again a set of local constraints even for vertices in the modulator~$U$ as these are independent for individual vertices.

Finally, we have to require the following---imposing that when the process stabilizes each vertex has either acquired both $a$ and $b$ or none of them.
\begin{equation}
	x^{a,\fin}_v = x^{b,\fin}_v \qquad \forall v \in X \label{eq:twotssNFold:FinalAEqualsB}
\end{equation}
Note that the local constraints are independent for any nonadjacent vertices, that is, each brick $x^{(i)}$ of variables consists of all variables $x^{c,t}_v$, where either $v \in C_i$ for some connected component $C_i$ of $G \setminus U$ or all of them are related to the same vertex $v \in U$.
As each connected component of $G \setminus U$ is of size at most $2$ and $\fin \in \Oh{k}$, the parameter~$s$ of the $N$-fold IP fulfills~$s \in \Oh{k}$.

It is not difficult to verify that the parameter~$t$ of the $N$-fold IP fulfills~$t \in \Oh{k}$.

\paragraph*{Objective Value}
If we now minimize
\(
	\sum_{v \in V} \left( x^{a,0}_v + x^{b,0}_v \right) \,,
\)
we get the size of the smallest solution which is compatible with our initial guess.
Thus, we may verify whether the minimum overall value is at most the budget~$B$ (for the fixed guess).

\begin{theorem}\label{thm:3-pvc}
	Let $G$ be a graph with $3$-pvc number~$k$.
	There is a $k^\Oh{k^3} \cdot n\log^2 n$ time algorithm that solves \twoTSS.
\end{theorem}
\begin{proof}
	For each guess of activation rounds for the vertices in a 3-pvc $U$, we build an \mbox{$N$-fold~IP}.
	Each $N$-fold IP can be solved in time $k^\Oh{k^3} \cdot n \log^2 n$ by the result of \citet{EisenbrandHKKLO19}, as we have already observed that we have~$r,s \in \Oh{k}$ for the parameters of the $N$-fold IP and the maximum feasible value of the objective is $2n$.
	Since we generate $k^{\Oh{k}}$ such IPs, the upper bound on the total running time of our algorithm follows.
	
	Now we show the correctness.
	
	Let us first assume~$(G, S_a, S_b, f_1, f_2, \fin, B)$ is a \yesI of \twoTSS.
	Thus, there are sets~$T_a, T_b$ with $|T_a| + |T_b| \le B$ and a corresponding activation process $\mathcal{P}$.
	Our aim is to show that at least one IP of the generated IPs is feasible and its optimum is at most~$B$.
	Consider the guess, where for every $u \in U$ and every $c \in \{a,b\}$ we have $r_c(u) = 0$ if $u \in S_c \cup T_c$, $r_c(u) = \infty$ if $u \notin P^{\fin}_c$, and $r_c(u)$ is such that $u \notin P^{r_c(u)-1}_c$ and $u \in P^{r_c(u)}_c$ otherwise.
	
	For every vertex $v \in V$, $c \in \{a,b\}$ and $t \in \{0,\ldots,\fin+1\}$, we let
	\begin{equation}\label{eq:twotssNFold:settingXFromActivationProcess}
	x^{c,t}_v = \left[ v \in P^t_c \right].
	\end{equation}
	From that we get that 
	\[
	\varphi^{c,t}_v = \left|\{ u \in N(v) \cap U \mid r_c(u) \le t\}\right| =  \left| N(v) \cap U \cap P^t_c \right|.
	\]
	Since $\mathcal{P}$ stabilizes in at most $\fin$ rounds, conditions \eqref{eq:twotssNFold:ProcessStabilizes} are fulfilled.
	Note that this allows us to define the variables $z^{c,t}_v(f)$ for $v \in X$ according to~\eqref{eq:twotssNFold:settingZ}.
	Now, our aim is to prove that, under this assignment of variables, all constraints are fulfilled.
	First, we observe that~\eqref{eq:twotssNFold:FinalAEqualsB} holds for all $v \in X$, which follows from the fact that $v \in P^{\fin}_a$ if and only if $v \in P^{\fin}_b$.
	Verifying the appropriate subset of global constraints~\eqref{eq:twotssNFold:globalLowerbound:case1}--\eqref{eq:twotssNFold:globalUpperbound:case4} is rather straightforward, since the left-hand sides of these constraints evaluate the number of neighbors a vertex~$v \in U$ has in $P^{r_c(v)-2}_c$, $P^{r_c(v)-1}_c$, $P^{\fin}_c$, or $P^{r_{\neg c}(v)-1}_c$, respectively.
	It remains to verify~\eqref{eq:twotssNFold:boundingXT:1}--\eqref{eq:twotssNFold:boundingXT:4} which we do by case distinction.
	\begin{enumerate}
		\item
		Suppose $v \in P^{t-1}_c$; thus $x^{c, t-1}_v = x^{c,t}_v = 1$.
		Trivially, we get that~\eqref{eq:twotssNFold:boundingXT:1}--\eqref{eq:twotssNFold:boundingXT:3} are fulfilled.
		Finally,~\eqref{eq:twotssNFold:boundingXT:4} holds, since $x^{c,t-1}_v = 1$ and thus the right-hand side evaluates to (at least)~$1$.
		\item
		Suppose $v \notin P^{t}_a \cup P^{t}_b$.
		We immediately see that the right-hand sides of~\eqref{eq:twotssNFold:boundingXT:2} and \eqref{eq:twotssNFold:boundingXT:3} evaluate to~$0$.
		It holds that~$[x^{a,t-1}_v + x^{a,t-1}_v = 0] = 1$, since $v \notin P^{t-1}_a \cup P^{t-1}_b$ (i.e., ${x^{a,t-1}_v = x^{b,t-1}_v = 0}$).
		However, since $v \notin P^t_c$, we get $z^{c,t}_v(f_1(v)) = 0$ by~\eqref{eq:twotssNFold:settingZ}.
		Thus, from~\eqref{eq:twotssNFold:boundingXT:1}--\eqref{eq:twotssNFold:boundingXT:3} we get a (void) constraint $x^{c,t}_v \ge 0$.
		Furthermore, since the right-hand side of~\eqref{eq:twotssNFold:boundingXT:4} is a sum of the right-hand sides of~\eqref{eq:twotssNFold:boundingXT:1}--\eqref{eq:twotssNFold:boundingXT:3}, we get $x^{c,t}_v \le 0$; which is fulfilled if (and only if) $x^{c,t}_v = 0$.
		\item
		Suppose $v \notin P^{t-1}_a \cup P^{t-1}_b$ but $v \in P^t_a$ (the case $v \in P^t_b$ follows from symmetry); that is, $x^{a,t-1}_v = x^{b,t-1}_v = 0$ and consequently $[x^{a,t-1}_v + x^{a,t-1}_v = 0] = 1$.
		We get $z^{a,t}_v(f_1(v)) = 1$ (which is if and only if $|P^{t-1}_a \cap N(v)| \ge f_1(v)$ by~\eqref{eq:twotssNFold:settingZ}).
		Thus, the right-hand side of~\eqref{eq:twotssNFold:boundingXT:1} evaluates to~$1$.
		It is not hard to see that the right-hand sides of \eqref{eq:twotssNFold:boundingXT:2} and \eqref{eq:twotssNFold:boundingXT:3} both evaluate to~$0$.
		Thus, the right-hand side of~\eqref{eq:twotssNFold:boundingXT:4} is $1$, and this collection of constraints is fulfilled if (and only if) $x^{a,t}_v = 1$.
		\item
		Suppose $v \in P^t_b$ and $v \in P^{t-1}_a \setminus P^{t-1}_b$; that is, $x^{a,t-1}_v = 1$, $x^{b,t-1}_v = 0$ and consequently $[x^{a,t-1}_v + x^{a,t-1}_v = 0] = 0$.
		We immediately see that the right-hand sides of~\eqref{eq:twotssNFold:boundingXT:1} and~\eqref{eq:twotssNFold:boundingXT:3} both evaluate to~$0$.
		The right-hand side of~\eqref{eq:twotssNFold:boundingXT:2} evaluates to~$1$, since by~\eqref{eq:twotssNFold:settingZ}~$z^{b,t}_v(f_2(v)) = 1$ if (and only if) $| P^{t-1}_b \cap N(v) | \ge f_2(v)$.
		It follows that the right-hand side of~\eqref{eq:twotssNFold:boundingXT:4} evaluates to~$1$ as well and thus this collection of constraints is fulfilled if (and only if) $x^{b,t}_v = 1$.
	\end{enumerate}
	All the above arguments hold ``if and only if'', therefore the IP has a solution if and only if it describes an activation process.
	We give here a formal description for the first case (as the others follow in a very similar manner):
	\begin{enumerate}
		\item
		We show that if $x^{c, t-1}_v = 1$, then $x^{c,t}_v = 1$.
		Trivially, by \eqref{eq:twotssNFold:boundingXT:3} we get $x^{c,t}_v = 1$.
		The as the left-hand sides of \eqref{eq:twotssNFold:boundingXT:1} and \eqref{eq:twotssNFold:boundingXT:2} are at most~1, these constraints are fulfilled.
		Finally, the right-hand side of \eqref{eq:twotssNFold:boundingXT:4} evaluates to~1 and thus all the constraints hold.
		Therefore, we can see that in the process~$\mathcal{P}$ obtained by \eqref{eq:twotssNFold:settingXFromActivationProcess} we have that $v \in P^{t-1}_c$ implies $v \in P^{t}_c$ (for all $v \in X$ and $c \in \{a,b\}$).
	\end{enumerate}
	The theorem now follows since the mapping~\eqref{eq:twotssNFold:settingXFromActivationProcess} is revertible.
\end{proof}

The preceding theorem, together with \Cref{lem:2otss_generalizes_tss}, also implies an algorithm for the original \TSS problem parameterized by the $3$-path vertex cover running in time~$k^\Oh{k^3} \cdot n\log^2 n$ time. 

\begin{corollary}
	Let $G$ be a graph with $3$-pvc number~$k$.
	There is a $k^\Oh{k^3} \cdot n\log^2 n$ time algorithm solving \TSS{}.
\end{corollary}

We note that there is a simultaneously and independently announced algorithm of \citet{BanerjeeMP22} for this parameter running in $2^{2^\Oh{k}}\cdot n^\Oh{1}$ time. That is, our algorithm significantly improves the known upper bound for \TSS parameterized by the $3$-path vertex cover number. Interestingly, the algorithm of \citet{BanerjeeMP22} is based on Lenstra's algorithm. Thus, our work is another example of the phenomenon observed by \citet{BlazejKPS2024}, showing that an N-fold ILP formulation for such parameterizations often enables faster algorithms.

\subsection{Vertex Integrity}\label{sec:vi}
Let us now discuss the parameterization by the vertex integrity.
We note that this parameter is very closely related to the one investigated in the previous section.
For the rest of this section let $U$ be a modulator to components of size~$k$ such that $|U| \le k$.
By \Cref{cor:modulatorImpliesParametricBoundONTheNumberOfRounds} the number of rounds of any activation process is bounded by a function of the parameter, hence we rely on the same approach as in the previous section.
That is, we may guess `if' and `when' the vertices in the modulator should receive the two opinions; the linking constraints remain the same.
The activation process takes in this case $\Oh{k^2}$ rounds and thus the parameter $r$ of the $N$-fold IP model fulfills $r \in \Oh{k^2}$.
A slight change is needed in the case of local constraints as here we are facing further challenges---the number of neighbors of a vertex in a connected components of $G-U$ is now the parameter.
Again, we rely on rewriting result of \citet{KnopKM20}.

This change, however, only affects the variable $z^{c,t}_v(f)$.
For a vertex $v \in X$, ${t \in \{ 0,\ldots, \fin+1 \}}$, and $c \in \{ a,b \}$, let $\varphi^{c,t}_v$ be the number of vertices~$u \in N(v) \cap U$ for which we have guessed $r_c(u) \le t$.
Now, for a value $f \in \N$ we have an auxiliary binary variable~$z^{c,t}_v(f)$ with the following behavior:
\begin{equation}\label{eq:twotssNFoldVI:settingZ}
z^{c,t}_v(f) = \left[ \sum_{w \in N(v) \setminus U} x^{c,t-1}_w + \varphi^{c,t-1}_v \ge f \right] \,.
\end{equation}
Note that $\sum_{w \in N(v) \setminus U} x^{c,t-1}_w \in [k]$ and therefore we need this expression only in the case $1 \le f - \varphi^{c,t-1}_v \le |N(v) \setminus U| \le k$ (otherwise, we can directly set~$z^{c,t}_v(f)$ to either~$0$ or~$1$).

We will again present the rewriting of this constraint into a linear local constraint in a simplified setting (which then gives us the recipe to rewrite the above expression in general).
Suppose, we have $z = [x \ge \xi]$ for $1 \le \xi$ and let $0 \le x \le k$.
Then, we add the following constraint with binary variable $z$ (note that now the largest coefficient used in the IP is~$k$):
\begin{equation*}
	\xi \cdot z \le x \qquad\text{and}\qquad k \cdot z \ge 1 + x - \xi
\end{equation*}
Let us now examine this closer by a case analysis.
\begin{description}
	\item[$x < \xi$]
		In this case the right-hand side of the right condition is not positive and therefore does not impose any additional restriction on the variable $z$ (as $z$ is binary).
		The right-hand side of the left condition, however, is at most $\xi-1$ and thus requires $z=0$.
	\item[$x \ge \xi$]
		Now, the right-hand side of the left condition does not impose any restriction as this condition holds for $z \in \{ 0,1 \}$.
		The right-hand side of the right condition in this case is strictly positive thus implying $z=1$.
\end{description}
If we now use the ``new'' variables $z^{c,t}_v(f_1)$ and $z^{c,t}_v(f_2)$ in \eqref{eq:twotssNFold:boundingXT:1}--\eqref{eq:twotssNFold:boundingXT:4}, we obtain a valid IP model for a guess for the activation of the modulator.
Note that the local constraints are independent for vertices belonging to distinct connected components of $G-X$.
There are at most $k$ vertices in each connected component and we have to verify the constraints for each of them in $\Oh{k^2}$ rounds of the activation process; therefore, the parameter~$s$ of the $N$-fold IP fulfills~$s \in \Oh{k^3}$.
Thus, we arrive at the following.

\begin{theorem}\label{thm:vertex_integrity}
        There is an algorithm that solves \twoTSS in $k^\Oh{k^8} \cdot n\log^2 n$ time, where $k$ is the vertex integrity of the input graph.
\end{theorem}

Similarly to the algorithm for the parameterization by the $3$-path vertex cover number, \Cref{thm:vertex_integrity} together with \Cref{lem:2otss_generalizes_tss} yield an algorithm for the \TSS problem parameterized by the vertex-integrity of the network. 

\begin{corollary}
	There is an algorithm that solves \TSS in $k^\Oh{k^8} \cdot n\log^2 n$ time, where $k$ is the vertex integrity of the input graph.
\end{corollary}

To the best of our knowledge, the complexity picture of the \TSS problem with respect to the vertex integrity was not known and therefore, our result pushes the tractability boundary of this problem with respect to structural parameterization.

\subsection{Treewidth, Number of Rounds, and Maximum Threshold}

In this section, we give an \FPT algorithm for the combination of parameters treewidth, number of rounds, and the maximum threshold. It is clear that, under reasonable theoretical assumptions, there cannot be an algorithm that solves \twoTSSshort in \FPT time with respect to the treewidth alone, since the original \TSSshort is \Whard for this parameter \cite{Ben-ZwiHLN11}.

To simplify the description of our algorithm, we will describe the computation in terms of a \emph{nice tree decomposition} of the underlying graph $G$. Before we dive deep into the definition, let us fix some basic notation.

Recall that a tree decomposition of a graph~$G$ is a triple consisting of a tree~$T$, its root node~$r$, and a mapping $\beta \colon V(T) \to 2^{V(G)}$ that assigns to each node $x$ of the tree its \emph{bag} $\beta(x)$.
For a tree decomposition $(T, \beta, r)$ and a node $x \in V(T)$, we denote by $V_x$ the union of vertices in $\beta(x)$ and in~$\beta(y)$ for all descendants~$y$ of~$x$ in~$T$. By $E_x$, we denote the set of edges introduced in the subtree of $T$ rooted in $x$. Altogether, we denote by $G_x$ the graph $(V_x, E_x)$. We also denote by $\alpha(x)$ the set $V_x \setminus \beta(x)$.

\begin{definition}[Nice tree decomposition~{\cite[p.~168]{CyganFKLMPPS15}}]\label{def:niceTreeDecomposition}
	A tree decomposition of a graph~$G$ is \emph{nice} if $\beta(r) = \emptyset$, and each node $x \in V(T)$ is of one of the following five types:
	\begin{itemize}
		\item \emph{Leaf node}---$x$ has no children and $\beta(x) = \emptyset$;
		\item \emph{Introduce vertex node}---$x$ has exactly one child~$y$ and $\beta(x) = \beta(y) \cup \{u\}$ for some $u \in V(G) \setminus \beta(y)$;
		\item \emph{Introduce edge node}---$x$ has exactly one child~$y$, $\beta(x) = \beta(y)$, and an edge $\{u, v\} \in E(G)$ for $u, v \in \beta(x)$ is introduced (i.e., $E_x = E_y \cup \{\{u,v\}\}$);
		\item \emph{Forget node}---$x$ has exactly one child~$y$ and $\beta(x) = \beta(y) \setminus \{u\}$ for some $u \in \beta(y)$;
		\item \emph{Join node}---$x$ has exactly two children~$y, z$ and $\beta(x) = \beta(y) = \beta(z)$,
	\end{itemize}
	and each edge $e \in E(G)$ is introduced exactly once. %
\end{definition}

We note that any given tree decomposition of width $\w$ can be transformed into a nice one of the same width in $\Oh{\w^2\cdot n}$ time~\cite[Lemma 7.4, see also the discussion on page 168]{CyganFKLMPPS15}.

In this section, we prove the following theorem.

\begin{theorem}\label{thm:twoTSSIsFPTWrtRoundsFmaxTw}
	\twoTSS with a maximum threshold of $f_{\max}$ can be solved in \mbox{$(\fin \cdot f_{\max}+1) ^\Oh{\w} \cdot n$} time on graphs of treewidth at most~$\w$.
\end{theorem}

To prove \Cref{thm:twoTSSIsFPTWrtRoundsFmaxTw}, we describe a dynamic programming algorithm that works on a nice tree decomposition of~$G$.
The idea is as follows.
To keep track of the entire process for both opinions from the somewhat limited viewpoint of only vertices in~$\beta(x)$ we slightly alter the activation process and, more importantly, work with the thresholds for a specific opinion (either~$a$ or~$b$).

For a node $x$ of the decomposition the subproblems that we solve prescribe for every vertex $v \in \beta(x)$
\begin{enumerate}[(i)]
	\item \label{item:tw:pattern1}
	in exactly which round the vertex~$v$ gains the opinion $a$ and in which the opinion $b$---value $r_a(v),r_b (v) \in \{ 0, \ldots, \fin \} \cup \{ \infty \}$, respectively (the value $\infty$ corresponds to not gaining the opinion at all),%
	\item \label{item:tw:pattern2}
	at least how many neighbors of~$v$ in the graph~$G_x$ are already active in opinion~$a$ (opinion~$b$) in round $r_a-1$ ($r_b-1$), if applicable---value $g_a(v), g_b(v) \in \{ 0, \ldots, f_{\max} \}$, respectively (these values serve to validate that the vertex actually gains the opinion in the desired round),
	\item \label{item:tw:pattern3}
	at most how many neighbors of~$v$ in the graph~$G_x$ are active in opinion~$a$ (opinion~$b$) in round $r_a-2$ ($r_b-2$), if applicable---value $h_a(v), h_b(v) \in \{ 0, \ldots, f_{\max} \}$, respectively (these values serve to validate that the vertex does not gain the opinion prior to the desired round), and
	\item \label{item:tw:pattern4}
	at most how many neighbors of $v$ in the graph~$G_x$ are active in opinion~$a$ (opinion~$b$) in round $r_b-1$ ($r_a-1$), if applicable---value $\eta_a(v), \eta_b(v) \in \{ 0, \ldots, f_{\max} \}$, respectively (these values serve to validate that in case that the other opinion should be obtained earlier than the current one, the vertex does not gain the opinion prior or simultaneously with the other one).
\end{enumerate}
Note that \Cref{item:tw:pattern4} upper-bounds, e.g., the number of neighbors with opinion $b$ just prior to gaining $a$ (i.e., the other opinion), whereas \Cref{item:tw:pattern2} lower-bounds, e.g., the number of neighbors with opinion $a$ two rounds prior to gaining $a$ (i.e., the same opinion).
We say that the tuple~$(r_a,r_b,g_a,g_b,h_a,h_b,\eta_a,\eta_b)$ is a \emph{solution pattern for~$G_x$}.

The values $g_c$, $h_c$, and $\eta_c$, $c \in \{a,b\}$, are introduced as lower (upper) bounds, respectively, since the important thing is whether the number of active neighbors is strictly below the threshold, or at least the threshold, not the exact number of the neighbors. This also includes that it is completely irrelevant whether there are $f_{\max}$ or more active neighbors. This simplifies some parts of the algorithm, e.g., the computation in the forget node. 

Before we continue, we restrict the solution patterns we use---afterwards we only work with these (valid) solution patterns.

Fix a node~$x$ and a vertex $v \in \beta(x)$.
We say that a solution pattern~$(r_a,r_b,g_a,g_b,h_a,h_b,\eta_a,\eta_b)$ is \emph{valid for~$v$} if the following holds
\begin{itemize}
	\item
	$r_a(v) = \infty$ if and only if $r_b(v) = \infty$,
	\item
	$v \in S_c$ implies $r_c(v) = 0$ for each $c \in \{ a,b \}$,
	\item
	$r_c(v) \in \{0,\infty\}$ implies $g_c(v) = 0$, $r_c(v) \in \{0,1\}$ implies $h_c(v) = 0$ for each $c \in \{ a,b \}$ (these values are ignored, but we want the pattern to always be an octuple).
	\item 
	$r_c(v) \le r_{\neg c}(v)$ or $r_{\neg c}(v) \in \{0, \infty\}$ implies $\eta_c(v)=0$ and $r_{\neg c}(v)-1=r_c(v)-2$ implies $\eta_c(v)=h_c(v)$ for each $c \in \{ a,b \}$.
\end{itemize}
We say that
a solution pattern is \emph{valid for~$x$} if it is valid for all $v \in \beta(x)$.

Now, we define the modified activation process performed on $G_x$ for some node $x$.
While it follows the standard process for vertices of $\alpha(x)$, it is determined by a solution pattern for the vertices of $\beta(x)$.
Let~$x$ be a node, let $(r_a,r_b,g_a,g_b,h_a,h_b,\eta_a,\eta_b)$ be a valid solution pattern for~$x$, and let $T_a^x,T_b^x \subseteq V_x$ be a prospective solution with $r^{-1}_c(0)= (S_c \cup T_c^x) \cap \beta(x)$ for each $c \in \{ a,b \}$.
By a \emph{modified activation process}~$\mathcal{\widehat{P}}(x,r_a,r_b,T^x_a,T^x_b)$ we mean the following process starting from the sets $\widehat{P}_c^0 = T^x_c \cup (S_c \cap V_x)$ for each $c \in \{ a,b \}$.
For $i \ge 1$ and $c \in \{ a,b \}$ we set
\begin{align*}
\widehat{P}^i_c &=
\Big\{ v \in \alpha(x) \setminus \big(\widehat{P}_a^{i-1} \cup \widehat{P}_b^{i-1}\big) \,\Big|\, \big|N_{G_x}(v) \cap \widehat{P}_c^{i-1}\big| \ge f_1(v) \Big\} \\
&\cup
\Big\{ v \in \alpha(x) \cap \widehat{P}_{\neg c}^{i-1} \,\Big|\, \big|N_{G_x}(v) \cap \widehat{P}_c^{i-1}\big| \ge f_2(v) \Big\}\\
&\cup
\big\{ v \in \beta(x) \,\big|\, i = r_c(v) \big\}\\
&\cup \widehat{P}_c^{i-1} \,.
\end{align*}
Note again that the difference to the usual activation process consist in treating vertices in~$\beta(x)$ in a different way (by taking into account~$r_a(v)$ and~$r_b(v)$).

The modified activation process~$\mathcal{\widehat{P}}(x,r_a,r_b,T^x_a,T^x_b)$ is \emph{viable}, roughly speaking, if it corresponds to an activation process for~$G_x$ in which we treat $g_a,g_b$ as threshold values (lower bounds) and $h_a,h_b,\eta_a,\eta_b$ as upper bounds for vertices in~$\beta(x)$.
More formally, $\mathcal{\widehat{P}}(x,r_a,r_b,T^x_a,T^x_b)$ is \emph{viable for a solution pattern~$(r_a,r_b,g_a,g_b,h_a,h_b,\eta_a,\eta_b)$} if
\begin{enumerate}
	\item $\widehat{P}_a^{\fin}= \widehat{P}_a^{\fin+1} = \widehat{P}_b^{\fin}= \widehat{P}_b^{\fin+1}$,
	\item
	for all $v \in \beta(x)$ with $1 \le r_c(v) < \infty$ we have
	\begin{enumerate}
	  \item
	$g_c(v) \le \Big| N_{G_x}(v) \cap \widehat{P}^{r_c(v) - 1}_c \Big|$ and
	  \item
	$\Big| N_{G_x}(v) \cap \widehat{P}^{r_c(v) - 2}_c \Big| \le h_c(v)$
	\end{enumerate}
	for each $c \in \{ a,b \}$ where we set $\widehat{P}^{-1}_a = \widehat{P}^{-1}_b = \emptyset$, 
	\item 
    if $1 \le r_a(v) < r_b(v) < \infty$, then $\Big| N_{G_x}(v) \cap \widehat{P}^{r_a(v) - 1}_b \Big| \le \eta_b(v)$, and if $1 \le r_b(v) < r_a(v) < \infty$, then $\Big| N_{G_x}(v) \cap \widehat{P}^{r_b(v) - 1}_a \Big| \le \eta_a(v)$, and
	\item
	for all $v \in \beta(x)$ with $r_c(v) = \infty$ we have $\big| N_{G_x}(v) \cap \widehat{P}^{\fin}_c \big| \le h_c(v)$ for each $c \in \{ a,b \}$.
\end{enumerate}

A solution~$T^x_a,T^x_b \subseteq V_x$ \emph{complies} with a solution pattern~$(r_a,r_b,g_a,g_b,h_a,h_b,\eta_a,\eta_b)$ if we have ${r^{-1}_c(0) = (S_c \cup T^x_c) \cap \beta(x)}$ for each $c \in \{ a,b \}$ and the modified activation process~$\mathcal{\widehat{P}}(x,r_a,r_b,T^x_a,T^x_b)$ is viable for~$(r_a,r_b,g_a,g_b,h_a,h_b,\eta_a,\eta_b)$.
The \emph{size} of the solution~${T^x_a,T^x_b \subseteq V_x}$ is simply $|T^x_a| + |T^x_b|$.

The dynamic programming table $\DP_x$ (for a node~$x$) stores for each valid solution pattern the size of a smallest solution that complies with the pattern (or~$\infty$ if no such solution exists).

It is easy to observe that the sets $T_a,T_b \subseteq V$ form a solution if and only if they comply with the solution pattern $(\emptyset,\emptyset,\emptyset,\emptyset,\emptyset,\emptyset,\emptyset,\emptyset)$ at the root node $r$, where $\emptyset$ represents functions with an empty domain.
We formalize it in the following observation.

\begin{observation}\label{obs:treewidthSolutionAtRoot}
	Sets $T_a,T_b \subseteq V$ form a solution if and only if they comply with the solution pattern $(\emptyset,\emptyset,\emptyset,\emptyset,\emptyset,\emptyset,\emptyset,\emptyset)$ at the root node $r$.
\end{observation}
\begin{proof}
	At first, we note that if $\beta(x)$ is the empty set, then the modified activation process is equal to the standard activation process.
	
	Suppose that the sets $T_a,T_b\subseteq V$ form a solution of an instance $(G,S_a,S_b,f_1,f_2,\fin,B)$ of \twoTSSshort. According to \Cref{def:niceTreeDecomposition} we know that $\beta(r) = \emptyset$. Since the elements of a solution pattern $(r_a,r_b,g_a,g_b,h_a,h_b,\eta_a,\eta_b)$ are defined for every $v\in\beta(x)$, it follows that the only solution pattern for the root node $r$ has $r_a = r_b = g_a = g_b = h_a = h_b = \eta_a = \eta_b = \emptyset$ and the sets $T_a,T_b$ comply with this solution pattern.
	
	On the other hand, let $T_a,T_b\subseteq V_r=V$ be two sets that comply with $SP = (\emptyset,\emptyset,\emptyset,\emptyset,\emptyset,\emptyset,\emptyset,\emptyset)$ in the root node $r$, that is, the modified activation process $\mathcal{\widehat{P}}(x,\emptyset,\emptyset,T_a^r,T_b^r)$ is viable for $(\emptyset,\emptyset,\emptyset,\emptyset,\emptyset,\emptyset,\emptyset,\emptyset)$. Since the modified process is viable for $SP$, the activation process stabilizes at the latest in round $\fin + 1$, that is, $P_c^\fin = P_c^{\fin+1}$, $c\in\{a,b\}$. As stated above, the modified process is equal to the standard process. Furthermore, the graph $G_r$ for which the sets $T_a^r$ and $T_b^r$ form a solution is equal to the graph $G$; thus, $T_a$ and $T_b$ form a solution.
\end{proof}

Hence, the solution is found simply by comparing the value $\DP_r[\emptyset,\emptyset,\emptyset,\emptyset,\emptyset,\emptyset,\emptyset,\emptyset]$ computed to~$B$. It remains to show how to compute the value of $\DP_x$.

\paragraph{Leaf Node}\label{sec:leaf}
The dynamic programming table of a leaf node~$x$ has only a single entry $\DP_x[\emptyset,\emptyset,\emptyset,\emptyset,\emptyset,\emptyset,\emptyset,\emptyset] = 0$.
Clearly, the graph~$G_x$ is an empty graph, the solution $(\emptyset, \emptyset)$ complies with the pattern $(\emptyset,\emptyset,\emptyset,\emptyset,\emptyset,\emptyset,\emptyset,\emptyset)$, and is of size~$0$.

\paragraph{Introduce Vertex Node}\label{sec:intro_vert}
Let~$x$ be a node introducing a vertex~$v$ with the child node~$y$.
Note that the newly introduced vertex~$v$ is isolated in~$G_x$ (as the edges incident with it can only be introduced once the vertex is present).
Let~$\bm{p}=(r_a,r_b,g_a,g_b,h_a,h_b,\eta_a,\eta_b)$ be a valid solution pattern for~$x$.
We use~$\bm{p}|_{\beta(y)}$ to denote the restriction of the pattern~$\bm{p}$ to vertices in~$\beta(y)$.
\begin{itemize}
	\item If $g_a(v) \ge 1$ or $g_b(v) \ge 1$, then let $\DP_x[\bm{p}] = \infty$, since no solution exists.
	\item Otherwise we let
	\[\DP_x[\bm{p}] = \DP_y\mkern-5mu\big[\bm{p}|_{\beta(y)}\big] + \cost(a, \bm{p}, v) + \cost(b, \bm{p}, v),\]
	where
	\begin{equation}\label{eq:cost}
	\cost(c,\bm{p},v)=
	\begin{cases}
	1&\text{if }r_c(v)=0\text{ and }v \notin S_c\\
	0&\text{ otherwise.}
	\end{cases}
	\end{equation}
\end{itemize}

\begin{lemma}\label{lem:twDP:IntroduceVertex}
	Let~$x$ be a node introducing a vertex~$v$ with the child node~$y$.
	If $\DP_y$ was correctly computed and $\DP_x$ is computed using the algorithm for the introduce vertex node, then $\DP_x$ is also computed correctly.
\end{lemma}

\begin{proof}
	Let~$\bm{p}=(r_a,r_b,g_a,g_b,h_a,h_b,\eta_a,\eta_b)$ be a valid solution pattern for~$x$.
	First, note that $\bm{p}|_{\beta(y)}$ is valid for $y$, as $\bm{p}$ is valid for all $v \in \beta(x)$ and the validity does not depend on the graph, only on the octuple, the seed-set $S_a$, and the seed-set $S_b$.
	Hence the entry $\DP_y\mkern-5mu\left[\bm{p}|_{\beta(y)}\right]$ is well defined.
	
	Now suppose that $g_a(v) \ge 1$. This implies that $r_a(v) \notin \{0,\infty\}$, as otherwise $\bm{p}$ would not be valid.
	Suppose, for contradiction, that there is a solution $T^x_a,T^x_b \subseteq V_x$  that complies with $\bm{p}$. Then, since $\mathcal{\widehat{P}}(x,r_a,r_b,T^x_a,T^x_b)$ is viable, we have $1 \le g_a(v) \le \big| N_{G_x}(v) \cap \widehat{P}^{r_a(v) - 1}_a \big|$, which is a contradiction, since $v$ is isolated in $G_x$. Therefore, if $g_a(v) \ge 1$, then no solution complies with $\bm{p}$. By a similar argument, this also holds if $g_b(v) \ge 1$ and the answer of the algorithm is correct in these cases.
	
	To finish the proof we need the following two claims.
\begin{claim}\label{cla:tw_IntroVert_from_child}
 If there is a solution $T^y_a,T^y_b \subseteq V_y$ of size $s$ that complies with $\bm{p}|_{\beta(y)}$, then there is also a solution $T^x_a,T^x_b \subseteq V_x$ of size $s +  \cost(a, \bm{p}, v) + \cost(b, \bm{p}, v)$ that complies with $\bm{p}$.
\end{claim}	

\begin{claimproof}	
	Let $T^y_a,T^y_b \subseteq V_y$ be a solution  of size $s$ that complies with $\bm{p}|_{\beta(y)}$.
	We let $T^x_a=T^y_a \cup \{v\}$ if $r_a(v)=0$ and $v \notin S_a$ and $T^x_a=T^y_a$ otherwise. Similarly, we let $T^x_b=T^y_b \cup \{v\}$ if $r_b(v)=0$ and $v \notin S_b$ and $T^x_b=T^y_b$ otherwise. By the definition of $\cost$, we get $|T^x_a|+|T^x_b|=|T^y_a|+|T^y_b|+\cost(a, \bm{p}, v) + \cost(b, \bm{p}, v)$ as desired. We also have $r^{-1}_c(0)= (S_c \cup T_c^x) \cap \beta(x)$ for each $c \in \{ a,b \}$.
	
	Now, if we denote by $\widehat{P}^{i}_c$ the sets obtained in the process $\mathcal{\widehat{P}}(y,r_a|_{\beta(y)},r_b|_{\beta(y)},T^y_a,T^y_b)$ and by $\overline{P}^{i}_c$ the sets obtained in the process $\mathcal{\widehat{P}}(x,r_a,r_b,T^x_a,T^x_b)$, then we have $\overline{P}^{i}_c = \widehat{P}^{i}_c$ for each $i < r_c(v)$ and $\overline{P}^{i}_c = \widehat{P}^{i}_c \cup \{v\}$ for $i \ge r_c(v)$ as $v$ is isolated in $G_x$. Note that since $T^y_a,T^y_b$ complies with $\bm{p}|_{\beta(y)}$, the process $\mathcal{\widehat{P}}(y,r_a|_{\beta(y)},r_b|_{\beta(y)},T^y_a,T^y_b)$ is viable. Therefore, we have $\overline{P}_a^{\fin}= \overline{P}_a^{\fin+1} = \overline{P}_b^{\fin}= \overline{P}_b^{\fin+1}$ as $\widehat{P}_a^{\fin}= \widehat{P}_a^{\fin+1} = \widehat{P}_b^{\fin}= \widehat{P}_b^{\fin+1}$ and $r_a(v) = \infty$ if and only if $r_b(v)=\infty$ according to the validity of $\bm{p}$. Furthermore, for all $v' \in \beta(y)$ with $1 \le r_c(v') < \infty$ we have $g_c(v') \le \big| N_{G_x}(v') \cap \overline{P}^{r_c(v') - 1}_c \big|$ and $\big| N_{G_x}(v') \cap \overline{P}^{r_c(v') - 2}_c \big| \le h_c(v)$ for each $c \in \{ a,b \}$, where we set $\overline{P}^{-1}_a = \overline{P}^{-1}_b = \emptyset$, and the same holds for $v'=v$ as $v$ is isolated in $G_x$, $g_c(v)=0$, and $h_c(v) \ge 0$.
	Similarly, for each $v' \in \beta(y)$ we have that if $1 \le r_a(v') < r_b(v') < \infty$, then $\Big| N_{G_x}(v') \cap \overline{P}^{r_a(v') - 1}_b \Big| \le \eta_b(v')$, and if $1 \le r_b(v') < r_a(v') < \infty$, then $\Big| N_{G_x}(v') \cap \overline{P}^{r_b(v') - 1}_a \Big| \le \eta_a(v')$, and the same holds for $v'=v$, if applicable, as $\eta_a(v) \ge 0$ and $\eta_b(v) \ge 0$.
	Finally, for all $v' \in \beta(y)$ with $r_c(v') = \infty$ we have $\big| N_{G_x}(v') \cap \overline{P}^{\fin}_c \big| \le h_c(v')$ for each $c \in \{ a,b \}$ and again the same holds for $v'=v$, if applicable, as $v$ is isolated in $G_x$ and $h_c(v) \ge 0$.
	Therefore, the process $\mathcal{\widehat{P}}(x,r_a,r_b,T^x_a,T^x_b)$ is viable for $\bm{p}$ and the solution $T^x_a,T^x_b$ is compatible with $\bm{p}$.
\end{claimproof}	

\begin{claim}\label{cla:tw_IntroVert_to_child}
 If there is a solution $T^x_a,T^x_b \subseteq V_x$ of size $s$ that complies with $\bm{p}$, then there is also a solution $T^y_a, T^y_b \subseteq V_y$ of size $s -  \cost(a, \bm{p}, v) - \cost(b, \bm{p}, v)$ that complies with $\bm{p}|_{\beta(y)}$.
\end{claim}

\begin{claimproof}
	Let $T^x_a,T^x_b \subseteq V_x$ be a solution of size $s$ that complies with $\bm{p}$.
	Let $T^y_a=T^x_a \setminus \{v\}$ and $T^y_b=T^x_b \setminus \{v\}$. Since $v \in T^x_a$ if and only if $r_a(v)=0$ and $v \notin S_a$, we have $|T^y_a|=|T^x_a|-  \cost(a, \bm{p}, v)$ and similarly for $T^y_b$, therefore the size of $T^y_a, T^y_b$ is as required. We also have $r^{-1}_c(0)= (S_c \cup T_c^y) \cap \beta(y)$ for each $c \in \{ a,b \}$.
	
	If we denote by $\overline{P}^{i}_c$ the sets obtained in the process $\mathcal{\widehat{P}}(x,r_a,r_b,T^x_a,T^x_b)$ and by $\widehat{P}^{i}_c$ the sets obtained in the process $\mathcal{\widehat{P}}(y,r_a|_{\beta(y)},r_b|_{\beta(y)},T^y_a,T^y_b)$, then we have $\widehat{P}^{i}_c = \overline{P}^{i}_c$ for each $i < r_c(v)$ and $v \in \overline{P}^{i}_c$ and $\widehat{P}^{i}_c = \overline{P}^{i}_c \setminus \{v\}$ for $i \ge r_c(v)$ as $v$ is isolated in $G_x$. Note that since $T^x_a,T^x_b$ complies with $\bm{p}$, the process $\mathcal{\widehat{P}}(x,r_a,r_b,T^x_a,T^x_b)$ is viable. Therefore, we have $\overline{P}_a^{\fin}= \overline{P}_a^{\fin+1} = \overline{P}_b^{\fin}= \overline{P}_b^{\fin+1}$ and thus $\widehat{P}_a^{\fin}= \widehat{P}_a^{\fin+1} = \widehat{P}_b^{\fin}= \widehat{P}_b^{\fin+1}$ as $r_a(v) = \infty$ if and only if $r_b(v)=\infty$ according to the validity of $\bm{p}$. Furthermore, for all $v' \in \beta(y)$ with $1 \le r_c(v') < \infty$ we have $g_c(v') \le \big| N_{G_x}(v') \cap \widehat{P}^{r_c(v') - 1}_c \big|$ and $\big| N_{G_x}(v') \cap \widehat{P}^{r_c(v') - 2}_c \big| \le h_c(v)$ for each $c \in \{ a,b \}$, where we set $\widehat{P}^{-1}_a = \widehat{P}^{-1}_b = \emptyset$, as $v$ is isolated in $G_x$.
	Also for each $v' \in \beta(y)$ we have that if $1 \le r_a(v') < r_b(v') < \infty$, then $\Big| N_{G_x}(v') \cap \widehat{P}^{r_a(v') - 1}_b \Big| \le \eta_b(v')$, and if $1 \le r_b(v') < r_a(v') < \infty$, then $\Big| N_{G_x}(v') \cap \widehat{P}^{r_b(v') - 1}_a \Big| \le \eta_a(v')$
	Finally, for all $v' \in \beta(y)$ with $r_c(v') = \infty$ we have $\big| N_{G_x}(v') \cap \widehat{P}^{\fin}_c \big| \le h_c(v')$ for each $c \in \{ a,b \}$, since $v$ is isolated in $G_x$.
	Therefore, the process $\mathcal{\widehat{P}}(y,r_a|_{\beta(y)},r_b|_{\beta(y)},T^y_a,T^y_b)$ is viable for $\bm{p}|_{\beta(y)}$ and the solution $T^y_a,T^y_b$ complies with $\bm{p}|_{\beta(y)}$.
\end{claimproof}	
	
	In summary, if $T^x_a,T^x_b \subseteq V_x$ is a minimum size solution that complies with $\bm{p}$ and is of size $s$, then the algorithm did not set $\DP_x[\bm{p}]$ to $s'$ with $s' < s$. 
	This could only happen if $\DP_y[\bm{p}|_{\beta(y)}] = s' -  \cost(a, \bm{p}, v) - \cost(b, \bm{p}, v)$.
	Since $\DP_y[\bm{p}|_{\beta(y)}]$ was computed correctly, this would imply the existence of a solution $\overline{T}^y_a,\overline{T}^y_b \subseteq V_y$ of size $s' -  \cost(a, \bm{p}, v) - \cost(b, \bm{p}, v)$ and, by \Cref{cla:tw_IntroVert_from_child}, also the existence of a solution
	$\overline{T}^x_a,\overline{T}^x_b \subseteq V_x$ of size $s'$, contradicting the minimality of $T^x_a,T^x_b$. 
	Furthermore, by \Cref{cla:tw_IntroVert_to_child}, there is a solution $T^y_a, T^y_b \subseteq V_y$ of size $s -  \cost(a, \bm{p}, v) - \cost(b, \bm{p}, v)$ that complies with $\bm{p}|_{\beta(y)}$. Thus, since $\DP_y[\bm{p}|_{\beta(y)}]$ was computed correctly, $\DP_y[\bm{p}|_{\beta(y)}] \le s -  \cost(a, \bm{p}, v) - \cost(b, \bm{p}, v)$ and therefore $\DP_x[\bm{p}]$ is set to at most $s$ by the algorithm. Therefore, $\DP_x[\bm{p}]=s$ and it is computed correctly.
\end{proof}

\paragraph{Introduce Edge Node}\label{sec:intro_edge}
Let~$x$ be a node introducing an edge~$\{u,v\}$ with the child node~$y$.
Let~$\bm{p}=(r_a,r_b,g_a,g_b,h_a,h_b,\eta_a,\eta_b)$ be a valid solution pattern for~$x$.
Based on this pattern, we compute a solution pattern~$\bm{p}'=(r_a,r_b,g'_a,g'_b,h'_a,h'_b,\eta'_a,\eta'_b)$ which we then use for the lookup in the table $\DP_y$.
We let $\bm{p}'|_{\beta(x) \setminus\{u,v\}} = \bm{p}|_{\beta(x) \setminus\{u,v\}}$ and $\bm{p}'|_{\{u,v\}}$ be created using the following for each $c \in \{ a,b \}$.
We only describe the cases for $r_c(v) \le r_c(u)$, the others are symmetric.
In this case, we set $g'_c(v) = g_c(v)$, $h'_c(v) = h_c(v)$, and $\eta'_c(v)=\eta_c(v)$.
\begin{itemize}
	\item If $r_c(v) = r_c(u)$, then we set $g'_c(u) = g_c(u)$, $h'_c(u) = h_c(u)$, and $\eta'_c(u)=\eta_c(u)$.
	\item If $r_c(v) + 1 = r_c(u)\le \fin$, then we set $g'_c(u) = \max(0,g_c(u)-1)$, $h'_c(u) = h_c(u)$, and $\eta'_c(u)=\eta_c(u)$.
	\item If $r_c(v) + 1 < r_c(u) \le \fin$, then we set $g'_c(u) = \max(0,g_c(u)-1)$ and $h'_c(u) = h_c(u)-1$.
	If also $r_c(v) + 1 \le r_{\neg c}(u) < r_c(u)$, then we set $\eta'_c(u)=\eta_c(u)-1$, otherwise we set $\eta'_c(u)=\eta_c(u)$.
	\item If $r_c(v) \le \fin$ and $r_c(u) = \infty$, then we set $g'_c(u) = g_c(u) = 0$, $h'_c(u) = h_c(u)-1$, and $\eta'_c(u)=\eta_c(u)=0$.
\end{itemize}
If~$\bm{p}'$ contains a $-1$, then we let $\DP_x[\bm{p}]= \infty$.
Otherwise, we let $\DP_x[\bm{p}] = \DP_y[\bm{p}']$.

\begin{lemma}\label{lem:twDP:IntroduceEdge}
	Let~$x$ be a node introducing an edge~$\{u,v\}$ with the child node~$y$.
	If $\DP_y$ was correctly computed and $\DP_x$ is computed using the algorithm for the introduce edge node, then $\DP_x$ is also computed correctly.
\end{lemma}
\begin{proof}
	Let~$\bm{p}=(r_a,r_b,g_a,g_b,h_a,h_b,\eta_a,\eta_b)$ be a valid solution pattern for~$x$ and $\bm{p}'$ be as computed by the algorithm.

\begin{claim}
    If $\bm{p}'$ contains a $-1$, then there is no solution that complies with $\bm{p}$.
\end{claim}

\begin{claimproof}
    Suppose first that $h'_c(u)=-1$ for some $c \in \{ a,b \}$. 
    This implies that $h_c(u)=0$ and either $r_c(v) + 1 < r_c(u) \le \fin$ or $r_c(v) \le \fin$ and $r_c(u) = \infty$. 
    Assume, for the sake of contradiction, that there is a solution $T^x_a,T^x_b \subseteq V_x$  that complies with $\bm{p}$. Consider the sets $\widehat{P}^{i}_c$ obtained in the process $\mathcal{\widehat{P}}(x,r_a,r_b,T^x_a,T^x_b)$. As $u$ and $v$ are adjacent in $G_x$, we have $v \in  N_{G_x}(u) \cap \widehat{P}^{r_c(v)}_c$. If $r_c(u) \le \fin$, then $r_c(u) - 2 \ge r_c(v)$, which implies that $\big| N_{G_x}(u) \cap \widehat{P}^{r_c(u) - 2}_c \big| \ge 1 > h_c(v)$.
	If $r_c(u) = \infty$, then this implies that $\big| N_{G_x}(u) \cap \widehat{P}^{\fin}_c \big| \ge 1 > h_c(v)$.
	In both cases, this contradicts the viability of the process $\mathcal{\widehat{P}}(x,r_a,r_b,T^x_a,T^x_b)$.
	Therefore, if $h'_c(u)=-1$, then no solution complies with $\bm{p}$. By a similar argument, this also holds if $h'_c(v)=-1$. 
	
	Now suppose that $\eta'_c(u)=-1$ for some $c \in \{ a,b \}$. 
    This implies that $\eta_c(u)=0$ and $r_c(v) + 1 \le r_{\neg c}(u)< r_c(u) \le \fin$. 
    Assume again, for the sake of contradiction, that there is a solution $T^x_a,T^x_b \subseteq V_x$  that complies with $\bm{p}$. 
    Consider again the sets $\widehat{P}^{i}_c$ obtained in the process $\mathcal{\widehat{P}}(x,r_a,r_b,T^x_a,T^x_b)$, we again have $v \in  N_{G_x}(u) \cap \widehat{P}^{r_c(v)}_c$. 
    As $r_{\neg c}(u) - 1 \ge r_c(v)$, this implies that $\big| N_{G_x}(u) \cap \widehat{P}^{r_{\neg c}(u) - 1}_c \big| \ge 1 > \eta_c(v)$.
	This contradicts the viability of the process $\mathcal{\widehat{P}}(x,r_a,r_b,T^x_a,T^x_b)$.
	Therefore, if $\eta'_c(u)=-1$, then no solution complies with $\bm{p}$. 
	By a similar argument, this also holds if $\eta'_c(v)=-1$. 
\end{claimproof}
    
    Therefore, the answer of the algorithm is correct in these cases.
	Note that otherwise $\bm{p}'$ is a valid pattern since $\bm{p}$ is valid and there are only two ways $\bm{p}'$ could be invalid.
	One is if $g_c(x)=0$, $h_c(x)=0$, or $\eta_c(x)=0$ and $g'_c(x) \neq 0$, $h'_c(x) \neq 0$, or $\eta'_c(x) \neq 0$ for some $x \in \{u,v\}$, which is handled by the maximum in the definition and the previous claim, respectively.
	The other is if $r_{\neg c}(u)-1 = r_c(u)-1$ and $\eta'_c(u) \neq h'_c(u)$, or similarly with $v$.
	However, in this case $\eta'_c(u) \neq h'_c(u)$ and either $r_c(v)+1 \ge r_c(u)$ in which case both $h'_c(u)=h_c(u)$ and $\eta'_c(u) = \eta_c(u)$ or $r_c(v)+1 < r_c(u)$ which also implies that $r_c(v)+1 \le r_{\neg c}(u)$ and, therefore, $h'_c(u)=h_c(u)-1$ and $\eta'_c(u) = \eta_c(u)-1$, again preserving the equality.
	Hence the entry $\DP_y[\bm{p}']$ is well defined.
	
	To finish the proof we need the following two claims.
\begin{claim}\label{cla:tw_IntroEdge_from_child}
    If there is a solution $T^y_a,T^y_b \subseteq V_y=V_x$ that complies with $\bm{p}'$, then it also complies with $\bm{p}$.
\end{claim}

\begin{claimproof}
    Let $T^y_a,T^y_b \subseteq V_y=V_x$ be a solution that complies with $\bm{p}'$. 
    The condition that $r^{-1}_c(0)= (S_c \cup T_c^x) \cap \beta(x)$ for each $c \in \{ a,b \}$ is the same as $\beta(x)=\beta(y)$. As the neighborhood of each vertex in $V_x \setminus \beta(x)$ is the same in $G_x$ and $G_y$, the sets $\widehat{P}^{i}_c$ obtained in the process $\mathcal{\widehat{P}}(y,r_a,r_b,T^y_a,T^y_b)$ are the same as in the process $\mathcal{\widehat{P}}(x,r_a,r_b,T^y_a,T^y_b)$. Note that since $T^y_a,T^y_b$ complies with $\bm{p}'$ by assumption, the process $\mathcal{\widehat{P}}(y,r_a,r_b,T^y_a,T^y_b)$ is viable. Condition $\widehat{P}_a^{\fin}= \widehat{P}_a^{\fin+1} = \widehat{P}_b^{\fin}= \widehat{P}_b^{\fin+1}$ remains valid.
	Furthermore, for all $v' \in \beta(y)\setminus\{u,v\}$, all $i \in \{0, \ldots, \fin\}$, and all $c \in \{ a,b \}$ we have $N_{G_x}(v') \cap \widehat{P}^{i}_c = N_{G_y}(v') \cap \widehat{P}^{i}_c$. As $g_c(v')=g'_c(v')$ and $h_c(v')=h'_c(v')$ for these vertices, the viability conditions remain valid for these vertices.
	
	For each $c \in \{ a,b \}$, if $r_c(v) \le r_c(u)$ and $r_c(v) \neq \infty$, then we have $N_{G_x}(v) \cap \widehat{P}^{i}_c = N_{G_y}(v) \cap \widehat{P}^{i}_c$ for $i \in \{r_c(v)-2,r_c(v)-1\}\cap \{0,\ldots,\fin\}$, $g'_c(v) = g_c(v)$, and $h'_c(v) = h_c(v)$. 
	If $1 \le r_{\neg c}(v) < r_c(v) < \infty$, then we also have $N_{G_x}(v) \cap \widehat{P}^{i}_c = N_{G_y}(v) \cap \widehat{P}^{i}_c$ for $i=r_{\neg c}(v) -1$ and $\eta'_c(v)=\eta_c(v)$.
	Hence, the viability condition remains valid for $v$ in this case. 
	Similarly, if $r_c(v)=r_c(u)=\infty$, then $N_{G_x}(v) \cap \widehat{P}^{\fin}_c = N_{G_y}(v) \cap \widehat{P}^{\fin}_c$.
	\begin{itemize}
	\item If $r_c(v) = r_c(u)$, then the argument also holds for $u$ by symmetry, as $g'_c(u) = g_c(u)$, $h'_c(u) = h_c(u)$, and $\eta'_c(u)=\eta_c(u)$ in this case.
	\item If $r_c(v) + 1 = r_c(u)\le \fin$, then we have
    \[ \big|N_{G_x}(u) \cap \widehat{P}^{r_c(u)-1}_c\big| = 1+ \big|N_{G_y}(u) \cap \widehat{P}^{r_c(u)-1}_c\big| \ge 1+g'_c(u) = 1+\max(0,g_c(u)-1) \ge g_c(u),\] 
    \[\big|N_{G_x}(u) \cap \widehat{P}^{r_c(u)-2}_c\big| = \big|N_{G_y}(u) \cap \widehat{P}^{r_c(u)-2}_c\big| \le h'_c(u) = h_c(u),\]
    and if $1 \le r_{\neg c}(u) < r_c(u)$, then 
    \[\big|N_{G_x}(u) \cap \widehat{P}^{r_{\neg c}(u)-1}_c\big| = \big|N_{G_y}(u) \cap \widehat{P}^{r_{\neg c}(u)-1}_c\big| \le \eta'_c(u) = \eta_c(u).\]
	\item If $r_c(v) + 1 < r_c(u) \le \fin$, then we have 
	\[\big|N_{G_x}(u) \cap \widehat{P}^{r_c(u)-1}_c\big| = 1+ \big|N_{G_y}(u) \cap \widehat{P}^{r_c(u)-1}_c\big| \ge 1+g'_c(u) = 1+\max(0,g_c(u)-1) \ge g_c(u)\] 
	and 
	\[\big|N_{G_x}(u) \cap \widehat{P}^{r_c(u)-2}_c\big| = 1+\big|N_{G_y}(u) \cap \widehat{P}^{r_c(u)-2}_c\big| \le 1+h'_c(u) = 1+h_c(u)-1=h_c(u).\]
	If also $r_c(v) +1 \le r_{\neg c}(u) < r_c(u)$, then 
	\[\big|N_{G_x}(u) \cap \widehat{P}^{r_{\neg c}(u)-1}_c\big| = 1+\big|N_{G_y}(u) \cap \widehat{P}^{r_{\neg c}(u)-1}_c\big| \le 1+\eta'_c(u) = 1+\eta_c(u)-1=\eta_c(u),\]
	otherwise
	\[\big|N_{G_x}(u) \cap \widehat{P}^{r_{\neg c}(u)-1}_c\big| = \big|N_{G_y}(u) \cap \widehat{P}^{r_{\neg c}(u)-1}_c\big| \le \eta'_c(u) = \eta_c(u).\]
	\item If $r_c(v) \le \fin$ and $r_c(u) = \infty$, then we have 
	\[\big|N_{G_x}(u) \cap \widehat{P}^{\fin}_c\big| = 1+\big|N_{G_y}(u) \cap \widehat{P}^{\fin}_c\big| \le 1+h'_c(u) = 1+h_c(u)-1=h_c(u).\]
	Therefore, the viability condition also remains valid for $u$ in this case. 
	\end{itemize}
	The other cases follow by a symmetric argument. Therefore, the process is viable and $T^y_a,T^y_b$ also complies with $\bm{p}$.
\end{claimproof}

\begin{claim}\label{cla:tw_IntroEdge_to_child}
    If there is a solution $T^x_a,T^x_b \subseteq V_x=V_y$ that complies with $\bm{p}$, then it also complies with $\bm{p}'$.
\end{claim}

\begin{claimproof}
    Let $T^x_a,T^x_b \subseteq V_x=V_y$ be a solution that complies with $\bm{p}$.
    As observed above, the only difference between the conditions imposed are the viability conditions for~$u$ and~$v$.
	Let again $\widehat{P}^{i}_c$ be the sets obtained in the process $\mathcal{\widehat{P}}(y,r_a,r_b,T^x_a,T^x_b)$ (or the other one, they are the same). For each $c \in \{ a,b \}$ if $r_c(v) \le r_c(u)$ and $r_c(v) \neq \infty$, then we have $N_{G_y}(v) \cap \widehat{P}^{i}_c = N_{G_x}(v) \cap \widehat{P}^{i}_c$ for $i \in \{r_c(v)-2,r_c(v)-1\}\cap \{0,\ldots,\fin\}$, $g'_c(v) = g_c(v)$ and $h'_c(v) = h_c(v)$. 
	If $1 \le r_{\neg c}(v) < r_c(v) < \infty$, then we also have $N_{G_x}(v) \cap \widehat{P}^{i}_c = N_{G_y}(v) \cap \widehat{P}^{i}_c$ for $i=r_{\neg c}(v) -1$ and $\eta'_c(v)=\eta_c(v)$.
	Therefore, the viability condition remains valid for $v$ in this case. Similarly, if $r_c(v)=r_c(u)=\infty$, then $N_{G_y}(v) \cap \widehat{P}^{\fin}_c = N_{G_x}(v) \cap \widehat{P}^{\fin}_c$.
	\begin{itemize}
	\item If $r_c(v) = r_c(u)$, then the argument also holds for $u$ by symmetry, as $g'_c(u) = g_c(u)$, $h'_c(u) = h_c(u)$, and $\eta'_c(u) = \eta_c(u)$ in this case.
	\item If $r_c(v) + 1 = r_c(u)\le \fin$, then we have 
	\[\big|N_{G_y}(u) \cap \widehat{P}^{r_c(u)-1}_c\big| = \big|N_{G_x}(u) \cap \widehat{P}^{r_c(u)-1}_c\big|-1 \ge g_c(u)-1,\] 
	so 
	\[\big|N_{G_y}(u) \cap \widehat{P}^{r_c(u)-1}_c\big| \ge \max(0,g_c(u)-1) = g'_c(u),\] 
	\[\big|N_{G_y}(u) \cap \widehat{P}^{r_c(u)-2}_c\big| = \big|N_{G_x}(u) \cap \widehat{P}^{r_c(u)-2}_c\big| \le h_c(u) = h'_c(u),\]
	and if $1 \le r_{\neg c}(u) < r_c(u)$, then 
	\[\big|N_{G_y}(u) \cap \widehat{P}^{r_{\neg c}(u)-1}_c\big| = \big|N_{G_x}(u) \cap \widehat{P}^{r_{\neg c}(u)-1}_c\big| \le \eta_c(u) = \eta'_c(u),\]
	\item If $r_c(v) + 1 < r_c(u) \le \fin$, then we have 
	\[\big|N_{G_y}(u) \cap \widehat{P}^{r_c(u)-1}_c\big| = \big|N_{G_x}(u) \cap \widehat{P}^{r_c(u)-1}_c\big|-1 \ge g_c(u)-1,\] 
	hence 
	\[\big|N_{G_y}(u) \cap \widehat{P}^{r_c(u)-1}_c\big| \ge \max(0,g_c(u)-1) = g'_c(u)\] 
	and 
	\[\big|N_{G_y}(u) \cap \widehat{P}^{r_c(u)-2}_c\big| = \big|N_{G_x}(u) \cap \widehat{P}^{r_c(u)-2}_c\big|-1 \le h_c(u)-1=h'_c(u).\]
	If also $r_c(v) +1 \le r_{\neg c}(u) < r_c(u)$, then 
	\[\big|N_{G_y}(u) \cap \widehat{P}^{r_{\neg c}(u)-1}_c\big| = \big|N_{G_x}(u) \cap \widehat{P}^{r_{\neg c}(u)-1}_c\big|-1 \le \eta_c(u)-1 = \eta'_c(u),\]
	otherwise
	\[\big|N_{G_y}(u) \cap \widehat{P}^{r_{\neg c}(u)-1}_c\big| = \big|N_{G_x}(u) \cap \widehat{P}^{r_{\neg c}(u)-1}_c\big| \le \eta_c(u) = \eta'_c(u).\]
	\item If $r_c(v) \le \fin$ and $r_c(u) = \infty$, then we have 
	\[\big|N_{G_y}(u) \cap \widehat{P}^{\fin}_c\big| = \big|N_{G_x}(u) \cap \widehat{P}^{\fin}_c\big|-1 \le h_c(u)-1 = h'_c(u).\] Therefore, the viability condition also remains valid for $u$ in this case. 
	\end{itemize}
	The other cases follow a symmetric argument. Therefore, the process is viable and $T^x_a,T^x_b$ also complies with~$\bm{p}'$.
\end{claimproof}	
	
	In summary, let $T_a^x,T_b^x \subseteq V_x$ be a minimum size solution that complies with $\bm{p}$ and is of size $s$.
	If we had $\DP_y[\bm{p}'] < s$, then, since $\DP_y[\bm{p}']$ was computed correctly, there would be a corresponding solution $T^y_a,T^y_b \subseteq V_y$ that is consistent with $\bm{p}'$ and is strictly smaller than~$s$. But, by \Cref{cla:tw_IntroEdge_from_child}, this solution would comply with $\bm{p}$, contradicting the minimality of $T^x_a,T^x_b$. 
	Therefore, we have $\DP_y[\bm{p}'] \ge s$. 
	Furthermore, since, by \Cref{cla:tw_IntroEdge_to_child}, $T^x_a,T^x_b$ complies with $\bm{p}'$, we have $\DP_y[\bm{p}'] \le s$, that is, $\DP_y[\bm{p}'] = s$. Therefore, $\DP_x[\bm{p}]=s$ and it is computed correctly.
\end{proof}

\paragraph{Forget Node}\label{sec:forget_vert}
Let~$x$ be a node forgetting a vertex~$v$ with the child node~$y$ and let~$\bm{p}$ be a valid solution pattern for~$x$.
For $r_a, r_b \in [\fin]$ we define $\bm{q}_{r_a,r_b} = (\overline{r}_a,\overline{r}_b,\overline{g}_a,\overline{g}_b,\overline{h}_a,\overline{h}_b,\overline{\eta}_a,\overline{\eta}_b)$ as (we only list those with $r_a \leq r_b$, the others are symmetric)
\[
\mkern-10mu
\begin{array}{ll}
(0,0,0,0,0,0,0,0)                               & \text{if } r_a = r_b = 0,\\
(0,1,0,f_2(v),0,0,0,0)                          & \text{if } r_a = 0, r_b = 1, \\
(0,r_b,0,f_2(v),0,f_2(v)-1,0,0)                 & \text{if } r_a = 0, r_b > 1, \\
(r_a,r_b,f_1(v),f_1(v),0,0,0,0)                 & \text{if } r_a = r_b =1,\\
(r_a,r_b,f_1(v),f_2(v),0,f_1(v)-1,0,0)             & \text{if } r_a =1,  r_b=2,\\
(r_a,r_b,f_1(v),f_2(v),0,f_2(v)-1,0,f_1(v)-1)   & \text{if } r_a=1,  r_b>2,\\
(r_a,r_b,f_1(v),f_1(v),f_1(v)-1,f_1(v)-1,0,0)   & \text{if } r_a = r_b > 1,\\
(r_a,r_b,f_1(v),f_2(v),f_1(v)-1,f_1(v)-1,0,f_1(v)-1)   & \text{if } 1 < r_a = r_b-1,\\
(r_a,r_b,f_1(v),f_2(v),f_1(v)-1,f_2(v)-1,0,f_1(v)-1)   & \text{if } 1 < r_a < r_b-1,\\
\end{array}
\]
and $\bm{q}_{\infty,\infty}=(\infty,\infty,0,0,f_1(v)-1,f_1(v)-1,0,0)$.
Let
\[
\mathcal{Q}' =
\begin{cases}
\{ \bm{q}_{0,0}\}
& \text{if } v \in S_a \cap S_b\\
\big\{ \bm{q}_{0,r_b} \,\big|\, r_b \in [\fin]\big\}
& \text{if } v \in S_a \setminus S_b\\
\big\{ \bm{q}_{r_a,0} \,\big|\, r_a \in [\fin]\big\}
& \text{if } v \in S_b \setminus S_a\\
\big\{ \bm{q}_{r_a,r_b} \,\big|\, r_a,r_b \in [\fin] \big\}\cup
\{\bm{q}_{\infty,\infty}\}
& \text{if } v \notin S_a \cup S_b\\
\end{cases}
\]
and
let $\mathcal{Q}$ be obtained from $\mathcal{Q}'$ by removing all octuples containing a $-1$.
Note that $Q$ is always non-empty, as $\bm{q}_{0,0} \in Q$.
Set $\DP_x[\bm{p}] = \min_{\bm{q} \in \mathcal{Q}} \Big(\DP_y\big[\bm{p} \cup (v \mapsto \bm{q})\big]\Big)$,
where the pattern $\bm{p}'=\bm{p} \cup (v \mapsto \bm{q})$ is such that $\bm{p}'|_{\beta(x)} = \bm{p}|_{\beta(x)}$ and
the values for $v$ are given by $\bm{q}$.

\begin{lemma}\label{lem:twDP:forgetNode}
	Let~$x$ be a node forgetting a vertex~$v$ with the child node~$y$.
	If $\DP_y$ was computed correctly, and $\DP_x$ is computed using the algorithm for the forget node, then $\DP_x$ is also computed correctly.
\end{lemma}
\begin{proof}%
	Let~$\bm{p}$ be a valid solution pattern for~$x$.
	Let us first note that for each $\bm{q} = (\overline{r}_a,\overline{r}_b,\overline{g}_a,\overline{g}_b,\overline{h}_a,\overline{h}_b,\overline{\eta}_a,\overline{\eta}_b)$ in $\mathcal{Q}$ the pattern $\bm{p}'=\bm{p} \cup (v \mapsto \bm{q})$ is valid.
	Indeed, $\bm{p}$ is valid for all $v' \in \beta(x)$ and for $v$ we have the following.
	First,  if $r'_a(v) = \overline{r}_a = \infty$, then $\bm{q} = \bm{q}_{\infty,\infty}$, thus $r'_b(v)= \overline{r}_b = \infty$, and vice versa. 
	Second, if $v \in S_c$ for some $c \in \{ a,b \}$, then for each $\bm{q} = (\overline{r}_a,\overline{r}_b,\overline{g}_a,\overline{g}_b,\overline{h}_a,\overline{h}_b,\overline{\eta}_a,\overline{\eta}_b)$ in $\mathcal{Q}$ we have $r'_c(v) = \overline{r}_c = 0$.
	Third, for each $\bm{q} = (\overline{r}_a,\overline{r}_b,\overline{g}_a,\overline{g}_b,\overline{h}_a,\overline{h}_b,\overline{\eta}_a,\overline{\eta}_b)$ in $\mathcal{Q}$ and each $c \in \{ a,b \}$, if $r'_c(v) = \overline{r}_c  \in \{0,\infty\}$, then $g'_c(v) = \overline{g}_c=0$, whereas if $r'_c(v) = \overline{r}_c  \in \{0,1\}$, then $h'_c(v) = \overline{h}_c=0$.
	Fourth, if $r'_c(v) = \overline{r}_c \le r'_{\neg c}(v) = \overline{r}_{\neg c}$ or $r'_{\neg c}(v) = \overline{r}_{\neg c} \in \{0,\infty\}$, then $\eta'_c(v) = \overline{\eta}_c=0$, and if $r'_{\neg c}(v)-1 = \overline{r}_{\neg c}-1 = r'_{c}(v)-2 = \overline{r}_{c}-2$, then $\eta'_c(v) = \overline{\eta}_c = \overline(h)_c = h'_c(v)$. 
	Therefore, $\bm{p}'$ is indeed valid and $\DP_y[\bm{p}']$ is well defined.
	
	To finish the proof we need the following two claims.
\begin{claim}\label{cla:tw_Forget_to_child}
    If there is a solution $T^x_a,T^x_b \subseteq V_x=V_y$  that complies with $\bm{p}$, then it also complies with $\bm{p}'=\bm{p} \cup (v \mapsto \bm{q})$ for some $\bm{q} = (\overline{r}_a,\overline{r}_b,\overline{g}_a,\overline{g}_b,\overline{h}_a,\overline{h}_b,\overline{\eta}_a,\overline{\eta}_b)$ in $\mathcal{Q}$.
\end{claim}

\begin{claimproof}	
	Let $T^x_a,T^x_b \subseteq V_x=V_y$ be a solution that complies with $\bm{p}$.
	Consider the sets $\widehat{P}^{i}_c$ obtained in process $\mathcal{\widehat{P}}(x,r_a,r_b,T^x_a,T^x_b)$ and let $\widehat{P}^{-1}_c=\emptyset$. Let $\overline{r}_c$ be such that $v \in \left(\widehat{P}^{\overline{r}_c}_c\setminus \widehat{P}^{\overline{r}_c-1}_c\right)$ or $\infty$ if $v \notin \widehat{P}^{\fin}_c$ for $c \in \{a,b\}$. Note that we have $v \in S_c \cup T^x_c =\widehat{P}^0_c$ if and only if $\overline{r}_c =0$. Thus $\bm{q}_{\overline{r}_a,\overline{r}_b} \in \mathcal{Q}'$.
	
	We want to show that also $\bm{q}_{\overline{r}_a,\overline{r}_b} \in \mathcal{Q}$ and that $T^x_a,T^x_b$ complies with $\bm{p}'=\bm{p} \cup (v \mapsto \bm{q}_{\overline{r}_a,\overline{r}_b})$. Assume that $\overline{r}_a \le \overline{r}_b$.
	If $\overline{r}_a \in \{1, \ldots, \fin\}$, then, since $v \in \left(\widehat{P}^{\overline{r}_a}_a\setminus \widehat{P}^{\overline{r}_a-1}_a\right)$ and $v \notin \widehat{P}^{\overline{r}_a-1}_b$, we have $f_1(v)=\overline{g}_a \le \big|N_{G_x}(v) \cap \widehat{P}^{\overline{r}_a-1}_a\big|$. If $\overline{r}_a \in \{2, \ldots, \fin\}$, then, as $v \notin \widehat{P}^{\overline{r}_a-1}_a$, we have $f_1(v)-1=\overline{h}_a \ge \big|N_{G_x}(v) \cap \widehat{P}^{\overline{r}_a-2}_a\big|$. Similarly, if $\overline{r}_a = \infty$, then, as $v \notin \widehat{P}^{\fin+1}_a$, we have $f_1(v)-1=\overline{h}_a \ge \big|N_{G_x}(v) \cap \widehat{P}^{\fin}_a\big|$. If $\overline{r}_b=\overline{r}_a$, then we have the same bounds for $\overline{g}_b$ and $\overline{h}_b$.
	If $\overline{r}_b=\overline{r}_a+1$, then, as $v \in \left(\widehat{P}^{\overline{r}_b}_b\setminus \widehat{P}^{\overline{r}_b-1}_b\right)$ and $v \in \widehat{P}^{\overline{r}_b-1}_a$ we have $f_2(v)=\overline{g}_b \le \big|N_{G_x}(v) \cap \widehat{P}^{\overline{r}_b-1}_b\big|$ and, as $v \notin \widehat{P}^{\overline{r}_b-2}_a$, that $f_1(v)-1=\overline{h}_b \ge \big|N_{G_x}(v) \cap \widehat{P}^{\overline{r}_b-2}_b\big|$.
	If $\overline{r}_b\ge\overline{r}_a+2$, then, as $v \in \left(\widehat{P}^{\overline{r}_b}_b\setminus \widehat{P}^{\overline{r}_b-1}_b\right)$ and $v \in \widehat{P}^{\overline{r}_b-1}_a$ we have $f_2(v)=\overline{g}_b \le \big|N_{G_x}(v) \cap \widehat{P}^{\overline{r}_b-1}_b\big|$ and, as $v \in \widehat{P}^{\overline{r}_b-2}_a$, that $f_2(v)-1=\overline{h}_b \ge \big|N_{G_x}(v) \cap \widehat{P}^{\overline{r}_b-2}_b\big|$.
	Furthermore, in the last two cases, as $v \notin \widehat{P}^{\overline{r}_a}_b$, we have $f_1(v)-1=\overline{\eta}_b \ge \big|N_{G_x}(v) \cap \widehat{P}^{\overline{r}_a-1}_b\big|$.
	In particular, these inequalities imply that all coordinates of $\bm{q}_{\overline{r}_a,\overline{r}_b}$ are at least $0$ and, hence, $\bm{q}_{\overline{r}_a,\overline{r}_b} \in \mathcal{Q}$.
	
	To show that $T^x_a,T^x_b$ complies with $\bm{p}'=(r'_a,r'_b,g'_a,g'_b,h'_a,h'_b,\eta'_a,\eta'_b)$ note that $r'^{-1}_c(0)=(S_c \cup T_c^x) \cap \beta(y)$ due to the way we defined $\overline{r}_c$. 
	Consider the sets $\overline{P}^{i}_c$ obtained in process $\mathcal{\widehat{P}}(y,r'_a,r'_b,T^x_a,T^x_b)$. Again, due to the way we defined $\overline{r}_c$ we have $\overline{P}^{i}_c = \widehat{P}^{i}_c$ for every $i \in \{0, \ldots, \fin\}$ and each $c \in \{a,b\}$. Hence the condition $\widehat{P}_a^{\fin}= \widehat{P}_a^{\fin+1} = \widehat{P}_b^{\fin}= \widehat{P}_b^{\fin+1}$ remains satisfied. Similarly, the viability conditions are satisfied for each vertex $v' \in \beta(x)$ as $\bm{p}'|_{\beta(x)} = \bm{p}|_{\beta(x)}$. The viability conditions for the vertex $v$ are given by the inequalities of the previous paragraph. Hence $\mathcal{\widehat{P}}(y,r'_a,r'_b,T^x_a,T^x_b)$ is viable and $T^x_a,T^x_b$ complies with $\bm{p}'$.
\end{claimproof}

\begin{claim}\label{cla:tw_Forget_from_child}
    If there is a solution $T^y_a,T^y_b \subseteq V_x=V_y$ that complies with $\bm{p}'=\bm{p} \cup (v \mapsto \bm{q})$ for some $\bm{q} = (\overline{r}_a,\overline{r}_b,\overline{g}_a,\overline{g}_b,\overline{h}_a,\overline{h}_b,\overline{\eta}_a,\overline{\eta}_b)$ in $\mathcal{Q}$, then it also complies with~$\bm{p}$. 
\end{claim}

\begin{claimproof}
	Let $T^y_a,T^y_b \subseteq V_x=V_y$ be a solution that complies with $\bm{p}'=\bm{p} \cup (v \mapsto \bm{q})$ for some $\bm{q} = (\overline{r}_a,\overline{r}_b,\overline{g}_a,\overline{g}_b,\overline{h}_a,\overline{h}_b,\overline{\eta}_a,\overline{\eta}_b)$ in $\mathcal{Q}$. 
	Let $\bm{p}'=(r'_a,r'_b,g'_a,g'_b,h'_a,h'_b,\eta'_a,\eta'_b)$. 
	Since $r_c=r'_c|_{\beta(x)}$, $r'^{-1}_c(0)=(S_c \cup T_c^y) \cap \beta(y)$ implies $r^{-1}_c(0)=(S_c \cup T_c^y) \cap \beta(x)$. 
	Our goal is to show the viability of the process $\mathcal{\widehat{P}}(x,r_a,r_b,T^y_a,T^y_b)$ based on the viability of the process $\mathcal{\widehat{P}}(y,r'_a,r'_b,T^y_a,T^y_b)$. To this end let $\widehat{P}^{i}_c$ be the sets obtained in process $\mathcal{\widehat{P}}(y,r'_a,r'_b,T^y_a,T^y_b)$ and $\overline{P}^{i}_c$ be the sets obtained in process $\mathcal{\widehat{P}}(x,r_a,r_b,T^y_a,T^y_b)$.
	We claim that $\widehat{P}^{i}_c = \overline{P}^{i}_c$ for each $i \in \{0, \ldots, \fin\}$ and each $c \in \{a,b\}$. Once this is proven, the viability of $\mathcal{\widehat{P}}(x,r_a,r_b,T^y_a,T^y_b)$ follows, since $\bm{p}=\bm{p}'|_{\beta(x)}$.
	
	We prove the claim by induction on $i$. For $i=0$ we have $\widehat{P}^{0}_c = \overline{P}^{0}_c = T^y_c \cup (S_c \cap V_x)$, constituting the base case of the induction. For $i \ge 1$ note that for each $v' \in V_x \setminus \{v\}$ the conditions for $v'$ to be included in $\widehat{P}^{i}_c$ and in $\overline{P}^{i}_c$ are the same, as $\widehat{P}^{i-1}_c = \overline{P}^{i-1}_c$.
	Assume that $\overline{r}_a \le \overline{r}_b$.
	\begin{itemize}
		\item
		If $i < \overline{r}_a$, then $v \notin \left(\widehat{P}^{i}_a \cup \widehat{P}^{i}_b\right)$ and, as the process $\mathcal{\widehat{P}}(y,r'_a,r'_b,T^y_a,T^y_b)$ is viable, we have 
		\[\big|N_{G_x}(v) \cap \widehat{P}^{i-1}_a\big| \le \big|N_{G_x}(v) \cap \widehat{P}^{\overline{r}_a-2}_a\big| \le h'_a(v) = \overline{h}_a= f_1(v)-1.\] As $\widehat{P}^{i-1}_a=\overline{P}^{i-1}_a$ by induction hypothesis, this implies that $v \notin \overline{P}^{i}_a$ and $\widehat{P}^{i}_a=\overline{P}^{i}_a$. %
		If $\overline{r}_a < \overline{r}_b$, then 
		\[\big|N_{G_x}(v) \cap \widehat{P}^{i-1}_b\big| \le \big|N_{G_x}(v) \cap \widehat{P}^{\overline{r}_a-1}_b\big| \le \eta'_b(v) = \overline{\eta}_b= f_1(v)-1,\]
		which again implies $v \notin \overline{P}^{i}_b$ and $\widehat{P}^{i}_b=\overline{P}^{i}_b$.
		This is the case even for $i=\overline{r}_a < \overline{r}_b$.
		\item
		If $i = \overline{r}_a$, then $v \in \widehat{P}^{i}_a$,  and, as the process $\mathcal{\widehat{P}}(y,r'_a,r'_b,T^y_a,T^y_b)$ is viable, we have $\big|N_{G_x}(v) \cap \widehat{P}^{i-1}_a\big|  \ge  g'_a(v) = \overline{g}_a=f_1(v)$. As $\widehat{P}^{i-1}_a=\overline{P}^{i-1}_a$ by induction hypothesis, this implies that $v \in \overline{P}^{i}_a$ and $\widehat{P}^{i}_a=\overline{P}^{i}_a$. Similar argument works for $\overline{P}^{i}_b$ if $\overline{r}_b = \overline{r}_a$.
		\item
		If $\overline{r}_a < i < \overline{r}_b$, then $v \in \widehat{P}^{i-1}_a= \overline{P}^{i-1}_a$, $v \notin \widehat{P}^{i}_b$ and, as the process $\mathcal{\widehat{P}}(y,r'_a,r'_b,T^y_a,T^y_b)$ is viable, we have 
		\[\big|N_{G_x}(v) \cap \widehat{P}^{i-1}_b\big|  \le  \big|N_{G_x}(v) \cap \widehat{P}^{\overline{r}_b-2}_b\big|  \le h'_b(v) = \overline{h}_b = f_2(v)-1.\] 
		As $\widehat{P}^{i-1}_b=\overline{P}^{i-1}_b$ by induction hypothesis, this implies that $v \notin \overline{P}^{i}_b$ and $\widehat{P}^{i}_b=\overline{P}^{i}_b$.
		\item
		If $i= \overline{r}_b > \overline{r}_a$, then $v \in \widehat{P}^{i-1}_a= \overline{P}^{i-1}_a$, $v \in \widehat{P}^{i}_b$ and, as the process $\mathcal{\widehat{P}}(y,r'_a,r'_b,T^y_a,T^y_b)$ is viable, we have $\big|N_{G_x}(v) \cap \widehat{P}^{i-1}_b\big|  \ge   g'_b(v) = \overline{g}_b = f_2(v)$. As $\widehat{P}^{i-1}_b=\overline{P}^{i-1}_b$ by induction hypothesis, this implies that $v \in \overline{P}^{i}_b$ and $\widehat{P}^{i}_b=\overline{P}^{i}_b$.
		\item
		If $i> \overline{r}_b$, then $v \in \widehat{P}^{i-1}_a= \overline{P}^{i-1}_a$ and $v \in \widehat{P}^{i-1}_b= \overline{P}^{i-1}_b$, thus $v \in \widehat{P}^{i}_a$, $v \in \overline{P}^{i}_a$, $v \in \widehat{P}^{i}_b$, and $v \in \overline{P}^{i}_b$. Therefore $\widehat{P}^{i}_a=\overline{P}^{i}_a$ and $\widehat{P}^{i}_b=\overline{P}^{i}_b$.
	\end{itemize}
	The proof for the case $\overline{r}_a > \overline{r}_b$ follows by a symmetric argument.
	This finishes the proof of the induction step and the proof of the claim.
\end{claimproof}	
	
	To sum up, let $T^x_a,T^x_b \subseteq V_x$ be a minimum size solution which complies with $\bm{p}$ and let it be of size $s$.
	If we had $\DP_y[\bm{p}'] < s$ for $\bm{p}'=\bm{p} \cup (v \mapsto \bm{q})$ for some $\bm{q} \in Q$, then, since $\DP_y[\bm{p}']$ was computed correctly, there would be a corresponding solution $T^y_a,T^y_b \subseteq V_y$ which complies with $\bm{p}'$ and is of size strictly less than $s$. 
	But, by \Cref{cla:tw_Forget_from_child}, this solution would comply with $\bm{p}$, contradicting the minimality of $T^x_a,T^x_b$. 
	Hence we have $\DP_y[\bm{p}'] \ge s$ for every $\bm{p}'=\bm{p} \cup (v \mapsto \bm{q}),\bm{q} \in Q$. 
	Moreover, as $T^x_a,T^x_b$ complies with $\bm{p}'$ for $\bm{p}'=\bm{p} \cup (v \mapsto \bm{q})$ for suitable $\bm{q} \in Q$ by \Cref{cla:tw_Forget_to_child}, we have $\DP_y[\bm{p}'] \le s$ for this $\bm{p}'$, that is, $\DP_y[\bm{p}'] = s$. 
	Hence $\DP_x[\bm{p}]=s$ and it is computed correctly.
\end{proof}

\paragraph{Join Node}\label{sec:join}
Let~$x$ be a join node with children~$y$ and~$z$ and let~$\bm{p}=(r_a,r_b,g_a,g_b,h_a,h_b,\eta_a,\eta_b)$ be a valid solution pattern for~$x$.
Set \[\DP_x[\bm{p}] = \inf (\DP_y[\bm{p}_y]+\DP_z[\bm{p}_z])- \sum_{c\in \{a,b\}}\sum_{v\in \beta(x)} \cost(c,\bm{p},v),\]
where $\cost$ is as in \Cref{eq:cost} and the infimum is taken over all pairs of valid patterns ${\bm{p}_y=(r^y_a,r^y_b,g^y_a,g^y_b,h^y_a,h^y_b,\eta^y_a,\eta^y_b)}$ and $\bm{p}_z=(r^z_a,r^z_b,g^z_a,g^z_b,h^z_a,h^z_b,\eta^z_a,\eta^z_b)$ such that for every ${v \in \beta(x)}$ and for every $c \in \{a,b\}$ we have $r^y_c(v)=r^z_c(v)=r_c(v)$, $g^y_c(v)+g^z_c(v)=g_c(v)$, $h^y_c(v)+h^z_c(v)=h_c(v)$, and $\eta^y_c(v)+\eta^z_c(v)=\eta_c(v)$ .

\begin{lemma}\label{lem:twDP:joinNode}
	Let~$x$ be a join node with children~$y$ and~$z$.
	If $\DP_y$ and $\DP_z$ were computed correctly, and $\DP_x$ is computed using the algorithm for the join node, then $\DP_x$ is also computed correctly.
\end{lemma}
\begin{proof}
	Note first that, since each edge is introduced only once, for each vertex $v \in \beta(x)=\beta(y)=\beta(z)$ we have $N_{G_x}(v)=N_{G_y}(v) \cup N_{G_z}(v)$ and $N_{G_y}(v) \cap N_{G_z}(v) = \emptyset$. Note also that $V_y \cap V_z=\beta(x)$ and for each vertex $v \in \alpha(y)$ we have $N_{G_x}(v)=N_{G_y}(v) \subseteq V_y$ and similarly for $v \in \alpha(z)$.
	
	Let $\bm{p}=(r_a,r_b,g_a,g_b,h_a,h_b,\eta_a,\eta_b)$ be a valid pattern for~$x$. We need the following two claims. 

\begin{claim}\label{cla:tw_Join_from_child}
	Let $\bm{p}_y=(r^y_a,r^y_b,g^y_a,g^y_b,h^y_a,h^y_b,\eta^y_a,\eta^y_b)$ be a valid pattern for~$y$, and $\bm{p}_z=(r^z_a,r^z_b,g^z_a,g^z_b,h^z_a,h^z_b,\eta^z_a,\eta^z_b)$ be a valid pattern for~$z$ such that for every $v \in \beta(x)$ and every $c \in \{a,b\}$ we have $r^y_c(v)=r^z_c(v)=r_c(v)$, $g^y_c(v)+g^z_c(v)=g_c(v)$, $h^y_c(v)+h^z_c(v)=h_c(v)$, and $\eta^y_c(v)+\eta^z_c(v)=\eta_c(v)$.
	If $T^y_a,T^y_b \subseteq V_y$ is a solution which complies with $\bm{p}_y$ and $T^z_a,T^z_b\subseteq V_z$ is a solution which complies with $\bm{p}_z$, then 
	$T^x_a=T^y_a \cup T^z_a$ and $T^x_b=T^y_b \cup T^z_b$ form a solution which complies with $\bm{p}$ and has size $|T^y_a|+|T^y_b|+|T^z_a|+|T^z_b|-\sum_{c\in \{a,b\}}\sum_{v\in \beta(x)} \cost(c,\bm{p},v)$.
\end{claim}

\begin{claimproof}
    Let $\bm{p}_y, \bm{p}_z, T^y_a,T^y_b, T^z_a$, and $T^z_b$ satisfy the assumptions of the claim.
	Since $T^y_a,T^y_b$ complies with $\bm{p}_y$ we have that $v \in \beta(y)$ is in $T^y_c$ for $c \in \{a,b\}$ if and only if $r^y_c(v)=0$ and $v \notin S_c$. However, this is exactly when $\cost(c,\bm{p}_y,v)$ is $1$. Similarly for $T^z_a,T^z_b$. Since $r_c^y=r_c^z=r_c$, it follows that $T^y_c\cap \beta(x)=T^z_c\cap \beta(x)=T^x_c\cap \beta(x)$, $|T^y_c \cap T^z_c|=\sum_{v\in \beta(x)} \cost(c,\bm{p},v)$, and $|T^x_c|=|T^y_c \cup T^z_c|=|T^y_c|+|T^z_c|-\sum_{v\in \beta(x)} \cost(c,\bm{p},v)$, giving the size bound.
	
	We want to show the viability of the process $\mathcal{\widehat{P}}(x,r_a,r_b,T^x_a,T^x_b)$.
	To this end let $\widehat{P}^{i}_{(x),c}$ be the sets obtained in process $\mathcal{\widehat{P}}(x,r_a,r_b,T^x_a,T^x_b)$,
	$\widehat{P}^{i}_{(y),c}$ be the sets obtained in process $\mathcal{\widehat{P}}(y,r_a,r_b,T^y_a,T^y_b)$, and $\widehat{P}^{i}_{(z),c}$ be the sets obtained in process $\mathcal{\widehat{P}}(z,r_a,r_b,T^z_a,T^z_b)$.
	We claim that  $\widehat{P}^{i}_{(x),c} =\widehat{P}^{i}_{(y),c} \cup \widehat{P}^{i}_{(z),c}$ and $\widehat{P}^{i}_{(x),c} \cap \beta(x)= \widehat{P}^{i}_{(y),c} \cap \beta(x)=\widehat{P}^{i}_{(z),c} \cap \beta(x)$ for each ${i \in \{0, \ldots, \fin+1\}}$ and each $c \in \{a,b\}$.
	
	We prove the claim by induction on $i$. For $i=0$ we have $\widehat{P}^{0}_{(y),c}= T^y_c \cup (S_c \cap V_y)$, $\widehat{P}^{0}_{(z),c}= T^z_c \cup (S_c \cap V_z)$, and  
	\begin{align*}
	\widehat{P}^{0}_{(x),c}&= T^x_c \cup (S_c \cap V_x)\\
	&=T^x_c \cup \big(S_c \cap (V_y \cup V_z)\big)\\
	&= (T^x_c \cap V_y) \cup (T^x_c \cap V_z) \cup (S_c \cap V_y) \cup (S_c \cap V_z) \\
	&= T^y_c \cup (S_c \cap V_y) \cup T^z_c \cup (S_c \cap V_z)\\
	&=\widehat{P}^{0}_{(y),c} \cup \widehat{P}^{0}_{(z),c}.
	\end{align*}
	Moreover, as we have shown, 
	\[\widehat{P}^{0}_{(x),c} \cap \beta(x)= \widehat{P}^{0}_{(y),c} \cap \beta(x)=\widehat{P}^{0}_{(z),c} \cap \beta(x)=r^{-1}_c(0) \setminus S_c.\]
	This constitutes the base case of the induction.
	
	Assume that $i\ge 1$ and that the claim holds for all lesser $i$.
	For $v \in \beta(x)$ we have $v \in \widehat{P}^{i}_{(y),c}$ if and only if $r_c(v) \le i$, and, as $r_c=r^y_c=r^z_c$, the condition for inclusion in $\widehat{P}^{i}_{(z),c}$ and $\widehat{P}^{i}_{(x),c}$ is exactly the same. Hence $\widehat{P}^{i}_{(x),c} \cap \beta(x)= \widehat{P}^{i}_{(y),c} \cap \beta(x)=\widehat{P}^{i}_{(z),c} \cap \beta(x)$.
	For $v \in \alpha(y)$ we have $v \in \widehat{P}^{i}_{(y),c}$ if 
	\begin{itemize}
	\item either $v \in \widehat{P}^{i-1}_{(y),c}$, 
	\item or $v \notin (\widehat{P}_{(y),a}^{i-1} \cup \widehat{P}_{(y),b}^{i-1})$ and $\big|N_{G_y}(v) \cap \widehat{P}_{(y),c}^{i-1}\big| \ge f_1(v)$, 
	\item or $v \in \widehat{P}_{(y),\neg c}^{i-1}$ and $\big|N_{G_y}(v) \cap \widehat{P}_{(y),c}^{i-1}\big| \ge f_2(v)$.
	\end{itemize}
	However, since $N_{G_x}(v)=N_{G_y}(v) \subseteq V_y$, $\widehat{P}^{i-1}_{(x),c} =\widehat{P}^{i-1}_{(y),c} \cup \widehat{P}^{i-1}_{(z),c}$ and, thus, $\widehat{P}^{i-1}_{(y),c} = V_y \cap \widehat{P}^{i-1}_{(x),c}$, we have $\big|N_{G_y}(v) \cap \widehat{P}_{(y),c}^{i-1}\big|=\big|N_{G_x}(v) \cap \widehat{P}_{(x),c}^{i-1}\big|$. 
	Therefore, $v \in \widehat{P}^{i}_{(y),c}$ if 
	\begin{itemize}
	\item either $v \in \widehat{P}^{i-1}_{(x),c}$, 
	\item or $v \notin (\widehat{P}_{(x),a}^{i-1} \cup \widehat{P}_{(x),b}^{i-1})$ and $\big|N_{G_x}(v) \cap \widehat{P}_{(x),c}^{i-1}\big| \ge f_1(v)$,
	\item or 	$v \in \widehat{P}_{(x),\neg c}^{i-1}$ and $\big|N_{G_x}(v) \cap \widehat{P}_{(x),c}^{i-1}\big| \ge f_2(v)$. 
	\end{itemize}
	However, this is exactly if $v \in \widehat{P}^{i}_{(x),c}$.
	A completely analogous argument shows that for $v \in \alpha(z)$ we have $v \in \widehat{P}^{i}_{(z),c}$ if and only if $v \in \widehat{P}^{i}_{(x),c}$.
	This finishes the proof of the induction step and the proof of the correspondence of the processes.
	
	Since the processes $\mathcal{\widehat{P}}(y,r_a,r_b,T^y_a,T^y_b)$ and $\mathcal{\widehat{P}}(z,r_a,r_b,T^z_a,T^z_b)$ are viable, we have $\widehat{P}_{(y),a}^{\fin}= \widehat{P}_{(y),a}^{\fin+1} = \widehat{P}_{(y),b}^{\fin}= \widehat{P}_{(y),b}^{\fin+1}$ and
	$\widehat{P}_{(z),a}^{\fin}= \widehat{P}_{(z),a}^{\fin+1} = \widehat{P}_{(z),b}^{\fin}= \widehat{P}_{(z),b}^{\fin+1}$.
	As $\widehat{P}_{(x),c}^{\fin} = \widehat{P}_{(y),c}^{\fin} \cup \widehat{P}_{(z),c}^{\fin}$ and $\widehat{P}_{(x),c}^{\fin+1} = \widehat{P}_{(y),c}^{\fin+1} \cup \widehat{P}_{(z),c}^{\fin+1}$, it follows that $\widehat{P}_{(x),a}^{\fin}= \widehat{P}_{(x),a}^{\fin+1} = \widehat{P}_{(x),b}^{\fin}= \widehat{P}_{(x),b}^{\fin+1}$.
	
	For each $c \in \{ a,b \}$  and every $v \in \beta(x)$ with $1 \le r_c(v) < \infty$ we have 
	\begin{align*}
	\big| N_{G_x}(v) \cap \widehat{P}_{(x),c}^{r_c(v) - 1} \big| &= \big| (N_{G_y}(v) \cup N_{G_y}(v)) \cap \widehat{P}_{(x),c}^{r_c(v) - 1} \big|\\
	&= \big| N_{G_y}(v) \cap \widehat{P}_{(x),c}^{r_c(v) - 1} \big|+\big| N_{G_z}(v) \cap \widehat{P}_{(x),c}^{r_c(v) - 1} \big| \\
	&= \big| N_{G_y}(v) \cap \widehat{P}_{(y),c}^{r_c(v) - 1} \big|+\big| N_{G_z}(v) \cap \widehat{P}_{(z),c}^{r_c(v) - 1} \big| \\
	&\ge g^y_c(v)+g^z_c(v)\\
	&= g_c(v).
	\end{align*}
	If $r_c(v) \ge 2$, we also have 
	\begin{align*}
	\big| N_{G_x}(v) \cap \widehat{P}_{(x),c}^{r_c(v) - 2} \big|&= \big| N_{G_y}(v) \cap \widehat{P}_{(y),c}^{r_c(v) - 2} \big|+\big| N_{G_z}(v) \cap \widehat{P}_{(z),c}^{r_c(v) - 2} \big| \\
	&\le h^y_c(v)+h^z_c(v)\\
	&= h_c(v).
	\end{align*}
	If $1 \le r_{\neg c}(v) <r_c(v) < \infty$, then 
	\begin{align*}
	\big| N_{G_x}(v) \cap \widehat{P}_{(x),c}^{r_{\neg c}(v) - 1} \big|&= \big| N_{G_y}(v) \cap \widehat{P}_{(y),c}^{r_{\neg c}(v) - 1} \big|+\big| N_{G_z}(v) \cap \widehat{P}_{(z),c}^{r_{\neg c}(v) - 1} \big| \\
	&\le \eta^y_c(v)+\eta^z_c(v)\\
	&= \eta_c(v).
	\end{align*}
	Finally, for each $c \in \{ a,b \}$ and every $v \in \beta(x)$ with $r_c(v) = \infty$ we have 
	\begin{align*}
	\big| N_{G_x}(v) \cap \widehat{P}_{(x),c}^{\fin} \big|&= \big| N_{G_y}(v) \cap \widehat{P}_{(y),c}^{\fin} \big|+\big| N_{G_z}(v) \cap \widehat{P}_{(z),c}^{\fin} \big|\\
	&\le h^y_c(v)+h^z_c(v)\\
	&= h_c(v).
	\end{align*}
	Hence the process $\mathcal{\widehat{P}}(x,r_a,r_b,T^x_a,T^x_b)$ is viable for $\bm{p}$ and $T^x_a,T^x_b$ is a solution which complies with $\bm{p}$.
\end{claimproof}

\begin{claim}\label{cla:tw_Join_to_child}
    If $T^x_a,T^x_b \subseteq V_x$ is a solution which complies with $\bm{p}$, then there is a valid pattern $\bm{p}_y=(r_a,r_b,g^y_a,g^y_b,h^y_a,h^y_b,\eta^y_a,\eta^y_b)$ for~$y$, a valid pattern $\bm{p}_z=(r_a,r_b,g^z_a,g^z_b,h^z_a,h^z_b,\eta^z_a,\eta^z_b)$ for~$z$, a solution $T^y_a,T^y_b \subseteq V_y$ which complies with $\bm{p}_y$, and a solution $T^z_a,T^z_b\subseteq V_z$ which complies with $\bm{p}_z$ such that for every $v \in \beta(x)$ and every $c \in \{a,b\}$ we have $g^y_c(v)+g^z_c(v)=g_c(v)$, $h^y_c(v)+h^z_c(v)=h_c(v)$, $\eta^y_c(v)+\eta^z_c(v)=\eta_c(v)$ and $|T^y_a|+|T^y_b|+|T^z_a|+|T^z_b|=|T^x_a|+|T^x_b|+\sum_{c\in \{a,b\}}\sum_{v\in \beta(x)} \cost(c,\bm{p},v)$.
\end{claim}

\begin{claimproof}	
	Let $T^x_a,T^x_b \subseteq V_x$ be a solution which complies with $\bm{p}$.
	Let $T^y_c= T^x_c \cap V_y$ and $T^z_c= T^x_c \cap V_z$. Since $T^x_a,T^x_b$ complies with $\bm{p}$ we have that $v \in \beta(x)$ is in $T^x_c$ if and only if $r_c(v)=0$ and $v \notin S_c$ for each $c \in \{a,b\}$. However, this is exactly when $\cost(c,\bm{p},v)$ is $1$. Furthermore, the same holds for $T^y_c$ and $T^z_c$, as $\beta(x) = V_y \cap V_z$.
	It follows that $T^y_c\cap \beta(x)=T^z_c\cap \beta(x)=T^x_c\cap \beta(x)$, $|T^y_c \cap T^z_c|=\sum_{v\in \beta(x)} \cost(c,\bm{p},v)$, and $|T^y_c|+|T^z_c| = |T^y_c \cup T^z_c| + |T^y_c \cap T^z_c| = |T^x_c| + \sum_{v\in \beta(x)} \cost(c,\bm{p},v)$, giving the size bound.
	
	Before we define $\bm{p}_y$ and $\bm{p}_z$ we consider the modified processes corresponding to the solutions.
	Let $\widehat{P}^{i}_{(x),c}$ be the sets obtained in process $\mathcal{\widehat{P}}(x,r_a,r_b,T^x_a,T^x_b)$,
	$\widehat{P}^{i}_{(y),c}$ be the sets obtained in process $\mathcal{\widehat{P}}(y,r_a,r_b,T^y_a,T^y_b)$, and $\widehat{P}^{i}_{(z),c}$ be the sets obtained in process $\mathcal{\widehat{P}}(z,r_a,r_b,T^z_a,T^z_b)$.
	We claim that $\widehat{P}^{i}_{(y),c} =\widehat{P}^{i}_{(x),c} \cap V_y$ for each $i \in \{0, \ldots, \fin+1\}$ and each $c \in \{a,b\}$.
	
	We prove the claim by induction on $i$. For $i=0$ we have $\widehat{P}^{0}_{(x),c}= T^x_c \cup (S_c \cap V_x)$ while  $\widehat{P}^{0}_{(y),c}= T^y_c \cup (S_c \cap V_y) = \widehat{P}^{0}_{(x),c} \cap V_y$.
	This constitutes the base case of the induction.
	
	Assume that $i\ge 1$ and that the claim holds for all lesser $i$.
	For $v \in \beta(x)$ we have $v \in \widehat{P}^{i}_{(x),c}$ if and only if $r_c(v) \le i$, and the condition for inclusion in $\widehat{P}^{i}_{(y),c}$ is exactly the same.
	For $v \in \alpha(y)$ we have $v \in \widehat{P}^{i}_{(x),c}$ if 
	\begin{itemize}
	\item either $v \in \widehat{P}^{i-1}_{(x),c}$, 
	\item or $v \notin (\widehat{P}_{(x),a}^{i-1} \cup \widehat{P}_{(x),b}^{i-1})$ and $\big|N_{G_x}(v) \cap \widehat{P}_{(x),c}^{i-1}\big| \ge f_1(v)$, 
	\item or $v \in \widehat{P}_{(x),\neg c}^{i-1}$ and $\big|N_{G_x}(v) \cap \widehat{P}_{(x),c}^{i-1}\big| \ge f_2(v)$.
	\end{itemize}
	However, since $N_{G_x}(v)=N_{G_y}(v) \subseteq V_y$, $\widehat{P}^{i-1}_{(y),c} =\widehat{P}^{i-1}_{(x),c} \cap V_y$, we have $\big|N_{G_y}(v) \cap \widehat{P}_{(y),c}^{i-1}\big|=\big|N_{G_x}(v) \cap \widehat{P}_{(x),c}^{i-1}\big|$. 
	Therefore, $v \in \widehat{P}^{i}_{(x),c}$ if 
	\begin{itemize}
	\item either $v \in \widehat{P}^{i-1}_{(y),c}$, 
	\item or $v \notin (\widehat{P}_{(y),a}^{i-1} \cup \widehat{P}_{(y),b}^{i-1})$ and $\big|N_{G_y}(v) \cap \widehat{P}_{(y),c}^{i-1}\big| \ge f_1(v)$, 
	\item or $v \in \widehat{P}_{(y),\neg c}^{i-1}$ and $\big|N_{G_y}(v) \cap \widehat{P}_{(y),c}^{i-1}\big| \ge f_2(v)$. 
	\end{itemize}
	However, this is exactly if $v \in \widehat{P}^{i}_{(y),c}$.
	This finishes the proof of the induction step and the proof of the correspondence of the processes.
	
	A completely analogous argument shows, that $\widehat{P}^{i}_{(z),c} =\widehat{P}^{i}_{(x),c} \cap V_z$ for each ${i \in \{0, \ldots, \fin+1\}}$ and each $c \in \{a,b\}$. Since $\widehat{P}_{(x),a}^{\fin}= \widehat{P}_{(x),a}^{\fin+1} = \widehat{P}_{(x),b}^{\fin}= \widehat{P}_{(x),b}^{\fin+1}$, it follows that $\widehat{P}_{(y),a}^{\fin}= \widehat{P}_{(y),a}^{\fin+1} = \widehat{P}_{(y),b}^{\fin}= \widehat{P}_{(y),b}^{\fin+1}$ and similarly for $z$.
	
	Now we define the rest of $\bm{p}_y$ and $\bm{p}_z$. Let $c \in \{a,b\}$ and $v \in \beta(x)$.
	If $r_c(v) \in \{0,\infty\}$ then let $g^y_c(v) = g^z_c(v)=0$ (note that $g_c(v)=0$ in this case by the validity of $\bm{p}$).
	Otherwise, we let $g^y_c(v) = \big| N_{G_y}(v) \cap \widehat{P}_{(y),c}^{r_c(v) - 1} \big|$ and $g^z_c(v) =g_c(v)-g^y_c(v)$.
	Note that in this case, by viability of $\mathcal{\widehat{P}}(x,r_a,r_b,T^x_a,T^x_b)$ we have $g_c(v) \le \big| N_{G_x}(v) \cap \widehat{P}_{(x),c}^{r_c(v) - 1} \big|$ and as $N_{G_x}(v)=N_{G_y}(v) \cup N_{G_z}(v)$ and $N_{G_y}(v) \cap N_{G_z}(v) = \emptyset$ we have 
	\[\big| N_{G_z}(v) \cap \widehat{P}_{(z),c}^{r_c(v) - 1} \big|=\big| N_{G_x}(v) \cap \widehat{P}_{(x),c}^{r_c(v) - 1} \big|-\big| N_{G_y}(v) \cap \widehat{P}_{(y),c}^{r_c(v) - 1} \big| \ge g_c(v)-g^y_c(v)=g^z_c(v).\]
	
	If $r_c(v) \in \{0,1\}$ then let $h^y_c(v) = h^z_c(v)=0$ (note that $h_c(v)=0$ in this case).
	If $r_c(v) = \infty$, then we let $h^y_c(v) = \big| N_{G_y}(v) \cap \widehat{P}_{(y),c}^{\fin} \big|$ and $h^z_c(v) =h_c(v)-h^y_c(v)$.
	In this case $\big| N_{G_x}(v) \cap \widehat{P}_{(x),c}^{\fin} \big| \le h_c(v)$ and, thus, 
	\[\big| N_{G_z}(v) \cap \widehat{P}_{(z),c}^{\fin} \big|=\big| N_{G_x}(v) \cap \widehat{P}_{(x),c}^{\fin} \big|-\big| N_{G_y}(v) \cap \widehat{P}_{(y),c}^{\fin} \big| \le h_c(v)-h^y_c(v)=h^z_c(v).\]
	Otherwise let $h^y_c(v) = \big| N_{G_y}(v) \cap \widehat{P}_{(y),c}^{r_c(v) - 2} \big|$ and $h^z_c(v) =h_c(v)-h^y_c(v)$.
	In this case $\big| N_{G_x}(v) \cap \widehat{P}_{(x),c}^{r_c(v) - 2} \big| \le h_c(v)$ and 
	\[\big| N_{G_z}(v) \cap \widehat{P}_{(z),c}^{r_c(v) - 2} \big|=\big| N_{G_x}(v) \cap \widehat{P}_{(x),c}^{r_c(v) - 2} \big|-\big| N_{G_y}(v) \cap \widehat{P}_{(y),c}^{r_c(v) - 2} \big| \le h_c(v)-h^y_c(v)=h^z_c(v).\]
	If $r_c(v) \le r_{\neg c}(v)$ or $r_{\neg c}(v) \in \{0,\infty\}$, then let $\eta^y_c(v)=\eta^z_c(v)=0$ (note that $\eta_c(v)=0$ in this case).
	Otherwise, i.e., if $1 \le r_{\neg c}(v) < r_c(v) <\infty$, then let $\eta^y_c(v) = \big| N_{G_y}(v) \cap \widehat{P}_{(y),c}^{r_{\neg c}(v) - 1} \big|$ and $\eta^z_c(v) =\eta_c(v)-\eta^y_c(v)$.
	In this case $\big| N_{G_x}(v) \cap \widehat{P}_{(x),c}^{r_{\neg c}(v) - 1} \big| \le \eta_c(v)$ and 
	\[\big| N_{G_z}(v) \cap \widehat{P}_{(z),c}^{r_{\neg c}(v) - 1} \big|=\big| N_{G_x}(v) \cap \widehat{P}_{(x),c}^{r_{\neg c}(v) - 1} \big|-\big| N_{G_y}(v) \cap \widehat{P}_{(y),c}^{r_{\neg c}(v) - 1} \big| \le \eta_c(v)-\eta^y_c(v)=\eta^z_c(v).\]
	
	It follows directly from the definition and notes therein that $\bm{p}_y$ and $\bm{p}_z$ are valid and that the processes $\mathcal{\widehat{P}}(y,r_a,r_b,T^y_a,T^y_b)$ and $\mathcal{\widehat{P}}(z,r_a,r_b,T^z_a,T^z_b)$ are viable for them.
	Thus, $T^y_a,T^y_b$ complies with $\bm{p}_y$ and $T^z_a,T^z_b$ complies with $\bm{p}_z$.
\end{claimproof}
	
	To sum up, let $T^x_a,T^x_b \subseteq V_x$ be a minimum size solution which complies with $\bm{p}$ and let it be of size $s$.
	On one hand, if we had $\DP_y[\bm{p}_y] = s_y$ and $\DP_z[\bm{p}_z] = s_z$ for some pair of valid patterns $\bm{p}_y=(r^y_a,r^y_b,g^y_a,g^y_b,h^y_a,h^y_b,\eta^y_a,\eta^y_b)$ and $\bm{p}_z=(r^z_a,r^z_b,g^z_a,g^z_b,h^z_a,h^z_b,\eta^z_a,\eta^z_b)$ such that $s_y+s_z - \sum_{c\in \{a,b\}}\sum_{v\in \beta(x)} \cost(c,\bm{p},v) < s$ and for every $v \in \beta(x)$ and every $c \in \{a,b\}$ we have $r^y_c(v)=r^z_c(v)=r_c(v)$, $g^y_c(v)+g^z_c(v)=g_c(v)$, $h^y_c(v)+h^z_c(v)=h_c(v)$, and $\eta^y_c(v)+\eta^z_c(v)=\eta_c(v)$ then, since $\DP_y[\bm{p}_y]$ and $\DP_y[\bm{p}_z]$ were computed correctly, there would be  corresponding solutions $T^y_a,T^y_b \subseteq V_y$ which complies with $\bm{p}_y$ and is of size $s_y$ and $T^z_a,T^z_b \subseteq V_z$ which complies with $\bm{p}_z$ and is of size $s_z$. But then, by \Cref{cla:tw_Join_from_child}, there would be a solution which complies with $\bm{p}$ and is of size $s_y+s_z - \sum_{c\in \{a,b\}}\sum_{v\in \beta(x)} \cost(c,\bm{p},v) < s$, contradicting the minimality of $T^x_a,T^x_b$. Hence $\DP_x[\bm{p}] \ge s$.
	
	On the other hand, by \Cref{cla:tw_Join_to_child}, there is a pair of valid patterns $\bm{p}_y=(r^y_a,r^y_b,g^y_a,g^y_b,h^y_a,h^y_b,\eta^y_a,\eta^y_b)$ and $\bm{p}_z=(r^z_a,r^z_b,g^z_a,g^z_b,h^z_a,h^z_b,\eta^z_a,\eta^z_b)$ such that for every $v \in \beta(x)$ and every $c \in \{a,b\}$ we have $r^y_c(v)=r^z_c(v)=r_c(v)$, $g^y_c(v)+g^z_c(v)=g_c(v)$, $h^y_c(v)+h^z_c(v)=h_c(v)$, and $\eta^y_c(v)+\eta^z_c(v)=\eta_c(v)$, a solution $T^y_a,T^y_b \subseteq V_y$ of size $s_y$ which complies with $\bm{p}_y$, and a solution $T^z_a,T^z_b \subseteq V_z$ of size $s_z$ which complies with $\bm{p}_z$ such that $s_y+s_z - \sum_{c\in \{a,b\}}\sum_{v\in \beta(x)} \cost(c,\bm{p},v) = s$.
	Thus $\DP_y[\bm{p}_y] \le s_y$ and $\DP_z[\bm{p}_z] \le s_z$. 
	Therefore 
	\begin{align*}
	\DP_x[\bm{p}] &\le \DP_y[\bm{p}_y]+\DP_z[\bm{p}_z]- \sum_{c\in \{a,b\}}\sum_{v\in \beta(x)} \cost(c,\bm{p},v)\\
	&\le s_y+s_z - \sum_{c\in \{a,b\}}\sum_{v\in \beta(x)} \cost(c,\bm{p},v)\\
	&= s.
	\end{align*}
	Hence $\DP_x[\bm{p}]=s$ and it is computed correctly.
\end{proof}

\begin{proof}[Proof of \Cref{thm:twoTSSIsFPTWrtRoundsFmaxTw}]
	As already mentioned, we use bottom-up dynamic programming along a nice tree decomposition of width $\w'=\Oh{\w}$ which has $\w^{\Oh{1}} \cdot n$ nodes and can be found in $2^{\Oh{\w}}\cdot n$ time~\cite{CyganFKLMPPS15}.
	For each node there are at most $(\fin+2)^{2(\w'+1)} \cdot (f_{\max}+1) ^ {6(\w'+1)}$ solution patterns and hence the table has at most $(\fin+2)^{2(\w'+1)} \cdot (f_{\max}+1) ^ {6(\w'+1)}$ entries for each node.
	By the described algorithm, each entry can be computed by traversing all the entries of the child or the children.
	Hence, it can be done in $O\Big((\fin+2)^{4(\w'+1)} \cdot (f_{\max}+1) ^ {12(\w'+1)}\Big)$ time. Since the decomposition has $\w^{\Oh{1}} \cdot n$ nodes, the total running time is \mbox{$(\fin \cdot f_{\max}+1) ^\Oh{\w} \cdot n$}.
	
	The correctness of the algorithm follows by a bottom-up induction from the correctness of the computation in leaf nodes, \Cref{lem:twDP:IntroduceVertex,lem:twDP:IntroduceEdge,lem:twDP:forgetNode,lem:twDP:joinNode}, and \Cref{obs:treewidthSolutionAtRoot}.
\end{proof}

Finally, we note that the dynamic programming algorithm from \Cref{thm:twoTSSIsFPTWrtRoundsFmaxTw} shows that \twoTSSshort is in \XP parameterized by the treewidth alone.

\begin{corollary}
	\twoTSS can be solved in \mbox{$n^\Oh{\w}$} time on graphs of treewidth at most~$\w$.
\end{corollary}

\subsection{Treedepth and Maximum Threshold}\label{sec:treedepth}
In this section, we show that the algorithm from \Cref{thm:twoTSSIsFPTWrtRoundsFmaxTw} also applies for the combination of the treedepth of the input graph and the maximum threshold; in fact, we show that if the treedepth is bounded, then so is the length of any activation process (which might be of independent interest).

\begin{lemma}\label{lem:rounds_td}
	Let $t \ge 1$ and~$G$ be a graph and suppose that there exist sets $P^0_a$, and~$P^0_b$ such that the activation process in~$G$ from these sets takes at least~$t$ rounds, i.e., $P^t_a \neq P^{t-1}_a$ or $P^t_b \neq P^{t-1}_b$.
	Then $\td(G) \ge \log_3 (t+1)$.
\end{lemma}
It is well known that if a graph contains a path of length $h$, then its treedepth is at least $\lceil\log_2 (h+2)\rceil$;  see, e.g., \cite[pp. 117 and 118]{NesetrilM12Sparsity}.  %
For the proof of \Cref{lem:rounds_td}, we use the following generalization of this result.

\begin{lemma}\label{lem:walk_td}
	Let $G$ be a graph, $\mathring{G}$ be obtained from $G$ by adding a loop to each vertex, and $w$ be a walk of length~$h$ in~$\mathring{G}$ such that each vertex appears at most $r$ times on $w$.
	Then $\td(G) \ge \lceil\log_{r+1} (h+2)\rceil$.
\end{lemma}
\begin{proof}
	We prove the lemma by induction on the length $h$ of the walk $w$.
	
	If $G$ contains a nonempty walk, then it must have at least one vertex and, thus, ${\td(G) \ge 1}$.
	Therefore, the lemma is true whenever the length of the walk is $h \le r-1$.
	
	Assume now that $h \ge r$ and the lemma holds for all shorter walks.
	Note that the walk contains at least two different vertices and hence $|V(G)| > 1$.
	If $G$ is disconnected, then we can limit ourselves to the component which contains the walk, since $\td(G) = \max_{i \in [k]} \td(G_i)$, where  $G_1, \ldots, G_k$ are connected components of~$G$.
	Hence, we assume that $G$ is connected.
	Let $u$ be the vertex such that $\td(G) = 1+ \td(G \setminus \{u\}) = 1+\min_{v\in V(G)}\td(G \setminus \{v\})$.
	If $u$ is not part of the walk, then the walk $w$ must be contained in a single connected component $C$ of $\mathring{G} \setminus \{u\}$, and we can apply the lemma to this component, as $\td(G) = 1 + \td(G \setminus \{u\}) \ge \td(C)$.
	Hence, we assume that $u$ is part of walk $w$.
	
	As $u$ appears $q \leq r$ times on $w$, its removal splits the walk into at most $q +1$ parts, removing at most $2q$ edges. Let $w_1$ be a longest of these parts, its length is at least 
	\[\frac{h-2q}{q+1}=\frac{h-(2q+2)+2}{q+1}=\frac{h+2}{q+1}-2 \ge \frac{h+2}{r+1}-2.\] 
	As $w_1$ is a walk not containing $u$, it must be contained in a single component $C$ of $\mathring{G} \setminus \{u\}$. 
	Since $w_1$ contains each vertex at most $r$ times, the treedepth of $C$ is at least 
	\[\left\lceil\log_{r+1} \left(\frac{h+2}{r+1}-2+2\right)\right\rceil=\left\lceil(\log_{r+1} (h+2))-\log_{r+1} (r+1)\right\rceil=\left\lceil\log_{r+1} (h+2)\right\rceil-1\]
	by the induction hypothesis. Hence $\td(G) = 1+ \td(G \setminus \{u\}) \ge 1+ \td(C)$ is at least $\lceil\log_{r+1} (h+2)\rceil-1+1=\lceil\log_{r+1} (h+2)\rceil$. This finishes the proof.
\end{proof}

Now \Cref{lem:rounds_td} follows from \Cref{lem:walk_td} and \Cref{lem:rounds_walk}.

\begin{proof}[Proof of \Cref{lem:rounds_td}]
    Let $G$ be a graph and $\mathring{G}$ be obtained from $G$ by adding a loop to each vertex.
	By \Cref{lem:rounds_walk}, if there are sets such that the activation process in~$G$ from these sets takes at least~$t$ rounds, then there is a walk in $\mathring{G}$ of length $t-1$ such that each vertex appears at most twice on this walk.
	Hence, by \Cref{lem:walk_td}, the treedepth of $G$ is at least $\lceil\log_{3} (t+1)\rceil$, as required.
\end{proof}

\begin{corollary}\label{cor:td}
	\twoTSS with a maximum threshold of $f_{\max}$ can be solved in $( 3^{\d} \cdot f_{\max}+1)^{\Oh{\d}} \cdot n$ time on graphs of treedepth at most~$\d$.
\end{corollary}
\begin{proof}
	By \Cref{lem:rounds_td} we get that by setting
	\(
	{\fin := \min ( \fin, 3^{\d} )}
	\)
	we get an equivalent instance.
	It is well known that $\tw(G) \le \td(G)$ for every graph~$G$.
	Thus, we can apply \Cref{thm:twoTSSIsFPTWrtRoundsFmaxTw}.
\end{proof}

\section{Hardness Results}

In this section, we propose multiple parameterized reductions showing \Whness for different assumed (combinations of) structural parameters. All of our reductions start with the \PSI(\PSIshort) problem, which is defined as follows. We are given two undirected graphs $G$ and $H$ with $|V(H)| \le |V(G)|$ ($H$ is \emph{smaller})
and a~mapping~${\psi\colon V(G) \to V(H)}$. The~question is whether there is a mapping~${\phi\colon V(H) \to V(G)}$ such that $\phi$ is injective, ${\{\phi(u),\phi(v)\} \in E(G)}$ for each $\{u,v\} \in E(H)$, and $\psi \circ\phi$ is the identity. Since \textsc{Partitioned Clique}, which is known to be \Wcomplete with respect to the size of the clique~\cite{Pietrzak03}, is a special case of \PSIshort where $H$ is a complete graph, it follows that \PSIshort is \Whard with respect to~$|E(H)|$.

Most of our hardness results are accompanied with appropriate lower bounds based on the \emph{Exponential-Time Hypothesis} (ETH for short) of \citet{ImpagliazzoP01}. The hypothesis states that every algorithm solving \textsc{$3$-SAT} needs at least $2^{cn}$ time in worst-case, where $c>0$ is some universal constant and $n$ is the number of variables of the input formula. As the ETH originally handles \textsc{$3$-SAT}, we use the following subsequent theorem to be able to provide lower bounds directly from our \PSIshort reductions.

\begin{theorem}[{\citet[Corollary 6.3]{Marx10}}]\label{thm:PSI_ETH}
	If \PSI can be solved in time
	$f(m)\cdot n^{o(m/\log m)}$, where $f$ is an arbitrary computable function and
	$m$ is the number of edges of the smaller graph $H$, then ETH fails.
\end{theorem}

We note that the \PSIshort problem can be solved for each connected component of $H$ separately. Thus, we always assume that $|V(H)| \le |E(H)|+1$.

At first, to simplify our construction and argumentation used in hardness reductions, we introduce the following \emph{gadget} and independently prove some of its properties.

\paragraph{Selection Gadget: Selection of an Element in a Set}\label{sec:selection_gadget}
We first describe a gadget for the selection of a single element in a set~$W$; let $n_W = |W|$.
In what follows, please refer to Figure~\ref{fig:coloredVertexSelectionGadget}.
This gadget consists of~$n_W$ \emph{selection vertices} which are in one-to-one correspondence to the elements of~$W$ and two copies of a path on three vertices (\emph{guard paths}).
The central vertex of each guard path is connected to every selection vertex and one leaf (of each guard path) is in the set~$S_a$.
We stress that only selection vertices might be connected to other vertices in our hardness reductions.

\begin{figure}[tb!]
	\begin{center}
		\begin{tikzpicture}[node distance=.5cm]
  \tikzstyle{vertex}=[draw,thick,circle,minimum width=2pt,fill=white]
  \tikzstyle{vertexA}=[vertex,fill=red]
  \tikzstyle{vertexB}=[vertex,fill=blue]
  \tikzstyle{edge}=[thick]

  \newcommand{\n}{14}
  \newcommand{\firstGadgetAt}{2}
  \newcommand{\secondGadgetAt}{11}
  \newcommand{\topNodeDistance}{1.5cm}

  \node[vertex,label={[xshift=-3pt,yshift=5pt]180:$f_1=3$},label={[xshift=-3pt,yshift=-5pt]180:$f_2=3$}] (v1) {};
  \foreach \x[remember=\x as \xx (initially 1)] in {2,3,...,\n} {
    \node[vertex,right of=v\xx] (v\x) {};
  }
  \node[right of=v14,xshift=0.5em] (I) {$I$};

  \begin{scope}[on background layer]
    \node[draw,fill=gray!30,dashed,rounded corners,fit=(v1)(v\n)] {};
  \end{scope}

  \begin{scope}[node distance=\topNodeDistance]
    \node[vertex,label={[label distance=.4cm]$f_1 = 1$},label={$f_2 = 2$}] (vS) at ($(v\firstGadgetAt) + (0,1.5)$) {};
    \node[vertex,left of=vS,label={[label distance=.4cm]$f_1 = 1$},label={$f_2 = 2$}] (vSguard) {};
    \node[vertexA,right of=vS,label={$f_2 = 1$}] (vSselected) {};

    \draw[edge] (vSguard) to (vS) to (vSselected);
    \foreach \x in {1,2,...,\n} {
      \draw[edge] (vS) to (v\x);
    }
  \end{scope}

  \begin{scope}[node distance=\topNodeDistance]
    \node[vertex,label={[label distance=.4cm]$f_1 = 1$},label={$f_2 = 2$}] (vS) at ($(v\secondGadgetAt) + (0,1.5)$) {};
    \node[vertex,right of=vS,label={[label distance=.4cm]$f_1 = 1$},label={$f_2 = 2$}] (vSguard) {};
    \node[vertexA,left of=vS,label={$f_2 = 1$}] (vSselected) {};

    \draw[edge] (vSguard) to (vS) to (vSselected);
    \foreach \x in {1,2,...,\n} {
      \draw[edge] (vS) to (v\x);
      \draw[edge] (v\x) to +(255:.6);
      \draw[edge] (v\x) to +(285:.6);
    }
  \end{scope}
\end{tikzpicture}
		\caption{%
			An~overview of the vertex selection gadget.
			Red vertices are in the set~$S_a$.
			Individual values of thresholds~$f_1,f_2$ are above guard path vertices.
			All selection vertices have the same thresholds.
		}
		\label{fig:coloredVertexSelectionGadget}
	\end{center}
\end{figure}

\begin{lemma}\label{lem:atLeastOneBVertexInSelectionGadget}
	Let~$X$ and~$I$ be the vertex set and the set of selection vertices of a selection gadget, respectively.
	Any solution $T_a,T_b$ satisfies $|T_b \cap X| \ge 1$.
	Furthermore, if~$|T_b \cap X| = 1$ for a solution $T_a,T_b$, then $T_b \cap X = T_b \cap I$ and all the vertices of the guard paths receive both opinions in the activation process.
\end{lemma}
\begin{proof}
	Let  $T_a,T_b$ be a solution.
	Suppose for contradiction that~$T_b \cap X = \emptyset$.
	Let~$P$ be a guard path with vertex set~$\{v_1,v_2,v_3\}$, where~$v_2$ is the central vertex and~$v_1 \in S_a$.
	Now, either $v_2 \in T_a$ or in the first round of the activation process, vertex~$v_2$ receives opinion~$a$ from~$v_1$, i.e., $v_2 \in P^1_a$.
	Note that both $I \cap P^1_a$ and $I \cap P^1_b$ might be nonempty.
	Nevertheless, we get~$v_3 \in P^2_a$ and possibly $v_2 \in P^2_b$.
	But, since~$f_2(v_3) = 2$, we get that~$v_3 \notin P^i_b$ for any~$i \in \N$.
	Thus~$T_b$ is not a solution.
	
	We conclude that to have a solution, one of the following must hold:~$|N(v_2) \cap T_b| \ge 1$ or $v_2 \in T_a$ and $v_2 \in T_b$ (as otherwise~$v_3 \in P^2_a$ and $v_3 \notin P^2_b$).
	Thus, if~$|T_b \cap X| = 1$, then we have~$T_b \cap X = T_b \cap I$, since~$I = N(v_2) \cap N(v_2')$, where~$v_2'$ is the central vertex of the other guard path.
	Now, $v_1 \in P^2_b$ (since it receives the opinion~$b$ from~$v_2$) as well as~$v_3 \in P^2_a$ and~$v_3 \in P^2_b$ (receiving both opinions at the same time from~$v_2$).
\end{proof}

\subsection{Constant Maximum Threshold}\label{sec:hardness_const_tresholds}
In this section, we show that \twoTSSshort is \Whard with respect to the budget and the pathwidth, or the feedback vertex set of the input graph, even if all thresholds are bounded by $3$, presenting a reduction from \PSIshort. Let $(G,H,\psi)$ be an instance of \PSIshort.

While designing the selection gadget (as well as in the proof of \Cref{lem:atLeastOneBVertexInSelectionGadget}) we used a specific property of our model. %
One of the leaves of a guard path must receive both opinions, $a$ and~$b$, in the same round of the activation process.
We are going to utilize the same property once again when designing a gadget to check the incidence.
In the penultimate section we discuss possible variations of the model that prevent this behavior and how the hardness results carry over to these variations. 

\paragraph{Incidence Gadget}
An incidence gadget connects two selection gadgets---one for the selection of a vertex in~$V_w = \psi^{-1}(w)$ (as described earlier) and one for the selection of an edge in $E_{ww'} = \{\{u,v\}\mid \{u,v\} \in E(G),  u \in V_w, v \in V_{w'} \}$ for $w,w' \in V(H)$ with $\{w,w'\} \in E(H)$ (with the preselected vertex in the set~$S_b$, i.e., for an edge selection gadget we ``switch'' the role of red vertices).
We begin by enumerating vertices in~$V_w$ by numbers in~$\{ 1, \ldots, n \}$; let $\eta \colon V_w \to \{ 1, \ldots, n \}$ be the enumeration.
Now, the $ww'$-incidence check gadget consists of three \emph{connector} vertices (which we call vertex-connector, edge-connector, and super-connector), a \emph{sentry} vertex, and several paths as follows; refer to Figure~\ref{fig:incidenceGadget}.
We connect the vertex~$v \in V_w$ to the vertex-connector vertex by a path containing exactly~$n+\eta(v)$ additional vertices.
We connect the vertex~$e \in E_{ww'}$ to the edge-connector vertex by a path containing exactly~$n+\eta(v)$ additional vertices for $v = e \cap V_w$.
Finally, we connect the sentry vertex and vertex- and edge-connector to the super-connector vertex.
We set $f_1(x) = 1$ and $f_2(x) = 3$ for the sentry vertex~$x$ and the super connector vertex~$x$.
We set $f_1(x) = 1$ and $f_2(x) = 1$ for every vertex~$x$ in the paths connecting the selection vertices to the connector vertex and for the vertex and the edge connector vertices $x$.

\begin{figure}[tb!]
	\begin{center}
		\begin{tikzpicture}[node distance=.5cm]
  \tikzstyle{vertex}=[draw,thick,circle,minimum width=2pt,fill=white]
  \tikzstyle{vertexA}=[vertex,fill=red]
  \tikzstyle{vertexB}=[vertex,fill=blue]
  \tikzstyle{edge}=[thick]
  \tikzstyle{sentryStyle}=[vertex,draw=purple]
  \tikzstyle{connectorStyle}=[vertex,draw=yellow!80]
  \tikzstyle{longPath}=[edge,decorate,decoration={snake}]

  \newcommand{\nVertex}{7}
  \newcommand{\nEdge}{15}
  \newcommand{\sentryPosition}{5}

  \begin{scope}[yshift=8cm]
    \node[vertex,label={[xshift=-3pt,yshift=5pt]180:$f_1=3$},label={[xshift=-3pt,yshift=-5pt]180:$f_2=3$}] (v1) {};
    \foreach \x[remember=\x as \xx (initially 1)] in {2,3,...,\nVertex} {
      \node[vertex,right of=v\xx] (v\x) {};
    }
    \begin{scope}[on background layer]
      \node[draw,fill=gray!30,dashed,rounded corners,fit=(v1)(v\nVertex)] {};
    \end{scope}

    \draw[decorate,decoration={brace,amplitude=6pt}] ($(v\nVertex.south) + (.5,0)$) to node [midway,xshift=12pt] {\footnotesize $P_n$}  ($(v\nVertex.south) + (.5,0) - (0,1.8)$);
  \end{scope}

  \begin{scope}
    \node[vertex,label={[xshift=-3pt,yshift=5pt]180:$f_1=3$},label={[xshift=-3pt,yshift=-5pt]180:$f_2=3$}] (e1) {};
    \foreach \x[remember=\x as \xx (initially 1)] in {2,3,...,\nEdge} {
      \node[vertex,right of=e\xx] (e\x) {};
    }
    \begin{scope}[on background layer]
      \node[draw,fill=gray!30,dashed,rounded corners,fit=(e1)(e\nEdge)] {};
    \end{scope}
  \end{scope}

  \node[connectorStyle,label={0:$f_1=1 \,\, f_2=1$}] (connectorVertex) at ($(v\sentryPosition)!.5!(e\sentryPosition)$) {};
  \node[connectorStyle,label={0:$f_1=1 \,\, f_2=3$},below of=connectorVertex] (connectorMid) {};
  \node[connectorStyle,label={0:$f_1=1 \,\, f_2=1$},below of=connectorMid] (connectorEdge) {};
  \node[sentryStyle,label={[yshift=5pt]180:$f_1=1$},label={[yshift=-5pt]180:$f_2=3$}] (sentry) at ($(connectorMid) - (1,0)$) {};
  \draw[edge] (connectorMid) to (sentry);
  \draw[edge] (connectorVertex) to (connectorMid) to (connectorEdge);

  \node at ($(connectorVertex) + (0,2)$) {$\cdots$};

  \path (connectorVertex) to node[vertex,midway,xshift=-8pt] (v1first) {} (v1);
  \path (connectorVertex) to node[vertex,midway] (v2first) {} node[vertex,pos=.25] {} (v2);
  \path (connectorVertex) to node[vertex,midway] (vNfirst) {} node[vertex,pos=.125] {} node[vertex,pos=.25] {} node[vertex,pos=.375] {} (v\nVertex);

  \path (connectorEdge) to node[vertex,pos=.33] {} node[vertex,pos=.66] (v2EdgeLast1) {} (e2);
  \path (connectorEdge) to node[vertex,pos=.33] {} node[vertex,pos=.66] (v2EdgeLast2) {} (e7);
  \path (connectorEdge) to node[vertex,pos=.33] {} node[vertex,pos=.66] (v2EdgeLast3) {} (e\nEdge);

  \draw[longPath] (v1) to (v1first);
  \draw[longPath] (v2) to (v2first);
  \draw[longPath] (v\nVertex) to (vNfirst);

  \foreach \i/\j in {1/2,2/7,3/\nEdge} {
    \draw[longPath] (v2EdgeLast\i) to (e\j);
    \begin{scope}[on background layer]
      \draw[edge] (v2EdgeLast\i) to (connectorEdge);
    \end{scope}
  }

  \begin{scope}[on background layer]
    \draw[edge] (v1first) to (connectorVertex);
    \draw[edge] (v2first) to (connectorVertex);
    \draw[edge] (vNfirst) to (connectorVertex);
  \end{scope}
\end{tikzpicture}
		\caption{%
			An overview of the incidence gadget with vertex selection gadget on top and edge selection gadget in bottom.
			The curvy edge stands for a path of length~$n$.
			The yellow vertices are the connector vertices and the purple vertex is the sentry vertex.
		}
		\label{fig:incidenceGadget}
	\end{center}
\end{figure}

We denote the constructed underlying graph with~$\widehat{G}$.
It remains to set the length of the activation process and the budget for the reduced instance, which we do by setting~$\fin=2|V(\widehat{G})|$ and~$B = |V(H)| + |E(H)|$.
This completes the description of the reduction; we denote the resulting instance by~$\mathcal{I}$.
Observe that one can produce the reduced instance in time polynomial in the size of the given instance of PSI.
Note that there are exactly~$B$ selection gadgets in~$\mathcal{I}$.
Therefore, by \Cref{lem:atLeastOneBVertexInSelectionGadget} there must be exactly one vertex in the set~$T_b$ in each vertex selection gadget and exactly one vertex in the set~$T_a$ in each edge selection gadget.
We call such a pair of sets~$(T_a,T_b)$ \emph{good}.
It is not hard to see (using \Cref{lem:atLeastOneBVertexInSelectionGadget}) that if a pair of sets~$(T_a,T_b)$ with $T_a,T_b \subseteq V(\widehat{G})$ is not good, then the pair $(T_a,T_b)$ is not a solution for~$\mathcal{I}$.

\begin{lemma}\label{lem:goodTSSAtivatesAllInAtLeastOneOpinion}
	Let $(T_a, T_b)$ be a good pair for~$\widehat{G}$.
	It holds that~$P_a^{\fin} \cup P_b^{\fin} = V(\widehat{G})$.
\end{lemma}
\begin{proof}
	Since~$(T_a,T_b)$ is good, it contains exactly one vertex in every selection gadget (in a vertex selection gadget it is a vertex in~$T_b$ whereas in an edge selection gadget it is one in~$T_a$).
	Now observe that in at most $2n$ rounds we have that every vertex-connector vertex is in~$P_a \cup P_b$ as the activation process reaches it at the latest by the path from the selected vertex (since~$\eta(v) \le n$ for the selected vertex~$v$); similarly for the edge-connector vertex.
	To see this, recall that all vertices~$v$ introduced in the incidence gadget have~$f_1(v) = 1$.
	Note that the vertex-connector vertex may end up in~$P_a$, since the activation process for the edge selection gadget might reach it sooner than the one described above (again, the same holds for the edge-connector vertex by symmetry).
	Then, from a connector vertex, the activation process continues to the rest of the paths that connect it to its selection gadget.
	Furthermore, the super-connector vertex and the sentry vertex both receive at least one opinion.
	Thus, we have shown (together with \Cref{lem:atLeastOneBVertexInSelectionGadget}) that each vertex is in $P_a^{\fin} \cup P_b^{\fin}$.
\end{proof}

It follows from \Cref{lem:goodTSSAtivatesAllInAtLeastOneOpinion} that if a good pair~$(T_a,T_b)$ is a solution, then ${P_a^{\fin} = P_b^{\fin} = V(\widehat{G})}$.

\begin{lemma}\label{lem:goodSetTSSiffIncident}
	Let $(T_a, T_b)$ be a good pair for~$\widehat{G}$.
	Let $v \in V_w$ and let $e \in E_{ww'}$ be such that~$v$ is selected by~$T_b$ (i.e., $v \in T_b$) and~$e$ is selected by~$T_a$.
	For a sentry vertex~$s$ of the incidence gadget for~$V_w$ and~$E_{ww'}$ it holds that~$s \in P_a^{\fin} \cap P_b^{\fin}$ if and only if~$v \in e$.
\end{lemma}
\begin{proof}
	Let~$c_v$ and~$c_e$ be the vertex- and edge-connector vertex, respectively.
	Slightly abusing the notation, we define~$\eta(e) = \eta(V_w \cap e)$.
	It is not hard to verify that if $\eta(v) \le \eta(e)$, then $c_v \in P_b^{n+\eta(v)}$; similarly, if $\eta(e) \le \eta(v)$, then $c_e \in P_a^{n+\eta(e)}$.
	Note that in the first $2n$ rounds there is no interference between any two incidence gadgets, since the ``backpropagation'' from a connector vertex to the selection vertices uses at least two paths, both of length at least $n+1$.
	Consequently, we find that for the super-connector vertex~$c$ it holds that~$c \in P_a^{\ell} \cup P_b^{\ell}$ for $\ell = 1 + \min\{ n+\eta(v), n+\eta(e)\}$.
	Furthermore, we have $c \in P_a^{\ell} \cap P_b^{\ell}$ if and only if $\eta(v) = \eta(e)$, as otherwise either~$c \in P_a^{\ell}$ or~$c \in P_b^{\ell}$ depending on the minimizer of $\min\{ n+\eta(v), n+\eta(e)\}$.
	It follows that $s$ is in $P_a^{\ell + 1} \cap P_b^{\ell + 1}$ and in $P_a^{\fin} \cap P_b^{\fin}$ if and only if $\eta(v) = \eta(e)$, since otherwise it receives only one opinion and cannot receive the other due to $f_2(s) > \deg_{\widehat{G}}(s)$.
\end{proof}

\begin{figure}[tb!]
	\begin{center}
		\begin{tikzpicture}[node distance=.5cm]
  \tikzstyle{vertex}=[draw,thick,circle,minimum width=2pt,fill=white]
  \tikzstyle{vertexA}=[vertex,fill=red]
  \tikzstyle{vertexB}=[vertex,fill=blue]
  \tikzstyle{edge}=[thick]
  \tikzstyle{sentryStyle}=[vertex,draw=purple]
  \tikzstyle{connectorStyle}=[vertex,draw=yellow!80]
  \tikzstyle{longPath}=[edge,decorate,decoration={snake}]

  \newcommand{\nVertex}{7}
  \newcommand{\nVertexHalf}{4}
  \newcommand{\nEdge}{12}

  \begin{scope}[yshift=-1cm,xshift=-5mm]
    \node[vertex,label={[yshift=10pt]90:$f_1=\deg$},label={90:$f_2=\deg$}] (v1) {};
    \foreach \x[remember=\x as \xx (initially 1)] in {2,3,...,\nVertex} {
      \node[vertex,below of=v\xx] (v\x) {};
    }
    \begin{scope}[on background layer]
      \node[draw,fill=gray!30,dashed,rounded corners,fit=(v1)(v\nVertex),label={270:$V_w$}] (Vfitter) {};
    \end{scope}
  \end{scope}

  \begin{scope}[xshift=12cm]
    \node[vertex,label={[yshift=10pt]90:$f_1=\deg$},label={90:$f_2=\deg$}] (e1) {};
    \foreach \x[remember=\x as \xx (initially 1)] in {2,3,...,\nEdge} {
    \node[vertex,below of=e\xx] (e\x) {};
    }
    \begin{scope}[on background layer]
      \node[draw,fill=gray!30,dashed,rounded corners,fit=(e1)(e\nEdge),label={270:$E_{ww'}$}] (Efitter) {};
    \end{scope}
  \end{scope}

  \path (v1) to node[midway,vertex,label={90:$c^1_{ww'}$}, above=2cm] (c1) {} (Efitter);
  \node[vertex,label={20:$c^2_{ww'}$},label={[yshift=-3pt]270:$f_1=|V_w|$},label={[yshift=-13pt]270:$f_2=\deg$}] at ($(c1) - (0,5)$) (c2) {};

  \begin{scope}
    \node[vertex] at ($(c1) - (1.5,3)$) (a1) {};
    \foreach \x[remember=\x as \xx (initially 1)] in {2,3,...,\nVertex} {
      \node[vertex,right of=a\xx] (a\x) {};
    }
    \node[vertexA,label={0:$f_2=\deg$}, label={270:$s_{ww'}$}, below right of=a\nVertex, yshift=-10]  (A1selected) {};

    \foreach \x in {1,...,\nVertex} {
      \draw[edge, bend left] (A1selected) to (a\x);
      \draw[edge] (c1) to (a\x);
      \draw[edge] (c2) to (a\x);
    }

    \begin{scope}[on background layer]
      \node[draw,dashed,rounded corners,fit=(a1)(a\nVertex),label={0:$A_{ww'}$},label={[yshift=5pt]180:$f_1=1$},label={[yshift=-4pt]180:$f_2=2$}] (A1fitter) {};
    \end{scope}
  \end{scope}

  \begin{scope}
    \node[vertex] at ($(v1) + (1.5,1)$) (low1) {};
    \node[vertex,below of=low1] (low2) {};
    \node[vertex,below of=low2] (low3) {};

    \node[vertex] at ($(low3) - (0,3)$) (high1) {};
    \node[vertex,below of=high1] (high2) {};
    \node[vertex,below of=high2] (high3) {};
    \node[vertex,below of=high3,label={270:$f_1=1$},label={[yshift=-10pt]270:$f_2=1$}] (high4) {};

    \begin{scope}[on background layer]
      \node[draw,dashed,rounded corners,fit=(low1)(low2)(low3),label={[xshift=3pt]270:$\operatorname{low}(v)$}] (Vlow) {};
      \node[draw,dashed,rounded corners,fit=(high1)(high4),label={[xshift=3pt]90:$\operatorname{high}(v)$}] (Vhigh) {};
    \end{scope}
  \end{scope}

  \foreach \lh in {low1,low2,low3} {
    \draw[edge] (v3) to (\lh) to (c1);
  }

  \foreach \lh in {high1,high2,high3,high4} {
    \draw[edge] (v3) to (\lh) to (c2);
  }

  \begin{scope}
    \node[vertex] at ($(e3) + (-1.5,1)$) (ehigh1) {};
    \node[vertex,below of=ehigh1] (ehigh2) {};
    \node[vertex,below of=ehigh2] (ehigh3) {};
    \node[vertex,below of=ehigh3] (ehigh4) {};

    \node[vertex] at ($(ehigh4) - (0,3)$) (elow1) {};
    \node[vertex,below of=elow1] (elow2) {};
    \node[vertex,below of=elow2] (elow3) {};

    \begin{scope}[on background layer]
      \node[draw,dashed,rounded corners,fit=(elow1)(elow2)(elow3),label={[xshift=-5pt]90:$\operatorname{low}_w(e)$}] (Vlow) {};
      \node[draw,dashed,rounded corners,fit=(ehigh1)(ehigh4),label={[xshift=-5pt]270:$\operatorname{high}_w(e)$}] (Vhigh) {};
    \end{scope}
  \end{scope}

  \foreach \lh in {elow1,elow2,elow3} {
    \draw[edge] (e7) to (\lh) to (c2);
  }

  \foreach \lh in {ehigh1,ehigh2,ehigh3,ehigh4} {
    \draw[edge] (e7) to (\lh) to (c1);
  }

\end{tikzpicture}
		\caption{Illustration of the incidence gadget for the constant number of rounds.}
		\label{fig:incidenceGadgetForConstRounds}
	\end{center}
\end{figure}

\begin{theorem}\label{thm:twoTSSIsHardForTWandConstFMax}
	\twoTSS is \Whard when parameterized by $|S_a|+|S_b|$, the pathwidth, and the feedback vertex number of the input graph combined, even if the maximum threshold~$f_{\max}$ is~$3$ and $f_1(v) \le f_2(v)$ for every vertex~$v$.
	Moreover, unless ETH fails, there is no algorithm for \twoTSS with $f_{\max}=3$ and $f_1(v) \le f_2(v)$ for every vertex~$v$ running in $g(k)n^{o(k/\log k)}$, where $k$ is the sum of $|S_a|+|S_b|$, the pathwidth, and the feedback vertex number of the input graph and $g$ is an arbitrary computable function.
\end{theorem}
\begin{proof}
	Let us first verify the parameters of our reduction.
	It is straightforward to check that the maximum threshold $\max_{v \in V(\widehat{G})} \{ f_1(v), f_2(v) \}$ is~$3$.
	Next, if we remove from~$\widehat{G}$ all the connector vertices, as well as the vertices of both guard paths in each selection gadget (let us denote this set~$X$), we obtain a forest, where each tree has at most one vertex of degree more than $2$ (the selection vertex). The pathwidth of such a tree is at most $2$. As $|X| = 2\cdot 3\cdot |E(H)| + 6(|V(H)|+|E(H)|)= \Oh{|E(H)|}$, the pathwidth as well as the feedback vertex number of $\widehat{G}$ is $\Oh{|E(H)|}$. Furthermore, as $S_a \cup S_b \subseteq X$, we also have $B < |S_a|+|S_b| = \Oh{|E(H)|}$. Therefore, once we verify the correctness of the reduction, the results will follow from \Cref{thm:PSI_ETH} and the discussion afterwards.
	
	From \Cref{lem:goodSetTSSiffIncident} we next conclude that a good pair of sets $(T_a,T_b)$ is a solution for~$\widehat{G}$ if and only if the selected vertices are incident to the selected edges.
	However, this can only happen when the original instance of \PSIshort is \yesI.
	
	Let $(T_a, T_b)$ be a good pair for~$\widehat{G}$ such that the selection it imposes yields a solution to the given instance of PSI.
	We consider the incidence gadget for~$V_w$ and~$E_{ww'}$; let $v \in V_w$ and $e \in E_{ww'}$ be the vertex and edge selected by~$(T_a,T_b)$, respectively.
	Since we have $v \in e$, we have $\eta(v) = \eta(e)$.
	Furthermore, it follows from the proof of \Cref{lem:goodSetTSSiffIncident} that the super-connector vertex $c$ receives both opinions in round~$\ell = 1 + n + \eta(v)$.
	Notice that in round~$\ell-1$ both vertex- and edge-connector received opinion~$b$ and~$a$, respectively.
	Now, in round~$\ell+1$ the vertex- and edge-connector have both opinions.
	We observe that in round $\ell$ the vertices adjacent to the vertex-connector receive the opinion~$b$ (similarly for the vertices adjacent to the edge-connector).
	Since every vertex~$u$ on paths connecting the vertex-connector to the~$V_w$ selection gadget has $f_1(u)=f_2(u)=1$, we see that in round $\ell+2$ the vertices adjacent to the vertex-connector have both opinions.
	Now, if a vertex $u$ belongs to a path connecting the vertex-connector and $u \in P_b^q \setminus P_b^{q-1}$ for some $q \in \N$ (note that such a~$q$ exists), then $u \in P_a^{q+2}$.
	Thus, in round~$\ell+2n+2$ all vertices on all these paths have both opinions~$a$ and~$b$.
	Furthermore, in round $\ell+n+\eta(x)$ the vertex $x \in V_w \setminus \{v\}$ receives opinion~$b$ and in round $\ell+n+\eta(x)+2$ a vertex $x \in V_w$ receives the other opinion~$a$.
	Thus, if $v \in e$, then all vertices of the incidence gadget as well as the selection gadgets in it have both opinions in round $\ell+2n+3$.
	A symmetric argument holds for the edge selection gadget.
	Thus, the two instances are equivalent.
\end{proof}

\subsection{Constant Duration of Activation Process}\label{sec:hardness_const_duration}
In this section, we show how our ``colored'' selection gadgets and the ideas presented in the \Whness reduction (with respect to the treewidth) of \citet{Ben-ZwiHLN11} yield a hardness result even if we assume that the length of the activation process is bounded by a constant.
The reduction is again from \PSIshort and again uses the edge representation strategy. Let $(G,H,\psi)$ be an instance of \PSIshort.

The key idea behind the reduction of \citet{Ben-ZwiHLN11} is to assign two enumerations to every vertex in~$V_w$ for $w \in V(H)$ with $\low\colon V_w \to [|V_w|]$ and $\high\colon V_w \to [|V_w|]$ so that for every vertex $v \in V_w$ we have $\low(v) + \high(v) = |V_w|$; if the reader is familiar with the original proof, we must admit that our use of their gadgets works in a somewhat simpler regime.
This time, all our selection gadgets (that is, the vertex and edge selection gadgets) are the same as in \Cref{fig:coloredVertexSelectionGadget} (i.e., we only use~$T_b$ for the selection) and we set $f_1(v) = f_2(v) = \deg(v)$ for every selection vertex~$v$ in these gadgets.
Crucially, the reduction of \citet{Ben-ZwiHLN11} is highly ``compact'' in the sense that the result of their reduction has a very low diameter, which then yields a strong bound on the number of rounds of a (successful) activation process.
We set the budget~$B$ to the number of selection gadgets, that is, $B = |V(H)| + |E(H)|$.
Now, we are ready to give a description of the incidence gadget.

\paragraph{Incidence Gadget}
We introduce an incidence gadget for every edge $\{w,w'\} \in E(H)$ and node~$w$.
As usual, the purpose is to verify that the selected vertex in the set~$V_w$ is incident to the selected edge in the set~$E_{ww'}$.
To this end, we define the mappings $\low$ and $\high$ for the edges as well; however, now there are two such pairs of mappings for each edge.
That is, for an edge $e \in E_{ww'}$ we have mappings $\low_w,\high_w,\low_{w'}$, and $\high_{w'}$, where $\low_w(e) = \low(V_w \cap e)$ and similarly for the other mappings.

The incidence gadget connects the selection gadget for~$V_w$ and $E_{ww'}$ in the following way; please refer to \Cref{fig:incidenceGadgetForConstRounds}.
We introduce two new \emph{checking vertices} $c^1_{ww'}$ and $c^2_{ww'}$ and set their thresholds to $f_1(x) = |V_w|$ and $f_2(x) = \deg(x)$ (for $x \in \{c^1_{ww'},c^2_{ww'}\}$).
For each vertex $v \in V_w$ we introduce $|V_w|$ new vertices that we all connect to $v$ and then connect $\low(v)$ to $c^1_{ww'}$ and connect the rest (that is, $\high(v)$) to $c^2_{ww'}$.
Similarly, we connect the selection vertices in the set~$E_{ww'}$, however, this time we ``switch the roles'' of $c^1_{ww'}$ and $c^2_{ww'}$.
For each edge $e \in E_{ww'}$ we introduce $|V_w|$ new vertices, which we all connect to $e$ and then connect $\low_w(e)$ to $c^2_{ww'}$ and the rest (i.e., $\high_w(e)$) to $c^1_{ww'}$.
We set $f_1(x) = f_2(x) = 1$ for all the vertices described here.
The last group of vertices that we add is formed by a \emph{special vertex} $s_{ww'}$ that is in the input set~$S_a$ and a group $A_{ww'}$ of~$|V_w|$ vertices connected to the special vertex, as well as to the checking vertices $c^1_{ww'}$ and $c^2_{ww'}$.
For the special vertex~$s$ we have $f_2(s) = \deg(s)$ and for the other vertices~$x$ we have $f_1(x) = 1$ and $f_2(x)=2$.

This finishes the description of our reduction.
We denote the resulting graph~$\widehat{G}$.
It is not hard to verify that graph~$\widehat{G}$ can be constructed in time polynomial in the sizes of $G$ and~$H$.

By \Cref{lem:atLeastOneBVertexInSelectionGadget} we get the following.
\begin{itemize}
	\item
	$T_a = \emptyset$ in every solution of size~$B$ for~$\widehat{G}$ and
	\item
	$T_b$ contains exactly one vertex in each selection gadget in every solution of size~$B$ for~$\widehat{G}$.
\end{itemize}
We say that the pair $(T_a,T_b)$ with $T_a,T_b \subseteq V(\widehat{G})$ is \emph{good} if it meets both of the above conditions.

\begin{lemma}\label{lem:checkingVerticesInRoundTwoIFFIncident}
	Let $(T_a,T_b)$ be a good pair, and let $v \in V_w$ and $e \in E_{ww'}$ be the selected vertex and edge, respectively.
	Then we have
	\begin{enumerate}
	        \item $A_{ww'} \subseteq P_a^1$,
		\item 
		$c^1_{ww'},c^2_{ww'} \in P^2_a$,
		\item
		$c^1_{ww'},c^2_{ww'} \in P^2_b$ if and only if $v \in e$, and
		\item
		if $v \in e$, then we have $A _{ww'} \subset P^3_b$ and $s_{ww'} \in P^4_b$
		\item
		if $v \notin e$, then we have $\big( \{c^1_{ww'},c^2_{ww'},s_{ww'}\} \cup A _{ww'} \big) \cap P^i_b \subsetneq \{c^1_{ww'},c^2_{ww'}\}$ for every $i \in \N$.
	\end{enumerate}
\end{lemma}
\begin{proof}
	We begin with the first claim.
	Since every vertex in~$A_{ww'}$ neighbors the special vertex which is in~$S_a$ and no vertex in $S_b \cup T_b$, we have $A_{ww'} \subseteq P_a^1 \setminus P_b^1$.
	Then we have $N(c^1_{ww'}) \cap P_a^1 = N(c^2_{ww'}) \cap P_a^1 = A_{ww'}$, 
	Furthermore, since for $c \in \{c^1_{ww'},c^2_{ww'}\}$ we have $|A_{ww'}| = |V_w| = f_1(c)$ and since $(T_a,T_b)$ is good and thus $N(c) \cap P^0_b = \emptyset$, we get $c \in P^2_a$.
	
	To prove the third assertion, we observe that
	\[\left| N(c^1_{ww'}) \cap P^1_b \right| = \low(v) + \high_w(e) \]
	and
	\[\left| N(c^2_{ww'}) \cap P^1_b \right| = \high(v) + \low_w(e) \,.\]
	
	Thus, we have $c^1_{ww'} \in P^2_b$ if and only if $\low(v) + \high_w(e) \ge |V_w|$.
	We get the following set of equivalent expressions.
	\begin{align*}
	\low(v) + \high_w(e) &\ge |V_w| \\
	\low(v) + |V_w| - \low_w(e) &\ge |V_w| \\
	\low(v) &\ge \low_w(e)
	\end{align*}
	Similarly, we get $c^2_{ww'} \in P^2_b$ if and only if
	\[
	\low(v) \le \low_w(e) \,.
	\]
	Furthermore, we have $N(c) \cap P^0_a = \emptyset$ for $c \in \{ c^1_{ww'},c^2_{ww'} \}$.
	We conclude that $c^1_{ww'}, c^2_{ww'} \in P^2_b$ if and only if $\low(v) = \low_w(e)$ which is by the definition if and only if $v \in e \cap V_w$.
	
	The fourth claim follows from the previous ones.
	Since vertices in~$A_{ww'} \subseteq P^1_a$, $A_{ww'} \cap P^1_b = \emptyset$, and $A_{ww'} \cap P^2_b = \emptyset$, we have that a vertex $x \in A_{ww'}$ is in $P^3_b$ if and only if $\left| N(x) \cap P^2_b \right| \ge f_2(x) = 2$.
	If $v \in e$, then $c^1_{ww'}, c^2_{ww'} \in P^2_b$ by the previous claim and, hence, $A_{ww'} \subseteq P^{3}_b$.  
	For the special vertex~$s_{ww'}$ we have~$s_{ww'} \in S_a$ and $f_2(s_{ww'}) = |A_{ww'}|$, and thus, $s_{ww'} \in P^{4}_b$. 
	
	For the fifth claim, we use induction on $i$. We already argued that it holds for $i \in \{0,1,2\}$ as $P^2_b$ contains at most one of $c^1_{ww'}, c^2_{ww'}$ by the third claim and no vertex from $A_{ww'}$ nor $s_{ww'}$.
	Suppose now that the claim holds for some $i \ge 2$, we want to show that 
	$\big( \{c^1_{ww'},c^2_{ww'},s_{ww'}\} \cup A _{ww'} \big) \cap \big(P^{i+1}_b \setminus P^{i+1}_b\big) = \emptyset$.
	If $r \in \{1,2\}$ is such that $c^r_{ww'} \notin P^{i}_b$, then, as $c^r_{ww'} \in P^{i}_a$, in order to include   $c^r_{ww'}$ in $P^{i+1}_b$ all of its neighbors must be in $P^{i}_b$. But, by assumption, the vertices in $A _{ww'}$ are not in $P^{i}_b$ and, thus, $c^r_{ww'} \notin P^{i+1}_b$.
	A vertex $x \in A_{ww'}$ is in $P^{i+1}_b$ only if $\left| N(x) \cap P^i_b \right| \ge f_2(x) = 2$.
	But $N(x) = \{c^1_{ww'},c^2_{ww'},s_{ww'}\}$ and at most one of these vertices is in $P^i_b$.
	Thus $x \notin P^{i+1}_b$.	
	In order to have $s_{ww'} \in P^{i+1}_b$ we need $A_{ww'} \subseteq P^{i}_b$.
	As this is not the case, we also have $s_{ww'} \notin P^{i+1}_b$, finishing the induction.
\end{proof}

We now observe that, irrespective of the selection, the special vertices have the ability to spread their opinion~$a$ to every vertex within four rounds.

\begin{lemma}\label{lem:goodGivesAToAllInFourRounds}
	Let $(T_a,T_b)$ be a good pair.
	Then we have that $P_a^4 = V(\widehat{G})$.
\end{lemma}
\begin{proof}
	By \Cref{lem:checkingVerticesInRoundTwoIFFIncident,lem:atLeastOneBVertexInSelectionGadget} we have that $P^2_a$ contains all vertices in $A_{ww'}$, all special vertices, all checking vertices, and all vertices in guard paths.
	Fix a node $w \in V(H)$ and $\{w,w'\} \in E(H)$.
	We get that all low- and high-vertices (i.e., the vertices incident to $c^i_{ww'}$ not in $A_{ww'}$ for $i = 1,2$) are in $P^3_a$, since the checking vertices are in~$P^2_a$ by \Cref{lem:checkingVerticesInRoundTwoIFFIncident} and all of these vertices have their respective thresholds set to~$1$.
	Thus, for a selection vertex~$x$ (either in $V_w$ or in $E_{ww'}$) we have that $N(x) \subseteq P^3_a$.
	We get that $P_a^4 = V(\widehat{G})$ and the lemma follows.
\end{proof}

\begin{theorem}\label{thm:twoTSSIsHardForTDandConstRounds}
	\twoTSS is \Whard when parameterized by the treedepth, the feedback vertex number, the 4-path vertex cover number of the input graph and the budget $B$ combined, even if any successful activation process is guaranteed to stabilize in~$\fin = 4$ rounds and $f_1(v) \le f_2(v)$ for every vertex~$v$.
	
	Moreover, unless ETH fails, there is no algorithm for \twoTSS running in $g(k)n^{o(k/\log k)}$, where $k$ is the sum of the budget $B$, the treedepth, the feedback vertex number, and the 4-path vertex cover number of the input graph and $g$ is an arbitrary computable function, even if any successful activation process is guaranteed to stabilize in~$\fin = 4$ rounds and $f_1(v) \le f_2(v)$ for every vertex~$v$.
\end{theorem}
\begin{proof}
	Let us again first verify the parameters of our reduction.
	If we remove from~$\widehat{G}$ all the checking vertices as well as the vertices of both guard paths in each selection gadget (let us denote this set~$X$), we obtain a forest composed of stars of two kinds.
	Namely, the first kind of stars is formed by a special vertex $s_{ww'}$ and leaves $A_{ww'}$ for some~$\{w,w'\} \in E(H)$.
	The second kind is a selection vertex with several groups of adjacent vertices, one for each incidence gadget it is a part of.
	As both of them are stars, their treedepth is $2$ and they contain no paths on $4$ vertices. Since $|X| = 2\cdot 2\cdot |E(H)| + 6(|V(H)|+|E(H)|)= \Oh{|E(H)|}$, the treedepth, the feedback vertex number, and the 4-path vertex cover number of $\widehat{G}$ are $\Oh{|E(H)|}$. Furthermore, as $S_a \cup S_b \subseteq X$, we also have $B < |S_a|+|S_b| = \Oh{|E(H)|}$. Hence, once we verify the correctness of the reduction, the results will follow from \Cref{thm:PSI_ETH} and the discussion thereafter.
	
	For the correctness of the reduction, let us first assume that there is a solution $(T_a,T_b)$ for the constructed instance $(\widehat{G}, f_1, f_2, S_a, S_b, 4, B)$.
	By \Cref{lem:atLeastOneBVertexInSelectionGadget} $(T_a,T_b)$ is good.
	By \Cref{lem:goodGivesAToAllInFourRounds} we have $P^4_a = V(\widehat{G})$, that is, the activation process for $a$ stabilizes in four rounds.
	By \Cref{lem:checkingVerticesInRoundTwoIFFIncident} we have that the special vertices are in $P^i_b$ for some $i \in N$ only if the selected vertices are incident to the selected edges, i.e., only if the original instance of PSI is a \yesI.
	
	Now suppose that the instance of PSI is a \yesI.
	Create a pair $(T_a,T_b)$ by letting $T_a= \emptyset$ and $T_b$ be formed from the selection vertices selected according to a fixed solution for the PSI instance.
	Obviously $(T_a,T_b)$ is good.
	By \Cref{lem:goodGivesAToAllInFourRounds} we have $P^4_a = V(\widehat{G})$.
	By \Cref{lem:checkingVerticesInRoundTwoIFFIncident}, all checking vertices are in $P^2_b$, $A_{ww'} \subseteq P^3_b$ and all special vertices are in~$P^4_b$.
	Furthermore, we get that all vertices incident to the checking vertices are in $P^3_b$, as each such vertex~$v$ has $f_2(v)$ equal to the number of checking vertices in $N(v)$.
	Note that now the set $V(\widehat{G}) \setminus P^3_b$ consists only of the selection vertices not in~$T_b$; each such vertex has all its neighbors in $P^3_a \cap P^3_b$.
	From this it follows that $P^4_b=V(\widehat{G})$ and the theorem follows.
\end{proof}

Note that we proved \Cref{thm:twoTSSIsHardForTDandConstRounds,thm:twoTSSIsHardForTWandConstFMax} while using the assumption that for each vertex~$v$ it holds that~$f_1(v) \le f_2(v)$.
However, this might not be always the case in applications. 
Next we show that the \twoTSSshort problem remains hard even if $f_1(v) \ge f_2(v)$ holds for every vertex~$v$. %

\begin{theorem}\label{thm:twoTSSIsHardForTDandConstRoundsWithFOneLarger}
	\twoTSS is \Whard when parameterized by the treedepth of the underlying graph and the budget~$B$ combined even if~$\fin = 6$ rounds and $f_1(v) \ge f_2(v)$ for every vertex~$v$.
\end{theorem}
\begin{proof}
	In this proof, we once again take advantage of the known proof of \Whness of \TSSshort for parameterization by treedepth.
	Let $(G, f, k)$ be an instance of \TSSshort where~$G = (V,E)$ is an~$n$-vertex graph.
	We may assume that $f(v) \le n$ for all~$v \in V$ and that~$k < n$, as this is the case in the proof of \citet{Ben-ZwiHLN11}.
	We construct an instance of \twoTSSshort as follows:
	We attach~$n$ new vertices $v^1, \ldots, v^n$ to every vertex~$v$ in~$G$ and add all of these to the set~$S_a$.
	Call the resulting graph $H$.
	We leave the set~$S_b$ empty.
	Finally, we set $f_1(v) = n$ and $f_2(v) = f(v)$ for all $v \in V$ and $f_1(v) = n$ and $f_2(v) = 1$ for all $v \notin V$.
	This finishes the description of the \twoTSSshort instance~$(H,f_1, f_2, S_a, S_b,6, k)$.
	
	Let us first argue that if $(G,f,k)$ is a yes-instance, then $(H,f_1, f_2, S_a, S_b, 6,k)$ is a yes-instance.
	Let $T \subseteq V$ be a target set for~$G$ of size~$k$; it is worth pointing out that it follows from the proof in~\cite{Ben-ZwiHLN11} that~$T$ activates~$G$ in~$4$ rounds of the activation process.
	We show that $T_a = \emptyset$ and $T_b = T$ is a solution to $(H,f_1, f_2, S_a, S_b,6, k)$.
	We observe the following.
	\begin{claim}\label{clm:twoTSSOneLargerPA1claim}
		We have $P_a^1 = V(H)$, since every vertex not in~$S_a$ has at least $n$ neighbors in $S_a = P_a^0$. Note that this holds independently of $T_a$ and $T_b$.
		\hfill$\lhd$
	\end{claim}
	It holds that $P_b^1 \cap V = P_b^0 = T_b$, since~$k < n$.
	Now, since $T$ is a target set of~$G$ and since by the above claim for all vertices~$v$ of~$H$ the function $f_2(v)$ applies after the first round, we get that $V(G) \subseteq P_b^5$.
	Finally, since $f_2(v) = 1$ for all newly added vertices~$v$, we get $P_b^6 = V(H)$.
	
	Let us now argue that if $(H,f_1, f_2, S_a, S_b,6, k)$ is a yes-instance, then so is $(G, f, k)$.
	Let $(T_a, T_b)$ be a solution to $(H,f_1, f_2, S_a, S_b,6, k)$.
	By \Cref{clm:twoTSSOneLargerPA1claim} we have $P_a^1 = V(H)$ independently of $T_a$. Therefore $P_b^1 \cap V \subseteq T_a \cup T_b$, since each vertex $v \in V \setminus (T_a \cup T_b)$ has at most $|T_b|\le k < n=f_1(v)$ neighbors in $P_b^0$.
	
	If $T_a \neq \emptyset$, then consider sets $\widetilde{T}_a=\emptyset$ and $\widetilde{T}_b=T_a \cup T_b$, suppose that $P^t_b=V(H)$ for some $t \le 6$ and let $\widetilde{P}_a^0, \widetilde{P}_b^0, \widetilde{P}_a^1, \widetilde{P}_b^1, \ldots $ be the activation process arising from $\widetilde{T}_a$ and $\widetilde{T}_b$. We still have $\widetilde{P}_a^1 =V(H)$ and we have $\widetilde{P}_b^1 \cap V = \widetilde{T}_b = T_a \cup T_b \supseteq P_b^1 \cap V$. Therefore, we have $\widetilde{P}_b^i \supseteq P_b^i$ for every $i$ and, in particular, $\widetilde{P}_b^t =V(H)$. Therefore, we can assume that $T_a = \emptyset$.
	
	Next, we observe that we may assume $T_b \subseteq V(G)$.
	To see this, suppose we have $v^i \in T_b$ for some $v \in V(G)$ and $i \in [n]$.
	We claim that $\widehat{T}_b = (T_b \setminus \{v^i\}) \cup \{v\}$ is a solution to $(H,f_1, f_2, S_a, S_b,6, k)$ as well (we show this for a single vertex~$v$, however, by a similar reasoning one can do this for all such $v^i$ simultaneously).
	Let $P_b^0, P_b^1, \ldots, P_b^t$ be the activation process arising from $T_b$, i.e., $P_b^0 = T_b$, with $P_b^t = V(H)$.
	Let $\widehat{P}_b^0, \widehat{P}_b^1, \ldots, \widehat{P}^{\widehat{t}}_b$ be the activation process arising from $\widehat{T}_b$.
	We observe that $v^i$ has $v$ as a neighbor in $\widehat{P}_b^0=\widehat{T}_b$ and $f_2(v^i)=1$, whereas $v$ is the only vertex for which $v^i$ is a neighbor in $P_b^0=T_b$. Therefore we have $\widehat{P}_b^1 \supseteq P_b^1$, $\widehat{P}_b^i \supseteq P_b^i$ for every $i$ and, thus, $\widehat{P}_b^t =V(H)$.
	
	Therefore, we have $P_b^1 \cap V = P_b^0 = T_b$. As $P_b^t \cap V =V$, by setting $T=T_b$ we obtain a solution for $(G,f,k)$, finishing the proof.
\end{proof}

\section{Model Variations}\label{sec:variants}

In this section, we discuss possible variations of the model that might make it more realistic.
We discuss how the results of this paper carry over to the variations.

\subsection{Allowing Limited Imbalance}
\label{subs:imbalance}
We are aware that in real-world applications, the goal is not to balance the spread exactly equally, but to allow some imbalance between the spread of both opinions. However, to keep the model as simple as possible, we decided to study an equally balanced spread, as it is a special case of the more general model. 

To allow some imbalance, one would augment the instance with either a non-negative integer $D$ representing the number of unbalanced vertices, or a real number $d$ representing the ratio of vertices that can remain imbalanced at the end of the process. For any constant $D$ or $0 \le d <1$, this augmented model is at least as hard as \twoTSSshort{} as can be seen from the following simple reduction.
Indeed, given an instance of \twoTSSshort, we can add $D$ or $\frac{d}{1-d}|V|$ new isolated vertices with a single opinion to the network to obtain an instance of the augmented model.
These vertices deplete the imbalance limit, unless some part of the budget is invested in them, thus forcing the vertices of the original graph to be (almost) balanced.

In fact, in our hardness instances, a decrease in the budget leads to a larger increase in the imbalance (no vertex selected in a selection gadget results in at least two imbalanced vertices in that gadget, etc.). Thus, to achieve an outcome with imbalance at most $D$ or $d$, it is necessary to achieve complete balance on the original graph. It is easy to see that the addition of isolated vertices does not affect the structural parameters studied in this work. Therefore, as long as $|S_a|+|S_b|$ is not a parameter, the hardness results carry over to the augmented model, allowing limited imbalance.

Conversely, our algorithms can be easily adapted to handle this augmented model.
We let $D = d \cdot |V|$ in case only $d$ is given.
The modifications for the algorithm parameterized by the vertex cover number are not that straightforward. 
More precisely, while \Cref{rr:vertexWithBothOpinions} and \Cref{rr:vertexNotInVCfValueReduction} can also be used in this setting, \Cref{rr:vertexInSNeverGetSecondOpinion} and \Cref{rr:boundBudgetByVCNumber} need a nontrivial modification.
We refrain from describing those here.
Instead, we start by describing the modifications of the algorithms parameterized by 3-path vertex cover number and vertex integrity.
These can also be used for the case of vertex cover, although with slightly worse running time.

The first necessary modification is that we now need to guess that $r_a(u) < \infty$ and $r_b(u) = \infty$ (or vice-versa) for a vertex of $U$. 
Let $\varphi_U$ be the number of vertices in $U$ with this guess.
To handle this guess, we add the following linking constraint for $b$ (we again only show the case $r_a(u) < \infty$, $r_b(u) = \infty$):
\begin{equation}
 \sum_{v \in N(u)} x^{b,\fin}_v \le f_2(u)-1 \label{eq:twotssNFold:globalUpperbound:case5} 
\end{equation}
(note that the bound is now based on $f_2(u)$, instead of $f_1(u)$ in \eqref{eq:twotssNFold:globalUpperbound:case2}).
If $r_a(u) \ge 1$, then we add constraint \eqref{eq:twotssNFold:globalLowerbound:case1} for $a$, and if $r_a(u) \ge 2$, then also constraint \eqref{eq:twotssNFold:globalUpperbound:case1} for~$a$ and \eqref{eq:twotssNFold:globalUpperbound:case3} for~$b$.

The second modification is that we no longer impose the local constraint \eqref{eq:twotssNFold:FinalAEqualsB} to ensure the balanced outcome.
Instead, we introduce a new binary variable $d_v$ for each vertex $v \in V \setminus U$ that represents whether the vertex is imbalanced, i.e., subject to the following constraint:
\begin{equation}
 d_v = \left[x^{a,\fin}_v = x^{b,\fin}_v\right] \label{eq:twotssNFold:localimbalance}\,.
\end{equation}
Finally, we add a single new global constraint limiting the total number of imbalanced vertices:
\begin{equation}
 \sum_{v \in V \setminus U} d_v \le D - \varphi_U \label{eq:twotssNFold:globalimbalance}\,.
\end{equation}
It is straightforward to check that this modified N-fold IP indeed correctly represents the augmented model.

The treewidth algorithm also needs several simple modifications.
First of all, we add to the solution pattern a non-negative integer $m$ representing the maximum number of imbalanced vertices in this part of the graph.
Naturally, we omit the first point of the definition of a solution pattern being valid for a vertex $v$.
Similarly, when defining viability of the modified activation process, we drop the condition that $\widehat{P}_a^{\fin}= \widehat{P}_a^{\fin+1} = \widehat{P}_b^{\fin}= \widehat{P}_b^{\fin+1}$ and only require that $\widehat{P}_a^{\fin}= \widehat{P}_a^{\fin+1}$ and $\widehat{P}_b^{\fin}= \widehat{P}_b^{\fin+1}$, i.e., that the process stabilizes.
The third condition is applied even if $r_b(v) = \infty$ or $r_a(v)=\infty$, respectively.
A solution complies with a pattern, if it now additionally satisfies that $\left|\widehat{P}_a^{\fin} \setminus \widehat{P}_b^{\fin}\right| + \left|\widehat{P}_b^{\fin} \setminus \widehat{P}_a^{\fin}\right| \le m$, i.e., there are not too many imbalanced vertices.
Obviously, a pair of sets is then a solution to the augmented model if and only if it complies with the pattern $(\emptyset,\emptyset,\emptyset,\emptyset,\emptyset,\emptyset,\emptyset,\emptyset,D)$ at the root.

As to the algorithm itself, in a leaf node we again set the value to zero for any pattern $(\emptyset,\emptyset,\emptyset,\emptyset,\emptyset,\emptyset,\emptyset,\emptyset,m)$ with $m \ge 0$. 
In the introduce vertex node, if the introduced vertex $v$ is imbalanced according to $r_a$ and $r_b$, then we use $m-1$ in the pattern for $y$, otherwise we use $m$. 
In the introduce edge node, we use $m$ without modification in the pattern for $y$.
In the forget node we now also define 
\begin{align*}
\bm{q}_{0,\infty} &= (0,\infty,0,0,0,f_2(v)-1,0,0)                  \\
\bm{q}_{r_a,\infty} &= (r_a,\infty,f_1(v),0,f_1(v)-1,f_2(v)-1,0,f_1(v)-1)   & \text{for } r_a \in [\fin], r_a \ge 1
\end{align*}
(leaving out the symmetric cases)
and let 
\[
\mathcal{Q}' =
\begin{cases}
\{ \bm{q}_{0,0}\}
& \text{if } v \in S_a \cap S_b\\
\big\{ \bm{q}_{0,r_b} \,\big|\, r_b \in [\fin] \cup \{\infty\}\big\}
& \text{if } v \in S_a \setminus S_b\\
\big\{ \bm{q}_{r_a,0} \,\big|\, r_a \in [\fin]\cup \{\infty\}\big\}
& \text{if } v \in S_b \setminus S_a\\
\big\{ \bm{q}_{r_a,r_b} \,\big|\, r_a,r_b \in [\fin]\cup \{\infty\} \big\}
& \text{if } v \notin S_a \cup S_b\, .
\end{cases}
\]
The recursion is made exactly as before, we use $m$ without modification in the pattern for~$y$.
In the join node, we try all possible combinations of non-negative $m_y$ and $m_z$ in the patterns for $y$ and $z$, respectively, such that $m_y+m_z-m$ is exactly the number of vertices in $\beta(x)$ imbalanced according to $r_a$ and $r_b$.

The modification adds a factor of $D^2$ (which is at most $n^2$) to the running time (because of the computation in the join node).
The correctness of the modified algorithm can be proven along the lines of the correctness of the original algorithm and we omit it here.
 
\subsection{No Simultaneous Gain of Opinions}
For our hardness results, we heavily use that if both opinions pass the first threshold of an agent at the same time, she receives both these opinions in the next round.
What if, in the activation process, the agent prefers one of the opinions; i.e., if both opinions pass the first threshold, she only receives the preferred opinion.
The preferred opinion can be either agent-specific~\cite{TongDW18} (this might also correspond to a sample from a model where the agent receives a random opinion in this case~\cite{CarnesNWZ2007}) or instance-specific (such as in the ``good information'' mode of \citet{BudakAA2011}). Obviously, the former case is more general and, thus, at least as hard as the latter one.

In the following, we show how our gadgets can be tweaked to work for instance-specific preferred opinion, which is a special case of agent-specific preferences.

\paragraph{Synchronizing Gadget} The basic building block of all gadgets in the settings where agents do not gain opinions simultaneously is the \emph{synchronizing gadget}. This gadget consists of two ports and ensures that the ports obtain their first opinions in the same round, and moreover, these opinions are necessarily different, as otherwise the gadget is never balanced. Moreover, in the positive case, the ports obtain the second opinion exactly three rounds after the first one (unless they obtain it from outside earlier).

Formally, the synchronizing gadget consists of six vertices $v_1$, $v_2$, $v_3$, $v_4$, $v_5$, and $v_6$ such that $f_1(v_i) = 1$ and $f_2(v_i) = 3$ for all $i\in[6]$.
Each vertex from $v_1$, $v_2$, $v_3$ is connected to each vertex from $v_4$, $v_5$, $v_6$.
The \emph{$b$-port} is connected by an edge with all vertices of $\{v_i \mid i \in [6]\}$, while the \emph{$r$-port} is a neighbor only of vertices $v_1$, $v_2$, and $v_3$. 
The gadget is connected to the rest of the graph only using the ports. 
We let $f_2(r)=3$ and $f_2(b)=6$, while the thresholds $f_1(r)$ and $f_1(b)$ will be specified when using the gadget.
Let the preferred opinion be \emph{red}, and the other opinion be \emph{blue}. For the synchronizing gadget, we assume that the $b$-port never acquires red opinion before $r$-port, and the $r$-port never acquires the blue opinion before the $b$-port; this is secured by the rest of the construction. See \Cref{fig:syncGadgetExample} for an illustration of this gadget, together with running examples showing that the gadget indeed possesses the described behavior.
 
 \begin{figure}[tb!]
 	\centering
 	\definecolor{blue}{HTML}{C0DDFF}
 	\definecolor{red}{HTML}{EE4440}
 	\begin{subfigure}[t]{0.32\textwidth}
 		\centering
 		\begin{tikzpicture}
	 		\draw[rounded corners=0.25ex] (0.95,1.6) rectangle (0.8,1.3);
	 		\draw[rounded corners=0.25ex] (1.05,1.6) rectangle (1.2,1.3);
	 		\node[draw,circle] (r) at (1,2) {$r$};

	 		\draw[rounded corners=0.25ex,fill=blue] (4.3,1.6) rectangle (4.45,1.3);
	 		\draw[rounded corners=0.25ex] (4.55,1.6) rectangle (4.7,1.3);
	 		\node[draw,circle] (b) at (4.5,2) {$b$};
	 		
	 		\node[draw,circle,inner sep=0.1em] (v1) at (2,0.7) {\footnotesize$v_1$};
	 		\node[draw,circle,inner sep=0.1em] (v2) at (2,1.5) {\footnotesize$v_2$};
	 		
	 		\draw[rounded corners=0.25ex, fill=white] (1.95,3.4) rectangle (1.8,3.7);
	 		\draw[rounded corners=0.25ex, fill=white] (2.05,3.4) rectangle (2.2,3.7);
	 		\node[draw,circle,inner sep=0.1em] (v3) at (2,3) {\footnotesize$v_3$};

	 		\node[draw,circle,inner sep=0.1em] (v4) at (3.5,0.7) {\footnotesize$v_4$};
	 		\node[draw,circle,inner sep=0.1em] (v5) at (3.5,2.1) {\footnotesize$v_5$};
	 		
	 		\draw[rounded corners=0.25ex, fill=white] (3.45,3.4) rectangle (3.3,3.7);
	 		\draw[rounded corners=0.25ex, fill=white] (3.55,3.4) rectangle (3.7,3.7);
	 		\node[draw,circle,inner sep=0.1em] (v6) at (3.5,3) {\footnotesize$v_6$};
	 		
	 		\begin{scope}[on background layer] 
	 		    \fill[rounded corners=0.25ex, gray!20] (1.7,3.8) rectangle (2.3,0.4);
	 		    \fill[rounded corners=0.25ex, gray!20] (3.2,3.8) rectangle (3.8,0.4);
            \end{scope}
	 		
	 		\draw (r) edge (v1) edge (v2) edge (v3);
	 		\draw (b) edge (v1) edge (v2) edge (v3) edge (v4) edge (v5) edge (v6);
	 		\draw (v1) edge (v4) edge (v5) edge (v6);
	 		\draw (v2) edge (v4) edge (v5) edge (v6);
	 		\draw (v3) edge (v4) edge (v5) edge (v6);
 		\end{tikzpicture}
 		\caption{First, assume the case when the $b$-port gets its first opinion at least one round before the $r$-port obtains its first opinion. By our assumptions, the $b$-port gets necessarily the blue opinion.}
 		\label{fig:runnin_example_initial_setting}
 	\end{subfigure}
 	\hfill
 	\begin{subfigure}[t]{0.32\textwidth}
 		\centering
			\begin{tikzpicture}
			\draw[rounded corners=0.25ex,fill=red] (0.95,1.6) rectangle (0.8,1.3);
			\draw[rounded corners=0.25ex] (1.05,1.6) rectangle (1.2,1.3);
			\node[draw,circle] (r) at (1,2) {$r$};

			\draw[rounded corners=0.25ex,fill=blue] (4.3,1.6) rectangle (4.45,1.3);
			\draw[rounded corners=0.25ex] (4.55,1.6) rectangle (4.7,1.3);
			\node[draw,circle] (b) at (4.5,2) {$b$};
			
			\node[draw,circle,inner sep=0.1em] (v1) at (2,0.7) {\footnotesize$v_1$};
	 		\node[draw,circle,inner sep=0.1em] (v2) at (2,1.5) {\footnotesize$v_2$};
	 		
	 		\draw[rounded corners=0.25ex, fill=blue] (1.95,3.4) rectangle (1.8,3.7);
	 		\draw[rounded corners=0.25ex, fill=white] (2.05,3.4) rectangle (2.2,3.7);
	 		\node[draw,circle,inner sep=0.1em] (v3) at (2,3) {\footnotesize$v_3$};

	 		\node[draw,circle,inner sep=0.1em] (v4) at (3.5,0.7) {\footnotesize$v_4$};
	 		\node[draw,circle,inner sep=0.1em] (v5) at (3.5,2.1) {\footnotesize$v_5$};
	 		
	 		\draw[rounded corners=0.25ex, fill=blue] (3.45,3.4) rectangle (3.3,3.7);
	 		\draw[rounded corners=0.25ex, fill=white] (3.55,3.4) rectangle (3.7,3.7);
	 		\node[draw,circle,inner sep=0.1em] (v6) at (3.5,3) {\footnotesize$v_6$};
	 		
	 		\begin{scope}[on background layer] 
	 		    \fill[rounded corners=0.25ex, gray!20] (1.7,3.8) rectangle (2.3,0.4);
	 		    \fill[rounded corners=0.25ex, gray!20] (3.2,3.8) rectangle (3.8,0.4);
            \end{scope}
	 		
	 		\draw (r) edge (v1) edge (v2) edge (v3);
	 		\draw (b) edge (v1) edge (v2) edge (v3) edge (v4) edge (v5) edge (v6);
	 		\draw (v1) edge (v4) edge (v5) edge (v6);
	 		\draw (v2) edge (v4) edge (v5) edge (v6);
	 		\draw (v3) edge (v4) edge (v5) edge (v6);
		\end{tikzpicture}
 		\caption{In the next round, the agents $v_i$, $i \in [6]$ obtain the blue opinion, as $f_1(v_i) = 1$ and all of them neighbor with the $b$-port. Also the $r$-port might obtain the red opinion.}
 		\label{fig:runnin_example_round2}
 	\end{subfigure}
 	\hfill
 	\begin{subfigure}[t]{0.32\textwidth}
 	\centering
	 	\begin{tikzpicture}
	 		\draw[rounded corners=0.25ex,fill=red] (0.95,1.6) rectangle (0.8,1.3);
	 		\draw[rounded corners=0.25ex,fill=blue] (1.05,1.6) rectangle (1.2,1.3);
	 		\node[draw,circle] (r) at (1,2) {$r$};

	 		\draw[rounded corners=0.25ex,fill=blue] (4.3,1.6) rectangle (4.45,1.3);
	 		\draw[rounded corners=0.25ex,fill=red] (4.55,1.6) rectangle (4.7,1.3);
	 		\node[draw,circle] (b) at (4.5,2) {$b$};
	 		
	 		\node[draw,circle,inner sep=0.1em] (v1) at (2,0.7) {\footnotesize$v_1$};
		 	\node[draw,circle,inner sep=0.1em] (v2) at (2,1.5) {\footnotesize$v_2$};
		 	
		 	\draw[rounded corners=0.25ex, fill=blue] (1.95,3.4) rectangle (1.8,3.7);
		 	\draw[rounded corners=0.25ex, fill=white] (2.05,3.4) rectangle (2.2,3.7);
		 	\node[draw,circle,inner sep=0.1em] (v3) at (2,3) {\footnotesize$v_3$};

		 	\node[draw,circle,inner sep=0.1em] (v4) at (3.5,0.7) {\footnotesize$v_4$};
		 	\node[draw,circle,inner sep=0.1em] (v5) at (3.5,2.1) {\footnotesize$v_5$};
		 	
	 		\draw[rounded corners=0.25ex, fill=blue] (3.45,3.4) rectangle (3.3,3.7);
	 		\draw[rounded corners=0.25ex, fill=white] (3.55,3.4) rectangle (3.7,3.7);
	 		\node[draw,circle,inner sep=0.1em] (v6) at (3.5,3) {\footnotesize$v_6$};
	 		
	 		\begin{scope}[on background layer] 
	 		    \fill[rounded corners=0.25ex, gray!20] (1.7,3.8) rectangle (2.3,0.4);
	 		    \fill[rounded corners=0.25ex, gray!20] (3.2,3.8) rectangle (3.8,0.4);
            \end{scope}
	 		
	 		\draw (r) edge (v1) edge (v2) edge (v3);
	 		\draw (b) edge (v1) edge (v2) edge (v3) edge (v4) edge (v5) edge (v6);
	 		\draw (v1) edge (v4) edge (v5) edge (v6);
	 		\draw (v2) edge (v4) edge (v5) edge (v6);
	 		\draw (v3) edge (v4) edge (v5) edge (v6);		
 	\end{tikzpicture}
 	\caption{Finally, in some future rounds, the $b$-port and $r$-port might eventually get their second opinions. It is easy to see that no $v_i$, $i\in[6]$, can obtain the red opinion since $f_2(v_i) = 3$.}
 	\label{fig:runnin_example_round2}
 	\end{subfigure}

 	\begin{subfigure}[t]{0.32\textwidth}
 		\centering
 		\begin{tikzpicture}
 			\draw[rounded corners=0.25ex,fill=red] (0.95,1.6) rectangle (0.8,1.3);
 			\draw[rounded corners=0.25ex] (1.05,1.6) rectangle (1.2,1.3);
 			\node[draw,circle] (r) at (1,2) {$r$};

 			\draw[rounded corners=0.25ex] (4.3,1.6) rectangle (4.45,1.3);
 			\draw[rounded corners=0.25ex] (4.55,1.6) rectangle (4.7,1.3);
 			\node[draw,circle] (b) at (4.5,2) {$b$};
 			
 			\node[draw,circle,inner sep=0.1em] (v1) at (2,0.7) {\footnotesize$v_1$};
	 		\node[draw,circle,inner sep=0.1em] (v2) at (2,1.5) {\footnotesize$v_2$};
	 		
	 		\draw[rounded corners=0.25ex, fill=white] (1.95,3.4) rectangle (1.8,3.7);
	 		\draw[rounded corners=0.25ex, fill=white] (2.05,3.4) rectangle (2.2,3.7);
	 		\node[draw,circle,inner sep=0.1em] (v3) at (2,3) {\footnotesize$v_3$};

	 		\node[draw,circle,inner sep=0.1em] (v4) at (3.5,0.7) {\footnotesize$v_4$};
	 		\node[draw,circle,inner sep=0.1em] (v5) at (3.5,2.1) {\footnotesize$v_5$};
	 		
	 		\draw[rounded corners=0.25ex, fill=white] (3.45,3.4) rectangle (3.3,3.7);
	 		\draw[rounded corners=0.25ex, fill=white] (3.55,3.4) rectangle (3.7,3.7);
	 		\node[draw,circle,inner sep=0.1em] (v6) at (3.5,3) {\footnotesize$v_6$};
	 		
	 		\begin{scope}[on background layer] 
	 		    \fill[rounded corners=0.25ex, gray!20] (1.7,3.8) rectangle (2.3,0.4);
	 		    \fill[rounded corners=0.25ex, gray!20] (3.2,3.8) rectangle (3.8,0.4);
            \end{scope}
	 		
	 		\draw (r) edge (v1) edge (v2) edge (v3);
	 		\draw (b) edge (v1) edge (v2) edge (v3) edge (v4) edge (v5) edge (v6);
	 		\draw (v1) edge (v4) edge (v5) edge (v6);
	 		\draw (v2) edge (v4) edge (v5) edge (v6);
	 		\draw (v3) edge (v4) edge (v5) edge (v6);
 		\end{tikzpicture}
 		\caption{Next, assume the case when the $r$-port obtains the red opinion at least one round before the $b$-port obtain the blue opinion. Recall that the opposite case is forbidden by our assumption.}
 		\label{fig:runnin_example_initial_setting}
 	\end{subfigure}
 	\hfill
 	\begin{subfigure}[t]{0.32\textwidth}
 		\centering
 		\begin{tikzpicture}
 			\draw[rounded corners=0.25ex,fill=red] (0.95,1.6) rectangle (0.8,1.3);
 			\draw[rounded corners=0.25ex] (1.05,1.6) rectangle (1.2,1.3);
 			\node[draw,circle] (r) at (1,2) {$r$};

 			\draw[rounded corners=0.25ex,fill=blue] (4.3,1.6) rectangle (4.45,1.3);
 			\draw[rounded corners=0.25ex] (4.55,1.6) rectangle (4.7,1.3);
 			\node[draw,circle] (b) at (4.5,2) {$b$};
 			
 			\node[draw,circle,inner sep=0.1em] (v1) at (2,0.7) {\footnotesize$v_1$};
	 		\node[draw,circle,inner sep=0.1em] (v2) at (2,1.5) {\footnotesize$v_2$};
	 		
	 		\draw[rounded corners=0.25ex, fill=red] (1.95,3.4) rectangle (1.8,3.7);
	 		\draw[rounded corners=0.25ex, fill=white] (2.05,3.4) rectangle (2.2,3.7);
	 		\node[draw,circle,inner sep=0.1em] (v3) at (2,3) {\footnotesize$v_3$};

	 		\node[draw,circle,inner sep=0.1em] (v4) at (3.5,0.7) {\footnotesize$v_4$};
	 		\node[draw,circle,inner sep=0.1em] (v5) at (3.5,2.1) {\footnotesize$v_5$};
	 		
	 		\draw[rounded corners=0.25ex, fill=white] (3.45,3.4) rectangle (3.3,3.7);
	 		\draw[rounded corners=0.25ex, fill=white] (3.55,3.4) rectangle (3.7,3.7);
	 		\node[draw,circle,inner sep=0.1em] (v6) at (3.5,3) {\footnotesize$v_6$};
	 		
	 		\begin{scope}[on background layer] 
	 		    \fill[rounded corners=0.25ex, gray!20] (1.7,3.8) rectangle (2.3,0.4);
	 		    \fill[rounded corners=0.25ex, gray!20] (3.2,3.8) rectangle (3.8,0.4);
            \end{scope}
	 		
	 		\draw (r) edge (v1) edge (v2) edge (v3);
	 		\draw (b) edge (v1) edge (v2) edge (v3) edge (v4) edge (v5) edge (v6);
	 		\draw (v1) edge (v4) edge (v5) edge (v6);
	 		\draw (v2) edge (v4) edge (v5) edge (v6);
	 		\draw (v3) edge (v4) edge (v5) edge (v6);
 		\end{tikzpicture}
 		\caption{In the next round, the agents $v_i$, ${i \in [3]}$, obtain the red opinion since $f_1(v_i) = 1$ and all three of them neighbor with the $r$-port. Also, the $b$-port might obtain the blue opinion.}
 		\label{fig:runnin_example_round2}
 	\end{subfigure}
 	\hfill
 	\begin{subfigure}[t]{0.32\textwidth}
 		\centering
 		\begin{tikzpicture}
 			\draw[rounded corners=0.25ex,fill=red] (0.95,1.6) rectangle (0.8,1.3);
 			\draw[rounded corners=0.25ex] (1.05,1.6) rectangle (1.2,1.3);
 			\node[draw,circle] (r) at (1,2) {$r$};

 			\draw[rounded corners=0.25ex,fill=blue] (4.3,1.6) rectangle (4.45,1.3);
 			\draw[rounded corners=0.25ex,] (4.55,1.6) rectangle (4.7,1.3);
 			\node[draw,circle] (b) at (4.5,2) {$b$};
 			
 			\node[draw,circle,inner sep=0.1em] (v1) at (2,0.7) {\footnotesize$v_1$};
	 		\node[draw,circle,inner sep=0.1em] (v2) at (2,1.5) {\footnotesize$v_2$};
	 		
	 		\draw[rounded corners=0.25ex, fill=red] (1.95,3.4) rectangle (1.8,3.7);
	 		\draw[rounded corners=0.25ex, fill=white] (2.05,3.4) rectangle (2.2,3.7);
	 		\node[draw,circle,inner sep=0.1em] (v3) at (2,3) {\footnotesize$v_3$};

	 		\node[draw,circle,inner sep=0.1em] (v4) at (3.5,0.7) {\footnotesize$v_4$};
	 		\node[draw,circle,inner sep=0.1em] (v5) at (3.5,2.1) {\footnotesize$v_5$};
	 		
	 		\draw[rounded corners=0.25ex, fill=red] (3.45,3.4) rectangle (3.3,3.7);
	 		\draw[rounded corners=0.25ex, fill=white] (3.55,3.4) rectangle (3.7,3.7);
	 		\node[draw,circle,inner sep=0.1em] (v6) at (3.5,3) {\footnotesize$v_6$};
	 		
	 		\begin{scope}[on background layer] 
	 		    \fill[rounded corners=0.25ex, gray!20] (1.7,3.8) rectangle (2.3,0.4);
	 		    \fill[rounded corners=0.25ex, gray!20] (3.2,3.8) rectangle (3.8,0.4);
            \end{scope}
	 		
	 		\draw (r) edge (v1) edge (v2) edge (v3);
	 		\draw (b) edge (v1) edge (v2) edge (v3) edge (v4) edge (v5) edge (v6);
	 		\draw (v1) edge (v4) edge (v5) edge (v6);
	 		\draw (v2) edge (v4) edge (v5) edge (v6);
	 		\draw (v3) edge (v4) edge (v5) edge (v6);
 		\end{tikzpicture}
 		\caption{Since the red opinion is preferred by all agents, the remaining agents $v_4$, $v_5$, and $v_6$ obtain it. It is easy to see that no $v_i$, $i\in[6]$, can obtain the blue opinion since $f_2(v_i) = 3$.}
 		\label{fig:runnin_example_round2}
 	\end{subfigure}

 	\begin{subfigure}[t]{0.32\textwidth}
 	\centering
 	\begin{tikzpicture}
 		\draw[rounded corners=0.25ex,fill=red] (0.95,1.6) rectangle (0.8,1.3);
 		\draw[rounded corners=0.25ex] (1.05,1.6) rectangle (1.2,1.3);
 		\node[draw,circle] (r) at (1,2) {$r$};

 		\draw[rounded corners=0.25ex,fill=blue] (4.3,1.6) rectangle (4.45,1.3);
 		\draw[rounded corners=0.25ex] (4.55,1.6) rectangle (4.7,1.3);
 		\node[draw,circle] (b) at (4.5,2) {$b$};
 		
 		\node[draw,circle,inner sep=0.1em] (v1) at (2,0.7) {\footnotesize$v_1$};
	 		\node[draw,circle,inner sep=0.1em] (v2) at (2,1.5) {\footnotesize$v_2$};
	 		
	 		\draw[rounded corners=0.25ex, fill=white] (1.95,3.4) rectangle (1.8,3.7);
	 		\draw[rounded corners=0.25ex, fill=white] (2.05,3.4) rectangle (2.2,3.7);
	 		\node[draw,circle,inner sep=0.1em] (v3) at (2,3) {\footnotesize$v_3$};

	 		\node[draw,circle,inner sep=0.1em] (v4) at (3.5,0.7) {\footnotesize$v_4$};
	 		\node[draw,circle,inner sep=0.1em] (v5) at (3.5,2.1) {\footnotesize$v_5$};
	 		
	 		\draw[rounded corners=0.25ex, fill=white] (3.45,3.4) rectangle (3.3,3.7);
	 		\draw[rounded corners=0.25ex, fill=white] (3.55,3.4) rectangle (3.7,3.7);
	 		\node[draw,circle,inner sep=0.1em] (v6) at (3.5,3) {\footnotesize$v_6$};
	 		
	 		\begin{scope}[on background layer] 
	 		    \fill[rounded corners=0.25ex, gray!20] (1.7,3.8) rectangle (2.3,0.4);
	 		    \fill[rounded corners=0.25ex, gray!20] (3.2,3.8) rectangle (3.8,0.4);
            \end{scope}
	 		
	 		\draw (r) edge (v1) edge (v2) edge (v3);
	 		\draw (b) edge (v1) edge (v2) edge (v3) edge (v4) edge (v5) edge (v6);
	 		\draw (v1) edge (v4) edge (v5) edge (v6);
	 		\draw (v2) edge (v4) edge (v5) edge (v6);
	 		\draw (v3) edge (v4) edge (v5) edge (v6);
 	\end{tikzpicture}
 	\caption{Finally, we examine the positive case where, in the same round, the $r$-port obtains the red opinion and the $b$-port obtains the blue opinion. This is the only case which leads to a balanced outcome.}
 	\label{fig:runnin_example_initial_setting}
 	\end{subfigure}
 	\hfill
 	\begin{subfigure}[t]{0.32\textwidth}
 	\centering
 	\begin{tikzpicture}
 		\draw[rounded corners=0.25ex,fill=red] (0.95,1.6) rectangle (0.8,1.3);
 		\draw[rounded corners=0.25ex] (1.05,1.6) rectangle (1.2,1.3);
 		\node[draw,circle] (r) at (1,2) {$r$};

 		\draw[rounded corners=0.25ex,fill=blue] (4.3,1.6) rectangle (4.45,1.3);
 		\draw[rounded corners=0.25ex] (4.55,1.6) rectangle (4.7,1.3);
 		\node[draw,circle] (b) at (4.5,2) {$b$};
 		
 		\node[draw,circle,inner sep=0.1em] (v1) at (2,0.7) {\footnotesize$v_1$};
	 		\node[draw,circle,inner sep=0.1em] (v2) at (2,1.5) {\footnotesize$v_2$};
	 		
	 		\draw[rounded corners=0.25ex,fill=red] (1.95,3.4) rectangle (1.8,3.7);
	 		\draw[rounded corners=0.25ex, fill=white] (2.05,3.4) rectangle (2.2,3.7);
	 		\node[draw,circle,inner sep=0.1em] (v3) at (2,3) {\footnotesize$v_3$};

	 		\node[draw,circle,inner sep=0.1em] (v4) at (3.5,0.7) {\footnotesize$v_4$};
	 		\node[draw,circle,inner sep=0.1em] (v5) at (3.5,2.1) {\footnotesize$v_5$};
	 		
	 		\draw[rounded corners=0.25ex, fill=blue] (3.45,3.4) rectangle (3.3,3.7);
	 		\draw[rounded corners=0.25ex, fill=white] (3.55,3.4) rectangle (3.7,3.7);
	 		\node[draw,circle,inner sep=0.1em] (v6) at (3.5,3) {\footnotesize$v_6$};
	 		
	 		\begin{scope}[on background layer] 
	 		    \fill[rounded corners=0.25ex, gray!20] (1.7,3.8) rectangle (2.3,0.4);
	 		    \fill[rounded corners=0.25ex, gray!20] (3.2,3.8) rectangle (3.8,0.4);
            \end{scope}
	 		
	 		\draw (r) edge (v1) edge (v2) edge (v3);
	 		\draw (b) edge (v1) edge (v2) edge (v3) edge (v4) edge (v5) edge (v6);
	 		\draw (v1) edge (v4) edge (v5) edge (v6);
	 		\draw (v2) edge (v4) edge (v5) edge (v6);
	 		\draw (v3) edge (v4) edge (v5) edge (v6);
 	\end{tikzpicture}
 	\caption{Even though the agents $v_i$, $i \in [3]$ neighbor with both ports, they obtain the red opinion since it is preferred. By contrast, the agents $v_4$, $v_5$, and $v_6$ obtain the blue opinion, as they neighbor only with $b$.}
 	\label{fig:runnin_example_round2}
 	\end{subfigure}
 	\hfill
 	\begin{subfigure}[t]{0.32\textwidth}
 	\centering
 	\begin{tikzpicture}
 		\draw[rounded corners=0.25ex,fill=red] (0.95,1.6) rectangle (0.8,1.3);
 		\draw[rounded corners=0.25ex] (1.05,1.6) rectangle (1.2,1.3);
 		\node[draw,circle] (r) at (1,2) {$r$};

 		\draw[rounded corners=0.25ex,fill=blue] (4.3,1.6) rectangle (4.45,1.3);
 		\draw[rounded corners=0.25ex,] (4.55,1.6) rectangle (4.7,1.3);
 		\node[draw,circle] (b) at (4.5,2) {$b$};
 		
 		\node[draw,circle,inner sep=0.1em] (v1) at (2,0.7) {\footnotesize$v_1$};
	 		\node[draw,circle,inner sep=0.1em] (v2) at (2,1.5) {\footnotesize$v_2$};
	 		
	 		\draw[rounded corners=0.25ex, fill=red] (1.95,3.4) rectangle (1.8,3.7);
	 		\draw[rounded corners=0.25ex, fill=blue] (2.05,3.4) rectangle (2.2,3.7);
	 		\node[draw,circle,inner sep=0.1em] (v3) at (2,3) {\footnotesize$v_3$};

	 		\node[draw,circle,inner sep=0.1em] (v4) at (3.5,0.7) {\footnotesize$v_4$};
	 		\node[draw,circle,inner sep=0.1em] (v5) at (3.5,2.1) {\footnotesize$v_5$};
	 		
	 		\draw[rounded corners=0.25ex, fill=blue] (3.45,3.4) rectangle (3.3,3.7);
	 		\draw[rounded corners=0.25ex, fill=red] (3.55,3.4) rectangle (3.7,3.7);
	 		\node[draw,circle,inner sep=0.1em] (v6) at (3.5,3) {\footnotesize$v_6$};
	 		
	 		\begin{scope}[on background layer] 
	 		    \fill[rounded corners=0.25ex, gray!20] (1.7,3.8) rectangle (2.3,0.4);
	 		    \fill[rounded corners=0.25ex, gray!20] (3.2,3.8) rectangle (3.8,0.4);
            \end{scope}
	 		
	 		\draw (r) edge (v1) edge (v2) edge (v3);
	 		\draw (b) edge (v1) edge (v2) edge (v3) edge (v4) edge (v5) edge (v6);
	 		\draw (v1) edge (v4) edge (v5) edge (v6);
	 		\draw (v2) edge (v4) edge (v5) edge (v6);
	 		\draw (v3) edge (v4) edge (v5) edge (v6);
 	\end{tikzpicture}
 	\caption{Next, the agents $v_i$, $i\in [6]$ obtain the second opinion, as they neighbor three vertices having it and $f_2(v_i) = 3$. In the next round the ports will obtain their second opinion as $f_2(r)=3$ and $f_2(b)=6$. %
 	}
 	\label{fig:runnin_example_round2}
 	\end{subfigure}
 	
 	\caption{An illustration of the synchronization gadget used in our hardness constructions in the settings with instance-wide preferred red opinion.%
 	}
 	\label{fig:syncGadgetExample}
 \end{figure}
 
\paragraph{Selection Gadget} Now we explain how to modify the selection gadget used in \Cref{sec:hardness_const_tresholds} and \Cref{sec:hardness_const_duration} (cf. \Cref{fig:coloredVertexSelectionGadget}). The selection gadget for the selection of a single element in a set~$W$ again consists of $|W|$ selection vertices which are in one-to-one correspondence with elements of $W$. In contrast to the original gadget, we replace the guard paths with synchronizing gadgets such that all selection vertices are connected to the $r$-ports of these synchronizing gadgets. Finally, we attach to the $b$-port of each synchronizing gadget its unique vertex $v$, and we add $v$ to $S_a$. See \Cref{fig:selectionGadgetPrefOpinion} for an illustration of the selection gadget which also contains details of the thresholds for each vertex of the gadget.
The edge selection gadget is obtained by putting the unique vertices into $S_b$ and reversing the synchronizing gadgets, i.e., the unique vertex will be connected to the $b$-port.
 
\begin{figure}[tb!]
	\centering
		\begin{tikzpicture}[node distance=.5cm]
			\tikzstyle{vertex}=[draw,thick,circle,minimum width=2pt,fill=white]
			\tikzstyle{vertexA}=[vertex,fill=blue]
			\tikzstyle{vertexB}=[vertex,fill=red]
			\tikzstyle{edge}=[thick]
			
			\newcommand{\n}{12}
			\newcommand{\firstGadgetAt}{2}
			\newcommand{\secondGadgetAt}{7}
			\newcommand{\topNodeDistance}{1cm}
			
			\node[vertex,label={[xshift=-3pt,yshift=5pt]180:\footnotesize$f_1=3$},label={[xshift=-3pt,yshift=-5pt]180:\footnotesize$f_2=3$}] (v1) {};
			\foreach \x[remember=\x as \xx (initially 1)] in {2,3,...,\n} {
				\node[vertex,right of=v\xx] (v\x) {};
			}
			\node[right of=v\n,xshift=0.5em] (I) {$I$};
			
			\begin{scope}[on background layer]
				\node[draw,fill=gray!30,dashed,rounded corners,fit=(v1)(v\n)] {};
			\end{scope}
			
			\begin{scope}[node distance=\topNodeDistance]
				\node[vertex,label={[yshift=.4cm]0:\footnotesize$f_1 = 1$},label={0:\footnotesize$f_2 = 6$}] (vS) at ($(v\secondGadgetAt) + (1.2,1.8)$) {};
				
				\node[vertex,draw,rectangle,fill=gray!30, above=of vS,inner sep=0.5cm,node distance=1cm] (sync) {};
				
				\node[vertex,above=of sync,label={[yshift=0.4cm]90:\footnotesize$f_1 = 1$},label={90:\footnotesize$f_2 = 3$}] (r) {};
				\node[vertexB,right of=r,label={\footnotesize$f_2 = 1$}] (vSselected) {};
				
				\draw[edge] (r) -- (sync.north);
				\draw[edge] (vS) -- (sync.south west);
				\draw[edge] (vS) -- (sync.south east);
				\draw[edge] (r) -- (vSselected);
				
				\foreach \x in {1,2,...,\n} {
					\draw[edge] (vS) to (v\x);
					\draw[edge] (v\x) to +(255:.6);
					\draw[edge] (v\x) to +(285:.6);
				}
				
				\node[vertex,label={[yshift=.4cm]180:\footnotesize$f_1 = 1$},label={180:\footnotesize$f_2 = 6$}] (vS) at ($(v\secondGadgetAt) + (-1.2,1.8)$) {};
				
				\node[vertex,draw,rectangle,fill=gray!30, above=of vS,inner sep=0.5cm,node distance=1cm] (sync) {};
				
				\node[vertex,above=of sync,label={[yshift=0.4cm]90:\footnotesize$f_1 = 1$},label={90:\footnotesize$f_2 = 3$}] (r) {};
				\node[vertexB,left of=r,label={\footnotesize$f_2 = 1$}] (vSselected) {};
				
				\draw[edge] (r) -- (sync.north);
				\draw[edge] (vS) -- (sync.south west);
				\draw[edge] (vS) -- (sync.south east);
				\draw[edge] (r) -- (vSselected);
				
				\foreach \x in {1,2,...,\n} {
					\draw[edge] (vS) to (v\x);
					\draw[edge] (v\x) to +(255:.6);
					\draw[edge] (v\x) to +(285:.6);
				}
			\end{scope}
		\end{tikzpicture}
		\caption{Modified (vertex) selection gadget.The inner details of synchronizing gadgets are schematically hidden in the gray rectangle. Threshold functions of all vertices are depicted near to them and the red vertices are initially part of the seed-set~$S_a$.}
		\label{fig:selectionGadgetPrefOpinion}
	\end{figure}

	\begin{figure}
	\centering
	\begin{subfigure}[t]{0.44\textwidth}
	 \begin{tikzpicture}[node distance=.5cm]
  \tikzstyle{vertex}=[draw,thick,circle,minimum width=2pt,fill=white]
  \tikzstyle{edge}=[thick]
  \tikzstyle{sentryStyle}=[vertex,draw=purple]
  \tikzstyle{connectorStyle}=[vertex,draw=yellow!80]
  \tikzstyle{longPath}=[edge,decorate,decoration={snake}]

  \newcommand{\nVertex}{7}
  \newcommand{\nEdge}{12}
  \newcommand{\sentryPosition}{5}

  \begin{scope}[yshift=8cm]
    \node[vertex,label={[xshift=-3pt,yshift=5pt]180:$f_1=3$},label={[xshift=-3pt,yshift=-5pt]180:$f_2=3$}] (v1) {};
    \foreach \x[remember=\x as \xx (initially 1)] in {2,3,...,\nVertex} {
      \node[vertex,right of=v\xx] (v\x) {};
    }
    \begin{scope}[on background layer]
      \node[draw,fill=gray!30,dashed,rounded corners,fit=(v1)(v\nVertex)] {};
    \end{scope}

    \draw[decorate,decoration={brace,amplitude=6pt}] ($(v\nVertex.south) + (.5,0)$) to node [midway,xshift=12pt] {\footnotesize $P_n$}  ($(v\nVertex.south) + (.5,0) - (0,1.8)$);
  \end{scope}

  \begin{scope}
    \node[vertex,label={[xshift=-3pt,yshift=5pt]180:$f_1=3$},label={[xshift=-3pt,yshift=-5pt]180:$f_2=3$}] (e1) {};
    \foreach \x[remember=\x as \xx (initially 1)] in {2,3,...,\nEdge} {
      \node[vertex,right of=e\xx] (e\x) {};
    }
    \begin{scope}[on background layer]
      \node[draw,fill=gray!30,dashed,rounded corners,fit=(e1)(e\nEdge)] {};
    \end{scope}
  \end{scope}

  \node[vertex,label={0:$f_1=1 \,\, f_2=6$}] (connectorVertex) at ($(v\sentryPosition)!.45!(e\sentryPosition)$) {};
  \node[vertex,draw,rectangle,fill=gray!30, inner sep=0.5cm, node distance=1.2cm, below of=connectorVertex] (connectorMid) {};
  \node[vertex,label={0:$f_1=1 \,\, f_2=3$}, node distance=1.2cm,below of=connectorMid] (connectorEdge) {};
  
  \draw[edge] (connectorEdge) -- (connectorMid.south);
  \draw[edge] (connectorVertex) -- (connectorMid.north west);
  \draw[edge] (connectorVertex) -- (connectorMid.north east);  
  
  \node at ($(connectorVertex) + (0,2)$) {$\cdots$};

  \path (connectorVertex) to node[vertex,midway,xshift=-8pt] (v1first) {} (v1);
  \path (connectorVertex) to node[vertex,midway] (v2first) {} node[vertex,pos=.25] {} (v2);
  \path (connectorVertex) to node[vertex,midway] (vNfirst) {} node[vertex,pos=.125] {} node[vertex,pos=.25] {} node[vertex,pos=.375] {} (v\nVertex);

  \path (connectorEdge) to node[vertex,pos=.33] {} node[vertex,pos=.66] (v2EdgeLast1) {} (e2);
  \path (connectorEdge) to node[vertex,pos=.33] {} node[vertex,pos=.66] (v2EdgeLast2) {} (e7);
  \path (connectorEdge) to node[vertex,pos=.33] {} node[vertex,pos=.66] (v2EdgeLast3) {} (e\nEdge);

  \draw[longPath] (v1) to (v1first);
  \draw[longPath] (v2) to (v2first);
  \draw[longPath] (v\nVertex) to (vNfirst);

  \foreach \i/\j in {1/2,2/7,3/\nEdge} {
    \draw[longPath] (v2EdgeLast\i) to (e\j);
    \begin{scope}[on background layer]
      \draw[edge] (v2EdgeLast\i) to (connectorEdge);
    \end{scope}
  }

  \begin{scope}[on background layer]
    \draw[edge] (v1first) to (connectorVertex);
    \draw[edge] (v2first) to (connectorVertex);
    \draw[edge] (vNfirst) to (connectorVertex);
  \end{scope}
\end{tikzpicture}
\caption{Modif.\ incidence gadget for \Cref{sec:hardness_const_tresholds}.}
\label{fig:incidenceGadgetPrefOpinion1}
	\end{subfigure}
\hfill
	\begin{subfigure}[t]{0.44\textwidth}
		\centering
		\begin{tikzpicture}[node distance=.5cm]
			\definecolor{blue}{HTML}{C0DDFF}
			\definecolor{red}{HTML}{EE4440}
			\tikzstyle{vertex}=[draw,thick,circle,minimum width=2pt,fill=white]
			\tikzstyle{vertexA}=[vertex,fill=red]
			\tikzstyle{vertexB}=[vertex,fill=blue]
			\tikzstyle{edge}=[thick]
			\tikzstyle{sentryStyle}=[vertex,draw=purple]
			\tikzstyle{connectorStyle}=[vertex,draw=yellow!80]
			\tikzstyle{longPath}=[edge,decorate,decoration={snake}]
			
			\newcommand{\nVertex}{3}
			\newcommand{\nVertexHalf}{2}
			\newcommand{\nEdge}{12}
			
			\node[vertex,label={[yshift=10pt]90:\footnotesize$f_1=|V_w|$},label={[xshift=-1pt]90:\footnotesize$f_2=6$}, above=2cm] (c1) {};
			
			\foreach \y in {2.25,2,1.75} {
				\draw[edge] (c1) -- (1,\y);
				\draw[edge] (c1) -- (-1,\y);
			}
			
			\node[vertex,rectangle,below=of c1,inner sep=0.5cm,fill=gray!30] (sync1) {};
			\node[vertex,below=of sync1,label={[yshift=5pt]180:\footnotesize$f_1=1$},label={[yshift=-4pt]180:\footnotesize$f_2=3$}] (r1) {};
			
			\draw[edge] (r1) -- (sync1.south);
			\draw[edge] (c1) edge (sync1.north west) edge (sync1.north east);

			\node[vertex,below=of sync1,label={[yshift=5pt]180:\footnotesize$f_1=1$},label={[yshift=-4pt]180:\footnotesize$f_2=3$}] (r2) at ($(r1) - (0,2.25)$) {};

			\node[vertex,rectangle,below=of r2,inner sep=0.5cm,fill=gray!30] (sync2) {};
			
			\node[vertex,label={-90:\footnotesize$f_1=|V_w|$},label={[yshift=-10pt,xshift=-1pt]-90:\footnotesize$f_2=6$},below=of sync2] (c2) {};

			\draw[edge] (r2) -- (sync2.north);
			\draw[edge] (c2) edge (sync2.south west) edge (sync2.south east);

			\foreach \y in {-6,-5.75,-5.5} {
				\draw[edge] (c2) -- (1,\y);
				\draw[edge] (c2) -- (-1,\y);
			}
			
			\begin{scope}
				\node[vertex] at ($(r1) - (0.5,1.25)$) (a1) {};
				\foreach \x[remember=\x as \xx (initially 1)] in {2,3,...,\nVertex} {
					\node[vertex,right of=a\xx] (a\x) {};
				}
				\node[vertexA,fill=red,label={90:\footnotesize$f_2=\deg$}, label={270:$s_{ww'}$}, right of=a\nVertex, xshift=50,yshift=-10]  (A1selected) {};
				
				\foreach \x in {1,...,\nVertex} {
					\draw[edge, bend left] (A1selected) to (a\x);
					\draw[edge] (r1) to (a\x);
					\draw[edge] (r2) to (a\x);
				}
				
				\begin{scope}[on background layer]
					\node[draw,dashed,rounded corners,fit=(a1)(a\nVertex),label={0:$A_{ww'}$},label={[yshift=5pt]180:\footnotesize$f_1=1$},label={[yshift=-4pt]180:\footnotesize$f_2=2$}] (A1fitter) {};
				\end{scope}
			\end{scope}
		\end{tikzpicture}
		\caption{Modif.\ incidence gadget for \Cref{sec:hardness_const_duration}.}
		\label{fig:incidenceGadgetPrefOpinion2}
	\end{subfigure}
	\caption{Illustrations of the modified incidence gadgets. The inner details of synchronizing gadgets are schematically hidden in the gray rectangle. Threshold functions of all vertices are depicted near to them and the red vertices are initially part of the seed-set~$S_a$.
	}
	\label{fig:gadgetsPrefOpinion}
\end{figure}

\paragraph{Incidence Gadget for \Cref{sec:hardness_const_tresholds}}
To modify the incidence gadget used in \Cref{sec:hardness_const_tresholds} (cf. \Cref{fig:incidenceGadget}), we replace the 3 connector vertices and the sentry vertex by a synchronizing gadget. More precisely, edges that originally connected to the vertex connector will now connect to the $b$-port and edges that originally connected to the edge connector will now connect to the $r$-port of the synchronizing gadget. We set $f_1(r)=f_1(b)=1$ (see \Cref{fig:incidenceGadgetPrefOpinion1} for an illustration of the modified gadget).

\paragraph{Incidence Gadget for \Cref{sec:hardness_const_duration}}
To modify the incidence gadget used in \Cref{sec:hardness_const_duration} (cf. \Cref{fig:incidenceGadgetForConstRounds}), we replace each of the checking vertices $c^1_{ww'}$ and $c^2_{ww'}$ with one synchronizing gadget. In both cases the edges towards the selection gadgets are connected to the $b$-port, while the edges towards the set $A_{ww'}$ connect to the $r$-port of the synchronizing gadget. We set $f_1(r)=1$ and $f_1(b)=|V_w|$ for both gadgets (see \Cref{fig:incidenceGadgetPrefOpinion2} for an illustration of the modified gadget). We can also reduce the number of vertices in set $A_{ww'}$ to a constant number.

It is straightforward to verify the correctness of the modified reductions.
As the synchronizing gadget is of constant size, the replacements only affect the studied structural parameters linearly. 
Also the necessary number of rounds is only increased by a constant.
Therefore the hardness results still hold.
Note that the proof of \Cref{thm:twoTSSIsHardForTDandConstRoundsWithFOneLarger} does not use simultaneous gain of opinions and, thus, does not need any modifications.

\paragraph{Modifications of the Algorithms for Vertex Cover, $3$-Path Vertex Cover, and Vertex Integrity}

To modify the algorithm parameterized by the vertex cover to handle the preference on opinions, only two simple modifications are necessary. 
First, \Cref{rr:vertexInSNeverGetSecondOpinion} needs to be applied to all vertices. 
Second, the equivalence relation $\sim$ on $V \setminus U$ needs to be refined by the preferred opinion of the vertex.

Now we move on to the algorithms parameterized by $3$-path vertex cover number and vertex integrity. 
First of all, while guessing $r_a(u)$ and $r_b(u)$ for the vertices $u \in U$, we can only guess that $r_a(u)=r_b(u)$ if $r_a(u)$ is either $0$ or $\infty$.
The linking constraints for $r_a(u) < r_b(u)$ only need a modification if $a$ is the preferred opinion. 
In this case we do not use the constraint \eqref{eq:twotssNFold:globalUpperbound:case3} stating that $u$ does not acquire the opinion $b$ before \emph{or together} with~$a$.
Instead, we only ensure that $u$ does not acquire the opinion $b$ \emph{before}~$a$.
To this end, in case $2 \le r_a(u)$, we use the following constraint:
        \begin{equation}
			\sum_{v \in N(u)} x^{b,r_a(u) - 2}_v \le f_1(u)-1 \label{eq:twotssNFold:globalUpperbound:case6} \,.
		\end{equation}  
		
Similarly, the local constraints can be used without any modification for the preferred opinion, whereas for the less preferred opinion $c$ we replace the condition 
\[
 \left[ x^{a,t-1}_v + x^{b,t-1}_v = 0 \right] 
\]
in \Cref{eq:twotssNFold:boundingXT:1} and \Cref{eq:twotssNFold:boundingXT:4} by 
\[
 \left[ x^{\neg c,t}_v + x^{c,t-1}_v = 0 \right].
\]

It is again straightforward to verify that this new formulation correctly models the preference on the opinions.

\paragraph{Modifications of the Algorithm for Treewidth}

To modify the algorithm parameterized by the treewidth, number of rounds, and the maximum threshold we again, similarly as above, strengthen the conditions on validity of a pattern, requiring that if $r_a(v)=r_b(v)$, then $r_a(u)$ is either $0$ or $\infty$.
If the opinion $c$ is less preferred at vertex $v$, then we also require $h_c(v)=0$ if $r_c(v) = r_{\neg c}(v)+1$ and $\eta_c(v)=0$ even if $r_{\neg c}(v)=2$, but we remove the condition on equality of $\eta_c(v)= h_c(v)$ for the case $r_c(v)=r_{\neg c}(v)+1$.
We do that, because $\eta_c(v)$ will be used in timestep $r_{\neg c}(v)-2$ instead of the original $r_{\neg c}(v)-1$.
The modified activation process will now also respect the opinion preference for vertices in $\alpha(x)$, while it will be still determined by $r_a$ and $r_b$ for vertices of $\beta(x)$.

For the viability of the process, the condition 2(b) is not required if $c$ is less preferred and $r_c(v)=r_{\neg c}(v)+1$.
Also, if $c$ is less preferred, then the third condition of viability will now state that if $2 \le r_{\neg c}(v) < r_c(v) < \infty$, then $\Big| N_{G_x}(v) \cap \widehat{P}^{r_{\neg c}(v) - 2}_c \Big| \le \eta_c(v)$.
At the same time, for the preferred opinion $c$, it will stay untouched, that is, if $1 \le r_{\neg c}(v) < r_c(v) < \infty$, then $\Big| N_{G_x}(v) \cap \widehat{P}^{r_{\neg c}(v) - 1}_c \Big| \le \eta_c(v)$.

Now the semantics of the table is exactly as before, and the solution is again found in $\DP_r[\emptyset,\emptyset,\emptyset,\emptyset,\emptyset,\emptyset,\emptyset,\emptyset]$ for the root~$r$.

For the computation, no modifications are needed for the Leaf node and Introduce Vertex node. Similarly, the computation in Introduce Edge node is performed as before for the preferred opinion $c$. If $c$ is the less preferred opinion, then we have to modify (only) the third point of the computation as follows.
\begin{itemize}
	\item If $r_c(v) + 1 < r_c(u) \le \fin$, then we set $g'_c(u) = \max(0,g_c(u)-1)$.
	Also, if $r_c(u) \neq r_{\neg c} + 1$, then we let $h'_c(u) = h_c(u)-1$, otherwise we set $h'_c=0$ (note that $h'_c=0$ in this case).
	If also $r_c(v) + 2 \le r_{\neg c}(u) < r_c(u)$, then we set $\eta'_c(u)=\eta_c(u)-1$, otherwise we set $\eta'_c(u)=\eta_c(u)$.	
\end{itemize}
For the forget node, we again only explicitly describe the computation for the case ${r_a \le r_b}$. Note that the cases $r_a=r_b=1$ and $r_a=r_b > 1$ are not used. If $a$ is the less preferred opinion, then no changes are needed. If $a$ is the preferred opinion, we define $\bm{q}_{r_a,r_b} = (\overline{r}_a,\overline{r}_b,\overline{g}_a,\overline{g}_b,\overline{h}_a,\overline{h}_b,\overline{\eta}_a,\overline{\eta}_b)$ as 
\[
\mkern-10mu
\begin{array}{ll}
(r_a,r_b,f_1(v),f_2(v),0,0,0,0)                 & \text{if } r_a =1,  r_b=2,\\
(r_a,r_b,f_1(v),f_2(v),0,f_2(v)-1,0,0)          & \text{if } r_a=1,  r_b>2,\\
(r_a,r_b,f_1(v),f_2(v),f_1(v)-1,0,0,f_1(v)-1)   & \text{if } 1 < r_a = r_b-1,\\
\end{array}
\]
and as before for all the other (relevant) cases.
No modifications are needed for the Join node.

The changes do not affect the running time of the algorithm.
The correctness can be proved along the lines of the correctness of the original algorithm.

\subsection{Probabilistic Model}
In practice, while the number of neighbors possessing the opinion can be the main factor influencing whether an agent gains the opinion, the decision will likely depend on a lot of different factors, many of them perhaps unknown even to the agent herself.
The standard approach to account for the factors that the model cannot capture is to treat these factors as a noise, i.e., to introduce randomness into the diffusion process.
Indeed, originally both main models for information diffusion are stochastic.
This amounts to predefined probability of activation along each edge in the Independent cascades model and random threshold for each agent in the Linear threshold model.
Although both these approaches can also be adopted in our setting, we focus only on the latter one here (we slightly discuss the other one in the next subsection).
To make the thresholds stochastic, one samples them from a specified distribution.
As the relation between the two thresholds might be more important than their actual value, it seems natural that a joint distribution is given over pairs of integers rather than two independent distributions for each of the thresholds separately.
Since obtaining such a distribution for a real world network might require many observations, a natural approach would be to cluster the agents into few types and use the same distribution for all the agents of the same type. If the agents of the same type do not have the same degree, we can sample, e.g., a real number between $0$ and $1$, scale it to the degree and round it to an integer.
In fact, the distribution is often estimated by comparing the behavior of the real world network to the modeled one.
Similarly, the seed sets might not be known exactly, but rather as a probability of an agent of a given type becoming a seed for each of the opinions.
Therefore, the input to this model would be a network with defined types of agents; for each type of agent, specified distributions for thresholds and the probability of becoming a seed for each opinion, and values $\fin$ and $B$ as usual.
As there is little hope of achieving an exactly balanced outcome if the network is actually stochastic, we would presumably also specify some allowed imbalance, as discussed in \Cref{subs:imbalance}.
On the one hand, our hardness results are for a special case of this model, namely that the distributions have support of size 1, making the model deterministic.
Indeed, it is not hard to check that, in fact, our constructions even use only constant number of agents types.
Therefore, these hardness results carry over to the stochastic model.
On the other hand, even for the single-opinion stochastic models, computing the expected number of agents activated by a certain seed set is \textsf{\#P}-complete~\cite{ChenWW10}.
Thus, the typical approach to estimating this number is to sample the network several times and average over the results~\cite{TongWD21,LiZCGSL2014,TongWD18,ChenCCKLRSWWY11,BudakAA2011,TongDW18}.
Indeed, a $(1 \pm \varepsilon)$-approximation of the expected number can be obtained with probability at least $1 - \delta$ by sampling $\Omega(n^2 \log(1/\delta)/\varepsilon^2)$ deterministic instances~\cite[Proposition 4.1]{KempeKT15}.
Therefore, it becomes essential to simulate behavior in a deterministic setting, as was argued by, e.g., \cite[Lemma~1]{LuCL2015}.
Our algorithms can be used to compute the optimal solution for each deterministic sample.
However, a more challenging task is to compute a single solution that would achieve a good balance for all samples of the thresholds.
If the seed sets are also stochastic, then it would be natural to specify the solution in terms of the number of agents of each type being part of the solution, also based on whether they are in the seed set or not, rather than specifying particular agents to be in the solution.

This stochastic setting is reminiscent of the two-campaign model of~\citet{GarimellaGPT17}, who also seek seed sets, subject to a budget constraint, that maximize the expected number of \emph{balanced} agents, that is, agents reached by both campaigns or by neither.
The two models differ in three important ways.
First, their diffusion follows the Independent cascades model with per-edge propagation probabilities, whereas in our model, the randomness resides in the agent thresholds.
Second, dependence between the two opinions is captured at the edge level in their work, whereas our joint distribution over pairs of thresholds captures this dependence at the \emph{agent} level, so that different agents can individually treat the two opinions as complementary or competing.
Third, since their objective is neither monotone nor submodular and is \textsf{\#P}-hard to evaluate, they develop constant-factor approximation algorithms, while our results give exact algorithms for the deterministic special case; extending these to a single solution robust across all threshold samples---in the spirit of their approximation framework---is one of several challenges we leave open for follow-up research.

\subsection{Other Possible Modifications}

Another important aspect of our model is that the thresholds as well as the network are the same for both opinions. 
In the spirit of~\citet{LuCL2015}, one could consider specific thresholds for each opinion, i.e., we would have four thresholds for each agent, distinguishing whether it is the first or second opinion to gain and also whether it is opinion $a$ or $b$. 
We could also have a different network for each of the opinions~\cite{BlazejKS2022}, i.e., some edges transmit only opinion $a$ while others transmit only opinion $ b$.
This could, e.g., correspond to a sample from a stochastic model, where the edges transmit the opinion only with some given probability (as in the Independent cascades model).

Obviously, our symmetric model is a special case of both these modifications and, thus, the hardness persists.
Conversely, the algorithms never really make use of the thresholds being symmetric or the network being the same.
Hence, they can be easily modified to handle the asymmetric setting as long as the union of both networks has a favorable structure.
However, we omit the description of the straightforward modifications necessary for these cases.

Finally, the proposed variations of the model can be combined.
Nevertheless, similarly can be combined the suggested modifications of the algorithms and the modifications of the hardness reductions.
This way a rather general model can be analyzed from the complexity view point, based on our results.

\section{Conclusions}\label{sec:conslusions}
We have initiated the study of the \twoTSSshort problem from the parameterized computational complexity perspective.
Similarly to \TSSshort, most combinations of natural parameters do not lead to efficient algorithms, with the parameterized complexity with respect to the number of rounds~$\fin$ and the maximum threshold of $f_{\max}$ (and possibly the budget $B$) being an interesting open problem. We also showed that there are promising algorithms for networks that are sparse (which is often the case in practice).
We believe that identifying and studying important special cases of \twoTSSshort (such as the majority version of the \TSSshort problem) is an interesting research direction.
Yet another similar flavored research direction is to identify further real-world applications of the \twoTSSshort{}.
It is worth mentioning that most of our results (both hardness and algorithmic) can be naturally adopted for multiple opinion models.

Finally, for our theoretical analysis, we have chosen the ``balancing objective''.
This is suitable for an initial study; however, we believe that more desirable objectives are yet to be found together with their applications (and, of course, results).
All in all, we would like to see further research directions based on (parts of) our model.

\section*{Acknowledgements}
\noindent This project has received funding from the European Research Council (ERC) under the European Union’s Horizon 2020 research and innovation programme (grant agreement No 101002854) and was co-funded by the European Union under the project Robotics and advanced industrial production (reg. no. CZ.02.01.01/00/22\_008/0004590). We thank the anonymous reviewers of AAAI 2022 and the Artificial Intelligence Journal for their valuable comments.

\begin{center}
    \vspace{0.25cm}
    \includegraphics[width=5cm]{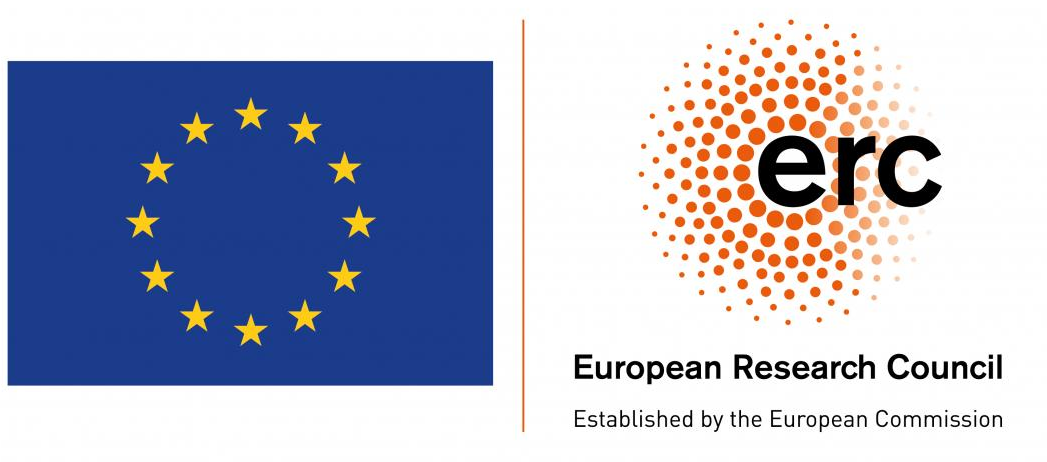}
\end{center}

\bibliographystyle{elsarticle-num-names}
\bibliography{references}

\end{document}